\newif\ifsubmit     %
\newif\ifllncs      %
\newif\ifexabs      %
\newif\ifblind      %

\submittrue

\ifllncs
  \documentclass[runningheads,a4paper]{llncs}

\else
  \documentclass[letterpaper,11pt,pdfa]{article}
  \usepackage[in]{fullpage}
\fi

\usepackage{iftex}
\ifPDFTeX
  \usepackage[utf8]{inputenc}
  \usepackage[noTeX]{mmap}
  \usepackage[T1]{fontenc}
\fi
\ifLuaTeX
  \usepackage{luatex85}
  \usepackage[noTeX]{mmap}
\fi

\usepackage{mdframed}
\usepackage{amsmath}
\usepackage{amsfonts}
\usepackage{amssymb}
\usepackage{amsthm}
\usepackage{color}

\usepackage{appendix}
\usepackage{algorithm}
\usepackage{algpseudocode}
\usepackage{braket}
\usepackage{comment}
\usepackage{url}
\usepackage{bbm}
\usepackage{multicol}
\usepackage{mdframed}
\usepackage{multirow}

\usepackage[bookmarks]{hyperref}
\usepackage[nameinlink]{cleveref}
\hypersetup{
  colorlinks,%
  allcolors=blue
}
\Crefformat{equation}{(#2#1#3)}

\ifllncs

  \spnewtheorem{claim}{Claim}{\bfseries}{\rmfamily}
  \crefname{claim}{claim}{claims}
  \Crefname{claim}{Claim}{Claims}
\else
  \newtheorem{theorem}{Theorem}[section]
  \newtheorem{definition}[theorem]{Definition}
  \newtheorem{remark}[theorem]{Remark}
  \newtheorem{lemma}[theorem]{Lemma}
  \newtheorem{corollary}[theorem]{Corollary}
  \newtheorem*{corollary*}{Corollary}

  \newtheorem{claim}[theorem]{Claim}
  \newtheorem*{remark*}{Remark}
  
  \newtheorem{fact}[theorem]{Fact}

  \newtheorem*{theorem*}{Theorem}
  \newtheorem*{fact*}{Fact}
  \newtheorem*{lemma*}{Lemma}

\usepackage[style=alphabetic,minalphanames=3,maxalphanames=4,maxnames=99,backref=true]{biblatex}
  \renewcommand*{\labelalphaothers}{\textsuperscript{+}}
  \renewcommand*{\multicitedelim}{\addcomma\space}
  \DeclareFieldFormat{eprint:iacr}{Cryptology ePrint Archive: \href{https://ia.cr/#1}{\texttt{#1}}}
  \DeclareFieldFormat{eprint:iacrarchive}{Cryptology ePrint Archive: \href{https://eprint.iacr.org/archive/#1}{\texttt{#1}}}
  \AtEveryBibitem{%
    \clearlist{address}
    \clearfield{date}
    \clearfield{isbn}
    \clearfield{issn}
    \clearlist{location}
    \clearfield{month}
    \clearfield{series}
 
    \ifentrytype{book}{}{%
      \clearlist{publisher}
      \clearname{editor}
    }
  }
\fi

\usepackage{appendix}
\usepackage{algorithm}
\usepackage{algpseudocode}
\usepackage{braket}
\usepackage{comment}
\usepackage{url}
\usepackage{multicol}
\usepackage{tikz}
\usetikzlibrary{arrows,automata,positioning}
\usepackage{caption}
\usepackage[caption=false]{subfig}
\usepackage[font=small,labelfont=bf]{caption}
\usepackage{comment}
\usepackage{pifont}
\usetikzlibrary{patterns}

\allowdisplaybreaks

\usepackage{enumitem}
\setlist[description]{noitemsep}
\setlist[enumerate]{noitemsep}
\setlist[itemize]{noitemsep}

\usepackage{soul, xcolor, xparse}
\makeatletter
  \ExplSyntaxOn
    \cs_new:Npn \white_text:n #1
    {
      \fp_set:Nn \l_tmpa_fp {#1 * .01}
      \llap{\textcolor{white}{\the\SOUL@syllable}\hspace{\fp_to_decimal:N \l_tmpa_fp em}}
      \llap{\textcolor{white}{\the\SOUL@syllable}\hspace{-\fp_to_decimal:N \l_tmpa_fp em}}
    }
    \NewDocumentCommand{\whiten}{ m }
    {
      \int_step_function:nnnN {1}{1}{#1} \white_text:n
    }
  \ExplSyntaxOff
  
  \NewDocumentCommand{ \varul }{ D<>{5} O{0.2ex} O{0.1ex} +m } {%
    \begingroup
    \setul{#2}{#3}%
    \def\SOUL@uleverysyllable{%
      \setbox0=\hbox{\the\SOUL@syllable}%
      \ifdim\dp0>\z@
      \SOUL@ulunderline{\phantom{\the\SOUL@syllable}}%
      \whiten{#1}%
      \llap{%
        \the\SOUL@syllable
        \SOUL@setkern\SOUL@charkern
      }%
      \else
      \SOUL@ulunderline{%
        \the\SOUL@syllable
        \SOUL@setkern\SOUL@charkern
      }%
      \fi}%
    \ul{#4}%
    \endgroup
  }
\makeatother

\newcommand{\Z}{\mathbb{Z}}
\newcommand{\I}{\mathbb{I}}

\usepackage{mleftright}

\newcommand{\As}{\mathcal{A}}
\newcommand{\Bs}{\mathcal{B}}
\newcommand{\Cs}{\mathcal{C}}
\newcommand{\Ds}{\mathcal{D}}

\newcommand{\cA}{\mathcal{A}}

\newcommand{\cB}{\mathcal{B}}

\newcommand{\tQ}{\Tilde{Q}}

\renewcommand{\P}{\sf P}

\newcommand{\prg}{{\sf PRG}}

\newcommand{\ch}{{\sf ch}}

\newcommand{\compressO}{{\sf cO}}
\newcommand{\csto}{{\sf CStO}}
\newcommand{\cphso}{{\sf CPhsO}}
\newcommand{\stddecomp}{{\sf StdDecomp}}

\newcommand{\A}{\mathsf{A}}
\newcommand{\B}{\mathsf{B}}
\newcommand{\C}{\mathsf{C}}

\newcommand{\Q}{\mathsf{Q}}
\newcommand{\R}{\mathsf{R}}

\newcommand{\F}{\mathbb{F}}
\newcommand{\D}{\mathsf{D}}
\newcommand{\M}{\mathsf{M}}
\newcommand{\N}{\mathsf{N}}

\newcommand{\vx}{{\mathbf{x}}}

\newcommand{\vy}{{\mathbf{y}}}
\newcommand{\vz}{{\mathbf{z}}}

\newcommand{\bE}{\mathbf{E}}
\newcommand{\bA}{\mathbf{A}}
\newcommand{\bB}{\mathbf{B}}
\newcommand{\bC}{\mathbf{C}}
\newcommand{\bW}{\mathbf{W}}

\newcommand{\cS}{{\cal S}}

\newcommand{\cE}{{\mathcal{E}}}

\newcommand{\cP}{{\mathcal{P}}}

\newcommand{\eps}{\varepsilon}

\mdfdefinestyle{figstyle}{ 
  linecolor=black!7, 
  backgroundcolor=black!7, 
  innertopmargin=10pt, 
  innerleftmargin=25pt, 
  innerrightmargin=25pt, 
  innerbottommargin=10pt 
}

\newcommand{\msg}{{\sf msg}}

\newcommand{\p}{\mathbf{p}}

\newcommand{\mc}[1]{{\mathcal{#1}}}

\DeclareMathOperator*{\E}{\mathbb E}

\newcommand{\X}{\mathbf{X}}
\newcommand{\Y}{\mathbf{Y}}
\newcommand{\K}{\mathbf{K}}
\newcommand{\KX}{\widetilde{\mathbf{X}}}

\newcommand{\epsQ}{\eps^*}
\newcommand{\epsC}{\eps}

\title{
Tight Characterizations for Preprocessing against Cryptographic Salting
}
\author{
Fangqi Dong\thanks{IIIS, Tsinghua University. Email: \texttt{dongfangqi77@gmail.com}.}
\and
Qipeng Liu\thanks{University of California, San Diego. Email: \texttt{qipengliu0@gmail.com}.}
\and
Kewen Wu\thanks{University of California, Berkeley. Email: \texttt{shlw\_kevin@hotmail.com}. Supported by a Sloan Research Fellowship and NSF CAREER Award CCF-2145474.}
}

\date{}

\begin{document}

\maketitle

\begin{abstract}
Cryptography often considers the strongest yet plausible attacks in the real world. Preprocessing (a.k.a. non-uniform attack) plays an important role in both theory and practice: an efficient online attacker can take advantage of advice prepared by a time-consuming preprocessing stage. 

Salting is a heuristic strategy to counter preprocessing attacks by feeding a small amount of randomness to the cryptographic primitive. 
We present general and tight characterizations of preprocessing against cryptographic salting, with upper bounds matching the advantages of the most intuitive attack.
Our result quantitatively strengthens the previous work by Coretti, Dodis, Guo, and Steinberger (EUROCRYPT'18).
Our proof exploits a novel connection between the non-uniform security of salted games and direct product theorems for memoryless algorithms.

For quantum adversaries, we give similar characterizations for property finding games, resolving an open problem of the quantum non-uniform security of salted collision resistant hash by Chung, Guo, Liu, and Qian (FOCS'20).
Our proof extends the compressed oracle framework of Zhandry (CRYPTO'19) to prove quantum strong direct product theorems for property finding games in the average-case hardness.
\end{abstract}

\clearpage

\section{Introduction}\label{sec:intro}

In modern cryptography, the notion of \emph{security} serves as the cornerstone for evaluating the robustness of cryptographic objects against various adversarial attacks.
In most constructions, the security boils down to the specific hardness of cryptographic primitives.
For instance, the hardness of hash functions refers to the complexity of finding hash collisions; the security of one-way functions hinges on the intractability of finding pre-images of designated outputs.

Traditionally under the lens of uniform security, the adversary's capability is assumed to be fixed and oblivious to the actual primitive.
However, as cryptographic systems evolve to tackle increasingly sophisticated adversaries, the concept of \emph{non-uniform security} emerges as a critical consideration \cite{yao1990coherent,C:Unruh07,EC:CDGS18}, which departs from uniform security by allowing the adversary to adaptively choose its attack strategies based on pre-computed knowledge or advice on the specific cryptographic primitive in question.
This departure reflects a more realistic model of real-world attacks, where the adversary performs extensive interactions with the given primitive with significant amount of time and computational power in an offline phase before actually trying to break the protocol in the online phase.
Given this offline preprocessing stage, adversaries can sometimes perform much better than their uniform counterparts. For example, finding collisions in a vanilla hash function becomes trivial if the adversary reads through its truth table and store one pair of collisions offline.

Cryptographic \emph{salting} is a fundamental and universal technique to address the above issue. By including a random value, known as a \emph{salt}, from the salt space $[K]$ to the primitive, the output becomes unique even if the input data is the same. 
Take collision resistant hash as an example. In the salted version, a random oracle $H\colon[K] \times [M] \to [N]$ is provided. 
The goal of the adversary is, given a uniformly random salt $k$, to find a pair of distinct inputs $x,x' \in [M]$ such that $H(k, x) = H(k, x')$. 
Intuitively, the randomness of the salt $k$ prevents adversaries from effectively using pre-computed tables unless they store look-up tables for each possible salt value, vastly increasing the computational resources required to mount a successful attack.

The ``salting'' approach, dating back to a work by Morris and Thompson~\cite{morris1979password}, guarantees heuristic non-uniform security. The provable consequences of salting  were first investigated by De, Trevisan, and Tulsiani~\cite{de2010time}. 
Later, Chung, Lin, Mahmoody, and Pass~\cite{chung2013power} studied salting for collision resistant hash, one of the most important cryptographic applications. 
Then Mahmoody and Mohammed~\cite{mahmoody2016power} used salting to obtain non-uniform black-box separation results and Dodis, Guo, and Katz~\cite{EC:DodGuoKat17} established tight bounds for some specific salted cryptographic applications. 

More recently, Coretti, Dodis, Guo, and Steinberger~\cite{EC:CDGS18} proved that salting generically defeats preprocessing. Let $G$ be a cryptographic game in the random oracle model (ROM) and $\epsC_G(T)$ be the optimal winning probability that a $T$-query algorithm can obtain. Let $G_K$ be the salted game with the salt space $[K]$ and $\epsC_{G_K}(S, T)$ be the maximum success probability that a non-uniform algorithm with $S$-bit advice and $T$ queries can achieve.
Then they showed that $\epsC_{G_K}(S,T)$ can be bounded by (a multiple of) $\epsC_G(T)$ with an additional non-uniform advantage term decaying as the salt space $[K]$ enlarges.

\begin{theorem}[{\cite[Corollary 18 and 19]{EC:CDGS18}}]\label{thm:CDGS}
For any cryptographic games $G$ in the ROM, we have
\begin{itemize}
\item \textbf{Multiplicative version.} $\epsC_{G_K}(S, T) \le\widetilde O\left(\epsC_G(T) + S T/K\right)$.
\item \textbf{Additive version.} $\epsC_{G_K}(S, T) \le \epsC_G(T) + \widetilde O\left(\sqrt{S T/K}\right)$.
\end{itemize}
Here we use $\widetilde O$ to hide low order terms for simplicity.
\end{theorem}

When $K$ is sufficiently large, the above theorem suggests that non-uniform security of any game is about the same as its uniform security. 
However quantitatively, their lower bounds does not match the most intuitive algorithm. For instance, when $G$ is the game for collision finding, we have $\epsC_G(T) = T^2/N$ by birthday paradox.\footnote{We omit low-order terms in the introduction for conciseness.}
The optimal non-uniform algorithm should be the following: store collisions for $S$ distinct salts; then in the online stage, it either outputs a stored collision for the challenge salt, or executes a birthday-paradox attack. The algorithm achieves advantage $T^2/N+S/K$, but \cite{EC:CDGS18} only provides an upper bound of $T^2/N + ST/K$. 

Furthermore, when a security game with challenges (unlike collision finding) is considered, \Cref{thm:CDGS} is even looser. Dodis, Guo, and Katz \cite{EC:DodGuoKat17} showed that for the function / permutation inversion problem, the optimal bound is $T/N + (ST)/(KN)$, whereas \Cref{thm:CDGS} only offers $T/N + S/K$.

\medskip

Beyond classical adversaries, the growing power of quantum computing poses significant threats to modern cryptography. Chung, Guo, Liu, and Qian~\cite{chung2020tight} and Liu~\cite{liu2023non} studied how salting could prevent quantum non-uniform attacks and provides bounds similar to \Cref{thm:CDGS}. However, even for the simplest yet most important cryptographic application --- collision finding in the quantum random oracle model (QROM) --- they were only able to give an upper bound $T^3/N + ST/K$, whereas the best known attack only achieves $T^3/N + S/K$, where $T^3/N$ comes from the quantum collision finding algorithm \cite{brassard1997quantum} and $S/K$ is the same advantage by storing collisions for $S$ salts.

\medskip

Therefore, in this  article, we ask:
\begin{center}
    {\it Can we provide tighter characterizations for preprocessing against salting, \\ in both the classical and quantum world?}
\end{center}

\subsection{Our Results}\label{sec:result}

\paragraph*{Classical Results.}
Let $G$, $G_K$, $\epsC_G$, and $\epsC_{G_K}$ be defined above. We present our first result.

\begin{theorem}[Consequence of \Cref{thm:classical_salt}]\label{thm:classical_informal}
For any cryptographic game $G$ in any idealized model, we have
\begin{itemize}
\item \textbf{Multiplicative version.} $\epsC_{G_K}(S,T)\le2\cdot\epsC_G(T)+2S/K$.
\item \textbf{Additive version.} $\epsC_{G_K}(S,T)\le\epsC_G(T)+4\sqrt{S/K}$.
\end{itemize}
\end{theorem}

For applications, the multiplicative version of \Cref{thm:classical_informal} is suited for games with small advantages (i.e., when $\epsC_G(T)\ll1$) and the additive version of \Cref{thm:classical_informal} should be used for games with noticeable advantages (e.g., $\epsC_G(T)\ge1/2$ when $G$ is a decision game).
We also remark that while the addend in the additive version of \Cref{thm:classical_informal} incurs a square-root-type loss compared with the multiplicative version, similar term appears and is shown to be tight in related works on non-uniform security of cryptographic primitives (see e.g., \cite{gravin2021concentration}).

Our \Cref{thm:classical_informal} improves \cite{EC:CDGS18} (see \Cref{thm:CDGS}) in all aspects:
\begin{itemize}
\item First, our result shaves the extra $T$ factor in front of each $S/K$ in \Cref{thm:CDGS}.

This improvement leads to a sharp characterization for games without challenges (like collision finding or zero-preimage-finding): with $S$ bits of advice, one can store answers for $\approx S$ salts and use $T$ queries to solve challenges for other salts. 
\item Second, our result works in any idealized model (ROM, the random permutation model, the ideal cipher model, the generic group model, or any distribution of oracles) with any query type (see \Cref{rmk:classical_salt} for detail), while \Cref{thm:CDGS} only works for ROM.
\item Finally, our bound is much cleaner and simpler, providing explicit control of the low order terms.
Our bound is even capable of handling small constant regimes of $S,T,K$ and can be useful in practice.
\end{itemize}

\medskip

When games with challenges are considered (e.g., salted function inversion), \Cref{thm:classical_informal} may no longer be tight. Intuitively, different challenges have completely different answers and a piece of $S$-bit advice cannot guarantee correctness on $S$ salts. We develop a general framework (\Cref{thm:classical_salt_mult}) for understanding salting for these kinds of games. 
To demonstrate its versatility and strength, we provide two non-exhaustive examples: we reprove the bound for salted function inversion as in \cite{EC:DodGuoKat17} and provide a tight characterization of salting when the advice is large. 

\begin{corollary}[\Cref{cor:funcinv} Restated, Salted Function Inversion]\label{cor:funcinv_intro}
    Let ${\sf Inv}$ be the game, where a challenge is a uniformly sampled $y \in [N]$ and the goal is to find $x \in [M]$ in a random function such that $x$ is a pre-image of $y$. Let ${\sf Inv}_K$ be the salted game. Then we have
    \begin{align*}
        \epsC_{{\sf Inv}_{K}}(S, T) \leq O\left(\frac{T}{N} + \frac{ST}{KN}\right).
    \end{align*}
\end{corollary}

\begin{corollary}[Consequence of \Cref{cor:classical_salt_mult}, Large Advice] \label{cor:large_advice}
    When $S = \Omega(K)$, for any cryptographic game $G$ in any idealized model, it holds that
    \begin{align*}
        \epsC_{G_K}(S, T) \leq O\left( \max_{\substack{n_1, \ldots, n_K \in \mathbb{Z}^+ \\ n_1 + \cdots + n_K = S}}  \sum_{i=1}^K \frac{\delta_{G}(n_i, T)}{K} \right),
    \end{align*}
    where $\delta_{G}(n, T)$ is the $n$-th root of the uniform security of the vanilla game $G$ equipped with $n$ independent challenges using $nT$ queries.
\end{corollary}

It turns out that $\delta_G(n,T)$ is an upper (and usually tight) bound on the non-uniform security of the game $G$ with $n$-bit advice and $T$ queries \cite{C:Unruh07,impagliazzo2010constructive,EC:CDGS18,chung2020tight,guo2021unifying}.
In these cases, our \Cref{cor:large_advice} is also tight and offers an intuitive strategy to achieve the optimality for large $S=n_1+\cdots+n_K$: 
reserve $n_k$ bits for each salt $k$ as the best $n_k$-bit advice.
Then the algorithm uses the corresponding advice in the online challenge phase, achieving success probability $(\delta(n_1, T) + \cdots + \delta(n_K, T))/{K}$.

\paragraph*{Quantum Results.}
Next, we turn our attention to the quantum case. \cite{chung2020tight,liu2023non} already studied general bounds for non-uniform quantum attacks against salting. They similarly have an extra factor of $T$ in front of $S/K$, making the bound non-tight even for collision resistant hash.

We start by aiming for a general quantum result as \Cref{thm:classical_informal}, which successfully shaves the additional factor $T$. However, in the proof of \Cref{thm:classical_informal}, we require an average-case direct product theorem (DPT) for general games. 
While a worst-case quantum DPT is known \cite{sherstov2011strong,lee2013strong}, an average-case quantum DPT is needed to analyze the idealized model, which turns out to be a long-standing open problem in quantum query complexity.

Therefore, we resort to the compressed oracle technique by Zhandry~\cite{zhandry2019record}, a powerful framework for proving average-case complexity in the quantum random oracle model (QROM). We show optimal quantum bounds for property finding games, which includes function inversion, collision finding, $k$-SUM, and more.

\begin{theorem}[Consequence of \Cref{thm:main_quantum}]
    Let ${\sf CRHF}$ be the game, where the goal of a quantum-query algorithm is to find $x \ne x' \in [M]$ such that they have the same image under the random oracle. Its quantum non-uniform security $\epsQ_{{\sf CRHF}_K}$ (with quantum advice) is at most
    \begin{align*}
        \epsQ_{{\sf CRHF}_K}(S, T) \leq \widetilde{O}\left( \frac{T^3}{N} + \frac{S}{K} \right).
    \end{align*}
\end{theorem}

At the center of our quantum results is an quantum strong DPT (in fact, we prove a threshold direct product theorem that is even stronger than direct product theorems) in Zhandry’s compressed oracle framework (see \Cref{thm:threshold-qsdpt}).
To the best of our knowledge, this is the first average-case DPT for quantum query algorithms, which may be of independent interests.

\paragraph*{Organization.}
We give an overview of our techniques in \Cref{sec:overview}.
Formal definitions are listed in \Cref{sec:prelim}.
Then \Cref{sec:salting_c} contains the analysis of cryptographic salting in the classical case and \Cref{sec:salting_q} is for the quantum random oracle case.
Missing proofs can be found in Appendix.

\section{Technique Overview}\label{sec:overview}

We give a high-level overview of our proof techniques.

\subsection{The Classical Setting}\label{sec:overview_c}

Let $\As$ be a classical algorithm with $T$ queries and $S$ bits of advice.
Upon receiving a uniformly random salt $k\sim[K]$, $\As$ aims to solve the $k$-th independent copy of the game $G$ in the salted game $G_K$.
Let $\eps$ be the winning probability of $\As$ and our goal is to give an upper bound on $\eps$.

\paragraph*{Removing Advice.}
To eliminate the effect of the non-uniform advice, we consider executing the above procedure multiple times (say, $L$ times).
The key idea here is that, the non-uniform advice is given in advance of the random salt, therefore its power is intuitively diluted across the salt space $[K]$ \cite{impagliazzo2010constructive,impagliazzo2011relativized}.

Let $k_1,\ldots,k_L$ be the $L$ salts chosen. Let $f_1,\ldots,f_K$ be the oracles of each salt.
For any fixed $f_1,\ldots,f_K$, the advice is also fixed.
Thus by the independence of $k_1,\ldots,k_L$, the probability of $\As$ winning on all these $L$ salts conditioned on $f_1,\ldots,f_K$ is simply
$$
\left(\E_k\left[\eps_k\right]\right)^L,
$$
where $\eps_k$ is the winning probability of $\As$ on salt $k$.
Considering the randomness of $f_1,\ldots,f_K$, we also have 
$$
\eps=\E_{f_1,\ldots,f_K}\left[\E_k\left[\eps_k\right]\right].
$$
Hence by convexity, it suffices to bound the probability $\delta$ of $\As$ winning $L$ salts simultaneously and then $\eps\le\delta^{1/L}$.

Intuitively, $\delta$ should be of order $o(1)^L$ since the definition of $\delta$ involves solving $L$ independent trials, where $o(1)$ is the advantage of solving one.
When $L\gg S$, it does not hurt if we replace the non-uniform $S$-bit advice by a random $S$-bit string, since the latter has a probability of $2^{-S}\gg2^{-L}$ to be equal to the former.

Formally, let $\eta$ be the probability of $\As$ winning $L$ salts with a uniformly random $S$-bit string as advice (which can in turn be hardwired oblivious to $f_1,\ldots,f_K$).
Then we have $\delta\le2^S\cdot\eta$ and, more importantly, since the oracle-dependent advice is removed, we only need to analyze uniform algorithms for $\eta$.
This corresponds to \Cref{eq:thm:classical_salt_2} in the proof of \Cref{thm:classical_salt}.

\paragraph*{Reducing to Direct Product Theorems.}
Let $\Bs$ be the $T$-query algorithm corresponding to $\As$ with hardwired advice. Then $\eta$ is equivalent to the success probability of the following procedure:
\begin{enumerate}
\item\label{itm:sec:overview_c_1} Sample salts $k_1,\ldots,k_L\sim[K]$.
\item\label{itm:sec:overview_c_2} Check whether $\Bs$ solves the $k_i$-th independent copy of $G$ for each $i\in[L]$.
\end{enumerate}
Let $U$ be the set of distinct salts in $k_1,\ldots,k_L$.
Then conditioned on \Cref{itm:sec:overview_c_1}, $U$ is fixed and the success of \Cref{itm:sec:overview_c_2} implies on the success of $\Bs$ solving independent copies of $G$ indexed by $U$.
Though $\Bs$ has access to all $f_1,\ldots,f_K$ even when just solving salt $k$, intuitively only $f_k$ should be relevant. This corresponds to \Cref{clm:thm:classical_salt}.

The above desired phenomenon is exactly characterized by \emph{direct product theorems}:
assume our best advantage is $\gamma$ for a single instance (i.e., the game $G$ here) with $T$ budget (i.e., total amount of queries here).
Then solving $u$ independent copies (i.e., $u=|U|$ here) of this instance with $uT$ budget should be $\gamma^u$, which is simply applying the optimal strategy for each copy independently.

Since $\gamma^u$ is attainable, such a direct product theorem is the best possible.
In many cases, such a strong result does not hold and the exact decay bound remains elusive for many important applications (see e.g., \cite{raz2010parallel}).
The study on direct product theorem is rich on its own with numerous works in circuit complexity (e.g., \cite{yao1982theory,levin1985one,impagliazzo2008uniform}), query complexity (e.g., \cite{shaltiel2003towards,drucker2012improved,sherstov2011strong,lee2013strong}), communication complexity (e.g., \cite{jain2021direct,jain2022direct}), property testing and coding theory (e.g., \cite{david2015direct,dinur2014direct}), and more. 

Our setting (the classical query model) unfortunately prohibits a general strong direct product theorem \cite{shaltiel2003towards}, and the best known bound incurs an inevitable polynomial loss \cite{drucker2012improved} dependent on $\gamma$.
Nevertheless, we are able to exploit the memoryless structure of the above procedure to obtain a strong direct product theorem through a similar analysis as presented in \cite{nisan1998products}. We explain it after we complete the analysis of the non-uniform security $\eps$ of $G_K$ against $T$-query algorithms with $S$-bit advice: 
\begin{itemize}
\item By considering $L$ independent salts, we have $\eps\le\delta^{1/L}$.
\item By guessing and removing the advice, we have $\delta\le2^S\cdot\eta$.
\item By strong direct product theorems, we have 
$$
\eta\le\E_U\left[\gamma^{|U|}\right]\le\left(\gamma+\frac LK\right)^L,
$$ 
where $\gamma$ is the uniform security of $G$ against $T$-query algorithms and $|U|$ is the number of distinct values in a draw-with-replacement experiment with $L$ draws from the universe $[K]$.
The second inequality follows from standard martingale analysis (see \Cref{lem:draw_with_repl}).
\end{itemize}
Rearranging and optimizing the choice of $L$ gives the desired bound (see \Cref{thm:classical_salt} for detail).

\paragraph*{Strong Direct Product Theorems for Memoryless Algorithms.}
Before presenting our argument, we remark that our analysis coincides with a much earlier work by Nisan, Rudich, and Saks \cite{nisan1998products} under a different background. While we only became aware of their work after the completion of our paper, we do \emph{not} consider our analysis for this part as our technical contribution, and we keep our argument within our setting just for completeness.

For simplicity, consider the following scenario and the general treatment can be found in \Cref{sec:memoryless_sdpt}: let $x^1,\ldots,x^n$ be $n$ i.i.d. random strings given oracle access to, and we execute a $T$-query algorithm $\Cs$ sequentially for $1,2,\ldots,n$ to predict $h(x^1),\ldots,h(x^n)$.
To compare with the analysis above, $n=|U|$ is the number of distinct salts and $x^1,\ldots,x^n$ are the corresponding oracles, $\Cs$ is essentially $\Bs$, $h$ is the predicate for the game (e.g., $h(x)$ is the pre-image of zero in $x$), and $\gamma$ is the optimal success probability of a $T$-query algorithm in predicting an \emph{individual} $h(x^i)$.
Our goal is to prove that the above joint prediction succeeds with probability at most $\gamma^n$.

The literature of direct product theorems usually assumes the joint prediction is provided by a much more powerful query algorithm, one single algorithm of $nT$ queries that makes all the predictions in the end \cite{drucker2012improved,sherstov2011strong,lee2013strong}.
While intuitively it makes no sense for the algorithm to query outside $x^i$ to predict $h(x^i)$, it can actually distribute the total $nT$ queries in a more clever and adaptive way that outperforms the naive strategy to predict $h(x^i)$ individually with $T$ queries on $x^i$.
As a consequence, a strong bound of $\gamma^n$ turns out to be false \cite{shaltiel2003towards} and the best known bound has an inevitable growing loss dependent on $\gamma$ itself \cite{drucker2012improved}, making it insufficient for our purposes.

The key belief that drives us towards a bound of $\gamma^n$ in our setting comes from the fact that $\Cs$ is executed separately to predict each $h(x^i)$ and it does not share memory between separate prediction tasks.
In other words, the above scenario has a memoryless structure that the workspace is cleaned repeatedly after each $T$ queries, which intuitively forbids any contrived attacks requiring coordination between predictions of $\Cs$.
Note that while this intuition checks out, in reality $\Cs$ can still query, say $x^j$, outside $x^i$ to predict $h(x^i)$, which makes the predictions $h(x^i),h(x^j)$ correlated and perhaps magically they can be correct simultaneously with higher probability.
In fact, examples in this spirit date back to Feige's famous counterexample for the sharp parallel repetition theorem \cite{feige1991success}.

To rigorously verify our belief, we start with the work by Shaltiel \cite{shaltiel2003towards}, which proves the tight $\gamma^n$ bound for \emph{fair} algorithms.
A fair algorithm is allowed to make all predictions in the end, potentially dependent on all the $nT$ queries, however, it is guaranteed that each $x^i$ can only be queried at most $T$ times.
Since the queries can still be intertwined among $x^i$'s, the predictions are, again, correlated. 
Nevertheless, since there is an upper bound on the queries for each $x^i$, a simple induction argument can be carried out.
In a bit more detail, one can do induction on the $n$-tuple $(T_1,\ldots,T_n)$ where each $T_i$ is the upper bound of the total number of queries on $x^i$. Then one can show that the optimal strategy for a fair algorithm is spending $T_i$ queries on $x^i$ to predict $h(x^i)$ independently.

Given the sharp bound for fair algorithms (\Cref{thm:fair_algo}), our goal is to reduce our memoryless algorithms to fair algorithms (\Cref{lem:reduction_memoryless_to_fair}).
Note that while the memory is cleaned after each $T$ queries, the whole execution can make $\gg T$ queries on $x^i$.
On the other hand, since a fair algorithm does not have to be memoryless, these two types of algorithms are incomparable to each other.
To get a fair algorithm from a memoryless algorithm, we further exploit the memoryless structure: if some $h(x^i)$ is already predicted, then we can privately simulate the future queries to $x^i$ without actually querying $x^i$, as $\Cs$ is separately executed for each prediction and these executions commute.

By this observation, we perform the following out-of-order simulation:
let $\Cs_1,\ldots,\Cs_n$ be $\Cs$ executed to predict $h(x^1),\ldots,h(x^n)$ respectively.
We initialize $\Cs_1,\ldots,\Cs_n$ and assume without loss of generality they are deterministic.
While there is some prediction not ready, we use a cycle / path elimination procedure to progress:
for each $i$, if $\Cs_i$ has not terminated and wants to query $x^j$, then we add a directed edge from $i$ to $j$.
In this directed graph, 
\begin{itemize}
\item Either we have a directed cycle $i_1\to i_2\to\cdots\to i_\ell\to i_1$.

Then we advance $\Cs_{i_1},\cdots,\Cs_{i_\ell}$ by one query to $x^{i_2},\ldots,x^{i_{\ell}},x^{i_1}$ respectively.
\item Or we have a directed path $i_1\to i_2\to\cdots\to i_\ell$.

Then we advance $\Cs_{i_1},\ldots,\Cs_{i_{\ell-2}}$ by one query to $x^{i_2},\ldots,x^{i_{\ell-1}}$ respectively.
In addition, we privately simulate the query of $\Cs_{i_{\ell-1}}$ on $x^{i_\ell}$ without actually querying $x^{i_\ell}$.
This is valid since $i_\ell$ is the endpoint of the path, thus $\Cs_{i_\ell}$ is already terminated and the prediction $h(x^{i_\ell})$ is already generated.
\end{itemize}
By the above procedure, each $x^i$ is queried at most $T$ times, since when $\Cs_i$ queries, either $x^i$ gets a query from the predecessor in the cycle / path, or the query is privately simulated.
Therefore this out-of-order execution is a fair algorithm and we can apply the known strong direct product theorem for fair algorithms.

\paragraph*{A Different Analysis.}
Finally we sketch a different analysis (see detail in \Cref{sec:classical_salt_refine}) that improves the previous analysis for games with numerous challenges (e.g., \Cref{cor:funcinv_intro}) or large advice (\Cref{cor:large_advice}).
Recall that in the reduction to direct product theorems, we focus on $U$, the set of distinct salts in $L$ independent uniform salts $k_1,\ldots,k_L$.
This is tight if the challenge distribution of $G$ has low entropy.
For example, if the challenge is null (e.g., collision finding requires no challenge description), then winning salts $k_1,\ldots,k_L$ is equivalent to winning salts in $U$.
However if the challenge distribution has high entropy (e.g., the function inversion asks to invert a random challenge point), then winning $U$ can be quite off from winning all the salts, since with high probability we get different challenges even if the salts are the same.

To avoid such a potential loss, we instead look directly into $k_1,\ldots,k_L$. For each salt $k\in[K]$, let $n_k$ be the number of its appearances in $k_1,\ldots,k_L$.
Then $n_1,\ldots,n_K$ are determined by $k_1,\ldots,k_L$.
Now for fixed $k_1,\ldots,k_L$, the algorithm $\Bs$ needs to solve $n_k$ independent challenges of salt $k$ for each $k\in[K]$.
Let $G^{n_k}$ be the multi-challenge version of the game $G$ where $n_k$ i.i.d. challenges are asked simultaneously.
Then the algorithm aims to win all of them in a similar memoryless fashion: it can use $n_kT$ queries to solve $G^{n_k}$ separately for each $k\in[K]$.
Though the games $G^{n_1},\ldots,G^{n_K}$ are not identical, strong direct product theorems for memoryless algorithms can be proved with an almost identical argument.
Therefore, we can relate the non-uniform security of salting to the advantages of the multi-challenge version of $G$ (see \Cref{cor:large_advice}).

We remark that the multi-challenge version of $G$ is closely related to the non-uniform security of $G$ itself (without salting), and (tight) upper bounds have been established for various games (see e.g., \cite{C:Unruh07,EC:CDGS18,chung2020tight,guo2021unifying}).

\subsection{The Quantum Random Oracle Setting}\label{sec:overview_q}

We now move to the quantum case. The aforementioned reduction in the classical setting works in the quantum setting as long as the non-uniform quantum algorithm is given \emph{classical advice}. In the described reduction, we construct a non-uniform algorithm solving multiple instances of a game, by iteratively executing an algorithm with fixed advice many times.
Such a reduction falls short when the advice is quantum, as outputting any answer would destroy the advice in an irreversible way. 

\paragraph*{Quantum Bit-Fixing.}
To address the quantum non-cloning issue, we resort to the quantum bit-fixing framework, which was first introduced by Guo, Li, Liu, and Zhang~\cite{guo2021unifying} and later refined by Liu~\cite{liu2023non}. 
In their work, they defined the so-called quantum bit-fixing model (QBF) and proved a novel connection between security in the QBF model and non-uniform security in the quantum random oracle model (QROM). Before introducing the QBF model,  let us first grasp the intuition behind it (see \Cref{subsec:prf-main_quantum} for more details).

Liu~\cite{liu2023non} considered a quantum-friendly adaptation of direct product games, where the quantum advice can be effectively reused to win ``multiple instances'' of the original games (formally, the alternating measurement game). This is achieved by a similar idea of the witness preserving amplification of QMA~\cite{marriott2005quantum}, to make sure that the quantum advice is preserved throughout the entire game. 
The $S$-fold alternating measurement game has $S$ rounds, with each round allowing the algorithm to make $T$ queries. An algorithm wins if and only if it succeeds all rounds in a sequential order, much like the memoryless algorithms mentioned earlier.
Since it is not the primary focus of this work, we refer interested readers to \cite{liu2023non} for more details of alternating measurement games. 
Eventually, if any uniform quantum-query algorithm (conducting $T$ queries per round) can win the $S$-fold alternating measurement game with probability at most $\eta^S$, the quantum non-uniform security $\epsQ_G(S, T)$ of the original game is at most $\eta$.

Now the goal is to bound the success probability of the $S$-fold game. For each $i\in[S]$, let $\mathbf{W}_i$ be the event that the algorithm wins the $i$-th round. 
We aim to show that for some $\eta$, the following holds: 
\begin{align*}
    \Pr\left[ \mathbf{W}_1 \wedge \mathbf{W}_2 \wedge \cdots \wedge \mathbf{W}_S \right] < \eta^S. 
\end{align*}
\noindent Define $\mathbf{W}_{<i}$ be the event that the algorithm wins the first $(i-1)$ rounds. Then we have
\begin{align*}
    \Pr\left[ \mathbf{W}_1 \wedge \mathbf{W}_2 \wedge \cdots \wedge \mathbf{W}_S \right] = \prod_{i=1}^S \Pr[ \mathbf{W}_i \,|\,  \mathbf{W}_{<i}]. 
\end{align*}
Since \cite{liu2023non} also showed that $\Pr[ \mathbf{W}_i \,|\,  \mathbf{W}_{<i}]$ is monotonically non-decreasing in $i$, if we can prove $\Pr[\mathbf{W}_S|\mathbf{W}_{<S}] < \eta$, it will imply  $\Pr[ \mathbf{W}_1 \wedge \cdots \wedge \mathbf{W}_S] < \eta^S$. 
Even further, we can assume $\Pr[\mathbf{W}_{<S}]$ is typical, say, $\Pr[\mathbf{W}_{<S}]\ge(1/N)^{S}$.\footnote{This lower bound $(1/N)^{S}$ is arbitrarily chosen in this section for clarity in presentation.} Since otherwise $\Pr[ \mathbf{W}_1 \wedge \cdots \wedge \mathbf{W}_S]\le\Pr[\mathbf{W}_{<S}] < (1/N)^{S} \leq \eta^S$ is automatically satisfied as the target $\eta$ is usually $\gg1/N$.

\textbf{The security in the QBF model} is defined as the maximum of $\Pr[\mathbf{W}_S|\mathbf{W}]$ over any $T$-query strategy for the $S$-th game and any event $\mathbf{W}$ that can be achieved by an $(ST)$-quantum-query algorithm with a ``non-negligible'' probability (e.g., $(1/N)^{S}$ as above).
In the salted case, $\mathbf{W}_S$ itself precisely represents the event that a uniform $T$-quantum-query algorithm wins the game $G$ on a random salt. 

Two immediate questions arise. First, even without conditioning on $\mathbf{W}$, bounding the probability of $\mathbf{W}_S$ necessitates tools from average-case quantum query complexity. Moreover, it becomes even more challenging when conditioning is considered.

In this study, we focus solely on the game for which Zhandry's compressed oracle technique \cite{zhandry2019record} serves as a potent tool for addressing average-case query complexity. We will first provide an exposition of this technique and then introduce the concepts required to handle conditioning.

\paragraph*{An Exposition of the Compressed Oracle Technique.}

In the rest of the overview, we will focus on the collision finding problem. Let us begin with the classical case. A $T$-classical-query algorithm can acquire information on at most $T$ input-output pairs of a random function. By the principle of deferred decision, one can assume that outputs are only sampled for those queried inputs, and every other output remains uniformly random. This is commonly referred to as ``lazy sampling''. Thus, we can say a random oracle is initialized as $D = \emptyset$ (an empty database), meaning every output is uniformly random at the beginning. For every query $x \in [M]$, if there is a pair $(x, y) \in D$, it outputs $y$; otherwise it updates $D \gets D \cup \{(x, y)\}$ for a uniformly random $y$ and returns $y$. The database contains all the information a classical algorithm learns so far. Most importantly, the probability of the algorithm finds a collision is roughly the probability that $D$ contains a collision. 
Hence, the objective shifts to bounding the probability that, following $T$ queries, $D$ contains a collision. 

The seminal work by Zhandry~\cite{zhandry2019record} invented a quantum analogue of the classical lazy sampling, called ``compressed oracle''. Since quantum queries can be made in superposition, one single classical database is no longer feasible to track all the information. 
However, the notion of a ``database'' is still meaningful if the ``database'' itself is also stored in superposition. A quantum-query algorithm interacting with a quantum-accessible random oracle can be equivalent simulated as follows:
\begin{itemize}
    \item The database register is initialized as $\ket \emptyset$: the algorithm has not yet queried anything.
    \item When the algorithm makes a quantum query, the database register gets updated in superposition: for a query $\ket x$ and a database $\ket D$, the simulator looks up $x$ in $D$.
    If $x$ is not in $D$, it initializes an equal superposition $\sum \ket y$ (up to normalization), and updates $\ket D$ to $\sum_y \ket {D\cup (x, y)}$. 
    It then returns $y$ also in superposition. 
\end{itemize}
The above description provides a high-level idea, albeit slightly inaccurate. Since a quantum algorithm can forget  previously acquired information, a formal ``compressed oracle'' necessitates an additional step to perform these forgetting operations. 

The compressed oracle exhibits several favorable properties akin to those enjoyed by classical lazy sampling technique: 
\begin{enumerate}
    \item \textbf{Local change.} Every quantum query will only change $\ket D$ locally: the addition of a new entry $(x, y)$, no change, or the removal of an existing entry (corresponding to the ability of forgetting). 
    \item \textbf{Knowledge upper bound.} The probability of a quantum algorithm finding a collision is approximately the same as the probability of measuring the database register and obtaining $D$ with collision.
\end{enumerate}
Hence, for the remainder of the overview, our focus will solely be on detecting collisions within the database register, rather than the output produced by the quantum algorithm.

\paragraph*{Probability as Path Integral.}

Before we approach our goal: the security of salted collision finding in the QBF model (i.e., to upper bound $\Pr[\mathbf{W}_S|\mathbf W]$), we first reprove the bound for the vanilla collision finding~\cite{zhandry2019record} against uniform quantum query algorithms. We reinterpret the analysis as a form of ``path integral''. This alternative view gives a new way to understand~\cite{zhandry2019record} and will help us establish a strong (threshold) direct product theorem to upper bound $\Pr[\mathbf{W}_S|\mathbf W]$.

For any $T$-query quantum algorithm, its interaction with a compressed  oracle can be written as $\compressO \, U_T \cdots \compressO \, U_2 \, \compressO \,  U_1 \ket 0 \otimes \ket \emptyset$, where $\ket 0$ is the algorithm's initial state, $\ket \emptyset$ is the initial database, each $U_i$ is a local unitary only on the algorithm's register, and $\compressO$ is the compressed oracle query. We are interested in the probability of observing a collision in $D$, i.e., the squared norm of the following (sub-normalized) state:
\begin{align*}
    \ket \phi := \Lambda \, \compressO \, U_T \, \cdots \compressO \, U_2 \, \compressO \,  U_1 \ket 0 \otimes \ket \emptyset,
\end{align*}
where $\Lambda$ is a projection onto databases with at least one collision. 

Feynman's interpretation of quantum mechanics (path integral), on a high level, postulates that the probability of an event (observing a collision) is given by the squared modulus of a complex number (as $\|\ket \phi\|$ above), and this amplitude is given by adding all paths contributing to the complex number. 
Since our ultimate goal is to observe a collision within the database, the collision must occur as a result of one of the oracle queries $\compressO$. Therefore, we contemplate all potential pathways leading to a collision. For any $t \in \{1, 2, \ldots, T\}$, we define the following (sub-normalized) state:
\begin{align*}
    \ket {\psi_t} := \underbrace{(\Lambda \, \compressO) \, U_T}_\textrm{$T$-th query} \, \underbrace{(\Lambda \, \compressO) \, U_{T-1}}_\textrm{$(T-1)$-th query}  \cdots \underbrace{(\Lambda \, \compressO) \, U_t}_\textrm{$t$-th query} \, \underbrace{(\mathbb{I}-\Lambda) \compressO \, U_{t-1}}_\textrm{$(t-1)$-th query}  \cdots \compressO \, U_2 \, \compressO \,  U_1 \ket 0 \otimes \ket \emptyset.
\end{align*}
The definition of $\ket{\psi_t}$ guarantees that, starting from the $t$-th query, the database register will always be supported on those databases containing collision.
Observe that $\ket \phi=\ket {\psi_1} + \ket {\psi_2} + \cdots + \ket {\psi_T}$.
Thus $\|\ket \phi\| \leq \|\ket {\psi_1}\| + \|\ket {\psi_2}\| + \cdots + \|\ket {\psi_T}\|$. 
Then by showing $\|\ket {\psi_t}\| \lesssim \sqrt{t/N}$, one deduces $\|\ket \phi\|^2 \lesssim T^3/N$. 

We can also treat the classical lazy sampling as a path integral, where we similarly define $\ket \phi$ (resp., $\ket {\psi_t}$) as the computation at the end (resp., after $t$ queries) projected onto databases with collisions. The only difference is that, all probabilities are summed directly in the classical case: $\|\ket \phi\|^2 = \|\ket {\psi_1}\|^2 + \|\ket {\psi_2}\|^2 + \cdots + \|\ket {\psi_T}\|^2 \lesssim T^2/N$. This is because in this path integral for the classical algorithm, the path decomposition $\ket{\psi_1},\ldots,\ket{\psi_T}$ are mutually orthogonal, for which we can use Pythagorean equation to replace the triangle inequality.

\paragraph*{Establishing the Security in the QBF Model.}

Recall that $G_K$ is the salted collision finding problem.
The security of $G_K$ in the QBF model requires us to upper bound $\Pr[\As \text{ wins } G_K \,|\, \mathbf{W}]$, where $\As$ is a $T$-quantum-query algorithm and the event $\mathbf{W}$ can be achieved by an $ST$-quantum-query algorithm with a ``non-negligible'' probability. 
We aim to show that the probability above is upper bounded by $O(T^3/N + S/K)$.

What does this conditioning enforce in the context of the compressed oracle? It implies that the database register does not start with $\ket \emptyset$, but rather as a superposition of databases. Moreover, since $\mathbf{W}$ can be achieved by an $ST$-quantum-query algorithm, according to the ``local change'' property, the database register will have non-zero support only over databases with at most $ST$ entries.  However, this condition alone is not sufficient to derive the desired bound.  Consider the following scenario: though every $\ket D$ with non-zero support has at most $ST$ entries, it exactly stores collisions for $ST/2$ distinct salts. In such a case, with a random challenge salt, the database with probability $ST/K$ consists of an answer for the salt even if $\As$ makes no further query. 
Given this example, we further separate the probability:
\begin{align*}
    \Pr[\As \text{ wins } G_K \,|\, \mathbf{W}] \leq \underbrace{\Pr[\As \text{ wins } G_K \,|\, \mathbf{W} \land \mathbf{E}]}_\textrm{(i)} + \underbrace{\Pr[\overline{\mathbf{E}}|\mathbf{W}]}_\textrm{(ii)},
\end{align*}
where $\mathbf{E}$ denotes the event that the database register contains $\ket D$ that stores collisions for at most $\widetilde{S} = 2 S \log N$ distinct salts. This idea of posing $\mathbf{E}$ was first proposed by Akshima, Guo, and Liu~\cite{akshima2022time} in the context of classical Merkle-Damg\r{a}rd hash. We adapt it to the quantum setting, combining with the compressed oracle technique. 

The term (i) follows closely to the ideas in the lower bound of the vanilla collision finding, with the distinction that the database register commences with databases satisfying the requirements outlined by both $\mathbf{W}$ and $\mathbf{E}$: each database has at most $ST$ entries and collisions for at most $\widetilde{S}$ salts. We show that (i) is $\lesssim T^3/N + \widetilde{S}/K$. 

To bound the term (ii), we show that $\Pr[\overline{\mathbf{E}}] < (T^3/N)^{\widetilde{S}} \lesssim(1/N)^{2 S}$. Since we assumed $\Pr[\mathbf W]\ge(1/N)^S$, we now have $\Pr[\overline{\mathbf{E}}|\mathbf{W}] < \Pr[\overline{\mathbf{E}}] / \Pr[{\mathbf{W}}] < 1/N$.  Upon accomplishing this, we establish security in the QBF model, along with the non-uniform security for salted collision finding.  Please refer to \Cref{subsec:prf-main_quantum} for more detail. 

\paragraph*{Strong (Threshold) Direct Product Theorem in the Compressed Oracle.}
Now we give more detail for the term (ii).
This final step involves bounding $\Pr[\overline{\mathbf{E}}]$. The probability corresponds to the following direct-product-type experiment:
there is a random oracle $f:[K] \times [M] \to [N]$ (simulated as a compressed oracle). 
An algorithm makes at most $ST$ quantum queries and the objective is to make the database register contain collisions for at least $\widetilde{S}$ salts. We let $\Lambda^{\geq r}$ be a projection onto databases with collisions on \emph{at least} $r$ salts and $\Lambda^{=r}$ projects onto database with collisions on \emph{exact} $r$ salts. Then $\Pr[\overline{\mathbf{E}}]$ is the squared norm of the following (un-normalized) state: 
\begin{align*}
    \ket \phi :=\ & \Lambda^{\geq \widetilde{S}} \, \compressO \, U_{ST} \, \cdots \compressO \, U_2 \, \compressO \,  U_1 \ket 0 \otimes \ket \emptyset\\
    =\ & \sum_{r\geq \widetilde{S}}\Lambda^{=r} \, \compressO \, U_{ST} \, \cdots \compressO \, U_2 \, \compressO \,  U_1 \ket 0 \otimes \ket \emptyset.
\end{align*}

We list all possible computation paths that lead to databases in $\Lambda^{=r}$ for some $r\ge \widetilde{S}$. 
For every $\ket D \in \Lambda^{=r}$, following its evolution, there exists $\vz = (z_1, z_2, \ldots, z_{r})$ such that before the $z_i$-th query, there are $(i-1)$ salts having a collision that the database will not forget in the future;
right after the $z_i$-th query, there are $i$ salts having a collision that the database will not forget.
We denote the computation path by the state $\ket {\psi_\vz}$. Similar to the vanilla collision finding, for the case of databases in $\Lambda^{=r}$, it holds that
\begin{align*}
    \Lambda^{=r}\ket \phi = \sum_\vz \ket {\psi_\vz}.
\end{align*}

Finally, we show that for every possible path $\vz$, the state $\ket {\psi_\vz}$ has a norm $\lesssim(T/N)^{r/2}$. This step necessitates a much more intricate analysis compared to bounding $\|\ket{\psi_{i}}\|$ in the vanilla case. When some salt obtains a collision in the database that will never be forgotten in the future, we must keep track of which salt $i$ it is, as well as how many queries / entries $q_i$ already exist for this particular salt. If we further split $\ket {\psi_\vz}$ into paths with all possible $(k_1, \ldots, k_r)$, $(q_1, \ldots, q_r)$ and still use the triangle inequality, we will not be able to obtain the desired bound. One of our novel and key ideas to finalize the proof in the direct product game is to show that all these paths have some form of orthogonality. We can leverage the orthogonality of all these cases\footnote{We only show that these paths are ``approximately'' orthogonal instead of perfectly orthogonal, and this suffices to conclude the result.}, and show that we can use Pythagorean equation instead of triangle inequality to complete the proof, much like the path integral for classical algorithms. 

Since the total number of $\vz$'s is $\binom{ST}{r}$, $\|\Lambda^{=r}\ket \phi\|^2 \lesssim \binom{ST}{r}^2 (T/N)^{r} \lesssim (T^3/N)^{r}$. By summing over every $r\ge \widetilde{S}$, we can conclude that $\|\ket \phi\|^2=\sum_{r\ge\widetilde{S}}\|\Lambda^{=r}\ket{\phi}\|^2\lesssim (T^3/N)^{r}$, from the summation of geometric series for the case of $T^3/N<1/2$.
More details on our strong (threshold) direct product theorem can be found in \Cref{sec:property_finding_sdpt}.

\section{Preliminaries}\label{sec:prelim}

For each positive integer $n$, we use $[n]$ to denote set $\{1,2,\ldots,n\}$.
We use $\mathbb{Z}$ to denote the set of integers, use $\mathbb{N}$ to denote the set of non-negative integers, and use $\mathbb{Z}_+$ to denote the set of positive integers. 

For a distribution $\mu$, we use $x\sim\mu$ to denote a sample $x$ from $\mu$.
For a finite set $S$, we use $x\sim S$ to denote a uniformly random element $x$ from $S$.

For a complex number $x\in\mathbb{C}$, we denote its norm $|x|=(x\overline{x})^{1/2}$, where $\overline{x}$ is the conjugate of the complex number $x$. For a complex vector $\vx\in\mathbb{C}^n$ or a quantum state $\ket{\phi}=\sum_{i\in[n]}\alpha_i\ket{i}$, we denote the $\ell_2$-norm $\|\vx\|=|\vx|_2=(\sum_{i\in[n]}x_i\overline{x_i})^{1/2}$ and $\|\ket{\phi}\|=(\sum_{i\in[n]}\alpha_i\overline{\alpha_i})^{1/2}$.

\paragraph*{Asymptotics.}
We use the standard $O(\cdot), \Omega(\cdot), \Theta(\cdot)$ notation, and emphasize that in this paper they only hide universal positive constants that do not depend on any parameter.

\paragraph*{Games.}
We start by defining the game in the general oracle model, then instantiate it in various settings.

\begin{definition}[Game]\label{def:game}
A game $G$ is described by an oracle distribution $\mu$ over functions $f\colon[M]\to[N]$,\footnote{Throughout the paper, we assume that oracles are functions with domain $[M]$ and range $[N]$ where the actual $M$ and $N$ may change under different contexts. This is without loss of generality due to the flexibility of $\mu$.} a challenge distribution $\pi_f$ specified by each possible oracle $f$ from $\mu$, and an accepting outcome set (aka predicate) $\mc{R}_{f,\ch}$ specified by each possible oracle $f$ from $\mu$ and each possible challenge from $\pi_f$.

We say a query algorithm $\As$ wins the game given oracle access to $f$ and challenge $\ch$ if the output of $\As$ is in $\mc{R}_{f,\ch}$, i.e., $A^f(\ch) \in \mc{R}_{f,\ch}$. 
The winning probability of $\As$ is then computed by 
$$
\Pr\left[\As^f(\ch)\in \mc{R}_{f,\ch}\right],
$$
where the randomness is over $f\sim\mu$, $\ch\sim\pi_f$, and potentially the internal randomness of $\As$.
\end{definition}
Note that in the above definition, $R_f$ is allowed to be empty, in which case the winning probability is simply zero.

During the analysis, we will need the oracle distribution conditioned on challenge. This naturally leads to the notion of plain games: games without any challenge.

\begin{definition}[Plain Game]\label{def:plain_game}
A plain game $G$ is described by an oracle distribution $\mu$ and an accepting outcome set $\mc{R}_f$.
The winning probability of a query algorithm is defined similarly as in \Cref{def:game}.
\end{definition}

Since games with or without challenge can be distinguished by whether there is a challenge distribution, we will address both of them as \emph{games} for convenience when context is clear.

\begin{definition}[Challenge-Conditioned Game]\label{def:game_cond}
Given a game $G=(\mu,\{\pi_f\},\{\mathcal{R}_{f,\ch}\})$ and a possible challenge $\ch$, we use $G_\ch=(\mu_\ch,\{\mathcal{R}_{f,\ch}\})$ to denote the plain game of $G$ conditioned on $\ch$, where $\mu_\ch$ is the oracle distribution conditioned on the challenge being $\ch$.
\end{definition}
We also have the following easy fact.

\begin{fact}\label{fct:oracle_cond}
$\mu=\E_\ch[\mu_\ch]$, where $\ch$ is sampled from $\pi_f$ for $f\sim\mu$.
\end{fact}

\paragraph*{Optimal Winning Probabilities.}
We start by defining the optimal winning probabilities for uniform and non-uniform query algorithms.
Let $G$ be a game and $S,T\in\mathbb{N}$.
\begin{itemize}
\item \textbf{Classical Uniform Game.}
In this case, $\As$ is a classical query algorithm with at most $T$ queries given challenge $\ch$ (if not plain) and classical oracle access to $f$. A classical oracle query to $f\colon[M]\to[N]$ is specified by a coordinate $x\in[M]$ and gets in return the value $f(x)\in[N]$.

We use $\epsC_G(T)$ to denote the maximal winning probability of such algorithms.
\item \textbf{Classical Non-Uniform Game.}
In this case, $\As$ is still limited to at most $T$ classical queries to $f$.
With classical non-uniform power, $\As$ now additionally receives an $S$-bit advice string $\sigma_f$ that can depend arbitrarily on $f$ (but not on the challenge $\ch$).
We use $\As^f(\sigma_f,\ch)$ (or $\As^f(\sigma_f)$ if plain) to explicitly denote that $\As$ takes the non-uniform advice as a separate input.

We use $\epsC_G(S,T)$ to denote the maximal winning probability of such algorithms.

\item \textbf{Quantum Uniform  Game.} Here, $\As$ is a $T$-quantum-query algorithm. In this paper, we are only interested in the case when $f$ is a uniformly random oracle. More details are provided in \Cref{sec:QROM}.

We use $\epsQ_G(T)$ to denote the maximal winning probability of such quantum algorithms.

\item \textbf{Quantum Non-Uniform Game.} Here, $\As$ is a $T$-quantum-query algorithm, with an additional piece of advice $\sigma_f$. In the quantum case, we consider the most general case: the advice $\sigma_f$ has $S$ qubits and can arbitrarily depend on $f$. 

We use $\epsQ_G(S,T)$ to denote the maximal winning probability of such non-uniform quantum algorithms.
\end{itemize}

Since the challenges are explicitly provided to the algorithm, the optimal winning probabilities for uniform algorithms are equivalent for a game and the corresponding challenge-conditioned games.

\begin{fact}\label{fct:opt_game_cond}
Let $G=(\mu,\{\pi_f\},\{\mathcal{R}_{f,\ch}\})$ be a game.
Then $\epsC_G(T)=\E_\ch[\epsC_{G_\ch}(T)]$.
\end{fact}

\paragraph*{Salted Games.}
Salting is a generic way to restrict non-uniform power. The method appends a piece of random data (aka salt) to inputs fed into a hash function (or more generally, an oracle); and the game will be now played under the random data. 

\begin{definition}[Salted Game]\label{def:salted_game}
    Let $G$ be a game (see \Cref{def:game}) specified by $(\mu,\{\pi_f\}_f,\{\mc{R}_{f,\ch}\}_{f,\ch})$.
    
    For each $K \in \mathbb{Z}_+$, we define the salted game $G_K$ as follows:
    \begin{itemize}
        \item The oracle distribution $\mu_K$ is over functions $g\colon[KM] \to [N]$. We identify $[KM]$ as $[K]\times[M]$ and $g=(f_1,\ldots,f_K)$ where each $f_i\colon[M]\to[N]$. The distribution $\mu_K$ is generated by independently sampling $g(i,\cdot)=f_i\sim\pi$ for each $i\in[K]$.
        \item For each possible oracle $g=(f_1,\ldots,f_K)$ from $\mu_K$, the challenge $\ch=(k,\ch_k)\sim\pi_g$ is produced by first picking a uniformly random $k\sim[K]$, then sampling $\ch_k\sim\pi_{f_k}$.
        \item For each possible oracle $g=(f_1,\ldots,f_K)$ from $\mu_K$ and challenge $\ch=(k,\ch_k)$ from $\pi_g$, the accepting outcome set $\mc{R}_{g,\ch}$ is defined as $\mc{R}_{f_k,\ch_k}$.
    \end{itemize}
\end{definition}

\begin{remark}\label{rmk:alter_salt}
Intuitively, the game $G_K$ samples $K$ independent copies $f_1,\ldots,f_K$ from $\mu$ as a joint oracle, then samples a uniformly random one $f_k$ for the challenge, followed by a concrete challenge $\ch_k$ from $\pi_k$ and the corresponding accepting outcome set $\mc{R}_{f_k,\ch_k}$.
In order to win this game, the algorithm $\As$ needs to output a valid answer for $\mc{R}_{f_k,\ch_k}$.
\end{remark}

Note that the definition of optimal winning probabilities carries over. In particular, our goal for the classical setting is to relate $\epsC_G(T)$ and $\epsC_{G_K}(S,T)$. We similarly want to relate $\epsQ_G(T)$ and $\epsQ_{G_K}(S,T)$ in the quantum random oracle model.

\paragraph*{Useful Inequalities.}
\Cref{lem:draw_with_repl} is a moment bound for the number of distinct elements in a draw-with-replacement experiment. Its proof can be found in \Cref{app:missing_prelim}.

\begin{lemma}\label{lem:draw_with_repl}
Let $K,L\in\mathbb{Z}_+$. 
Sample $k_1,\ldots,k_L\sim[K]$ independently and let $\ell$ be the number of distinct elements in $k_1,\ldots,k_L$.
Then for any $c\ge0$, we have
$$
\E[c^\ell]\le\left(c+\frac LK\right)^L.
$$
\end{lemma}

\Cref{fct:binom_choice} is a simple counting argument and the estimate follows directly from Stirling's formula, i.e. $a!>(a/e)^a$ for $a\in\mathbb{Z}_{+}$.
\begin{fact}\label{fct:binom_choice}
For $K,L\in\mathbb{Z}_+$, we have
$$
\left|\left\{n_1,\ldots,n_K\in\mathbb N\colon n_1+\cdots+n_K=L\right\}\right|
=\binom{K+L-1}{K-1}
\le\begin{cases}
(2e\cdot K/L)^L & L\le K,\\
(2e\cdot L/K)^K & L\ge K.
\end{cases}
$$
\end{fact}

\section{Non-Uniform Security of Salting: the Classical Setting}\label{sec:salting_c}

In this section, we prove the non-uniform security of salting in the classical case: we upper bound the advantage of an adversarial classical algorithm against salted oracles with the help of advice strings.

\begin{theorem}\label{thm:classical_salt}
Let $G$ be a game and let $G_K$ be the salted game for $K\in\mathbb{Z}_+$.
Then for any $S,T\in\mathbb N$, we have
$$
\epsC_{G_K}(S,T)\le\min_{L\in\mathbb{Z}_+}2^{S/L}\cdot\left(\epsC_G(T)+\frac LK\right).
$$
\end{theorem}

\begin{remark}\label{rmk:classical_salt}
Since we do not assume any structure of $G$ in \Cref{thm:classical_salt} and its bound does not depend on specific range parameters in $G$, this result is general enough to handle \emph{any} kind of classical query model.

Let $f\colon[N]\to[M]$ be the oracle for $G$.
Let $\mathcal{Q}$ be the set of allowed queries, e.g., bit query (each time queries some $i$-th bit of some $f(x)$), word query (each time queries some $f(x)$, which is the query model we use here), majority query (each time queries the majority of $\{f(x)\colon x\in S\}$ for some $S\subseteq[N]$), and so on.
Let $\mathcal S$ be the range of the possible answers of queries in $\mathcal Q$.

Then we can view each oracle $f$ as a function $f_{\mathcal{Q}}\colon\mathcal{Q}\to\mathcal S$.
In particular, $f_{\mathcal Q}(q)$ returns the query answer on $f$ given a query $q\in\mathcal Q$.
Then the allowed query $\mathcal Q$ on $f$ becomes the word query on $f_{\mathcal Q}$.
Therefore, we can simply treat $f_{\mathcal Q}$ as the oracle in the game $G$, change the accepting outcome set accordingly, and apply \Cref{thm:classical_salt}.\footnote{One subtlety here is that $f$ and $f_{\mathcal{Q}}$ may not be a one-to-one correspondence. Though $f$ defines $f_{\mathcal Q}$, it is possible that the answers to the queries does not uniquely identify the oracle. In this case, the distribution of $f_{\mathcal Q}$ is properly induced the distribution of $f$ and the predicate (aka accepting outcome set) will be probabilistic given $f_{\mathcal Q}$ ans algorithm's output. The argument still works out with minor changes and we leave this as an exercise for interested readers.}
\end{remark}

\begin{remark}\label{rmk:classical_salt_stronger}
Though \Cref{thm:classical_salt} is generally tight (see \Cref{thm:classical_informal}), it can be further strengthened in some setting.
Intuitively, if the game has numerous possible challenges, one cannot store sufficient amount of answers to cover all the challenges, especially given the presence of salts.
In other words, answers in a look-up table need to not only specify the salt, but also specify the exact challenge corresponding to the answer.

Given the above intuition, in \Cref{sec:classical_salt_refine} we modify the proof of \Cref{thm:classical_salt} to get a different estimate, which improves \Cref{thm:classical_salt} when there are many possible challenges. We also explicitly compute some examples to demonstrate the tightness.
\end{remark}

The strategy to prove \Cref{thm:classical_salt} is as follows:
let $g$ be the salted oracle sampled from $g\sim\mu_K$.
\begin{itemize}
\item We first relate the success probability of winning a single challenge $\ch\sim\pi_g$ to the success probability of winning independent challenges $\ch_1,\ldots,\ch_L\sim\pi_g$.
\item Then we show that the $S$-bit advice string does not boost the winning probability much when we have multiple independent challenges in the sense that a random string of $S$-bit can be taken as the advice.
\item Now that the advice string is substituted by a random string independent of $g$, we only need to bound the winning probability of a uniform query algorithm against multiple independent challenges. This will be handled by a direct product theorem.
\end{itemize}
One can view the direct product theorem as a type of concentration results, where one analyze the high order advantages of algorithms.
This reduction from non-uniform security to uniform security with concentration is analogous to earlier works of Impagliazzo and Kabanets \cite{impagliazzo2010constructive,impagliazzo2011relativized}, and has been revived in recent works \cite{gravin2021concentration,akshima2020time,ghoshal2022time,akshima2022time,liu2023non}.

We emphasize that the most general direct product theorems \cite{drucker2012improved} incur a necessary loss in parameters and cannot produce sharp bounds that \Cref{thm:classical_salt} needs.
Instead, we exploit the \emph{memoryless} property of our algorithms and show a strong direct product theorem for this class of algorithms.

\begin{proof}[Proof of \Cref{thm:classical_salt}]
Let $G$ be specified by $(\mu,\{\pi_f\}_f,\{\mathcal{R}_{f,\ch}\}_{f,\ch})$.
Recall \Cref{def:salted_game} and \Cref{rmk:alter_salt}. The game $G_K$ samples $K$ independent copies of $G$ and focuses on a uniformly random one of them.
As such, let $f_1,\ldots,f_K\sim\mu$ and $k\sim[K]$ be independent, then $g=(f_1,\ldots,f_K)$ is the joint oracle sampled from $\mu_K$ for $G_K$.
Let $\ch_k\sim\pi_{f_k}$, then $\ch=(k,\ch_k)$ is the challenge sampled from $\pi_g$ for $G_K$.

Fix an arbitrary $L\in\mathbb{Z}_+$.
Let $\As$ be an arbitrary non-uniform query algorithm for $G_K$ that takes an $S$-bit advice string $\sigma_g$ and makes at most $T$ queries to $g$.
Then the winning probability of $\As$ is $\Pr\left[\As^g(\sigma_g,\ch)\in\mathcal{R}_{g,\ch}\right]$, where $\mathcal{R}_{g,\ch}=\mathcal{R}_{f_k,\ch_k}$.
Since we only make classical queries, oracle access to $g$ is equivalent to oracle access to $f_1,\ldots,f_K$ separately.
Therefore it suffices to prove
\begin{equation}\label{eq:thm:classical_salt_1}
\Pr_{\substack{f_1,\ldots,f_K\\k,\ch_k}}\left[\As^{f_1,\ldots,f_K}(\sigma_g,k,\ch_k)\in\mathcal{R}_{f_k,\ch_k}\right]\le2^{S/L}\cdot\left(\epsC_G(T)+\frac LK\right),
\end{equation}
where we assume that $\As$ is a deterministic query algorithm given $\sigma_g,k,\ch_k$ by fixing the optimal randomness.

To remove the non-uniform advice from the LHS of \Cref{eq:thm:classical_salt_1}, we use a tensor trick.
Let challenges $(k_1,\ch_{k_1}),\ldots,(k_L,\ch_{k_L})\sim\pi_g$ be independent copies of $(k,\ch_k)$.
Then for any fixed $g$ (i.e., $f_1,\ldots,f_K$), the advice $\sigma_g$ is also fixed and we have
$$
\Pr_{(k_i,\ch_{k_i}),i\in[L]}\left[\As^{f_1,\ldots,f_K}(\sigma_g,k_i,\ch_{k_i})\in\mathcal{R}_{f_{k_i},\ch_{k_i}},\forall i\in[L]\right]
=\Pr_{k,\ch_k}\left[\As^{f_1,\ldots,f_K}(\sigma_g,k,\ch_k)\in\mathcal{R}_{f_k,\ch_k}\right]^L.
$$
Hence
\begin{align*}
\text{LHS of \Cref{eq:thm:classical_salt_1}}
&=\E_{f_1,\ldots,f_K}\left[\Pr_{k,\ch_k}\left[\As^{f_1,\ldots,f_K}(\sigma_g,k,\ch_k)\in\mathcal{R}_{f_k,\ch_k}\right]\right]
\notag\\
&\le\E_{f_1,\ldots,f_K}\left[\Pr_{k,\ch_k}\left[\As^{f_1,\ldots,f_K}(\sigma_g,k,\ch_k)\in\mathcal{R}_{f_k,\ch_k}\right]^L\right]^{1/L}
\tag{by convexity}\\
&=\Pr_{\substack{f_1,\ldots,f_K\\(k_i,\ch_{k_i}),i\in[L]}}\left[\As^{f_1,\ldots,f_K}(\sigma_g,k_i,\ch_{k_i})\in\mathcal{R}_{f_{k_i},\ch_{k_i}},\forall i\in[L]\right]^{1/L}.
\end{align*}
Now observe that if $\sigma_g$ is replaced by a uniformly random string $\sigma\sim\{0,1\}^S$, then $\sigma$ still has a probability of $2^{-S}$ to be identical with $\sigma_g$, from which we recover the true non-uniform algorithm.
Putting this into the above calculation, we have
\begin{align}
\text{LHS of \Cref{eq:thm:classical_salt_1}}
&\le\left(2^S\cdot\Pr_{\substack{f_1,\ldots,f_K\\(k_i,\ch_{k_i}),i\in[L]\\\sigma\sim\{0,1\}^S}}\left[\As^{f_1,\ldots,f_K}(\sigma,k_i,\ch_{k_i})\in\mathcal{R}_{f_{k_i},\ch_{k_i}},\forall i\in[L]\right]\right)^{1/L}
\notag\\
&\le2^{S/L}\cdot\Pr_{\substack{f_1,\ldots,f_K\\(k_i,\ch_{k_i}),i\in[L]}}\left[\Bs^{f_1,\ldots,f_K}(k_i,\ch_{k_i})\in\mathcal{R}_{f_{k_i},\ch_{k_i}},\forall i\in[L]\right]^{1/L},
\label{eq:thm:classical_salt_2}
\end{align}
where $\Bs$ is a uniform $T$-query algorithm obtained by fixing the optimal choice of $\sigma$ for $\As$.

For each possible $(k_1,\ldots,k_L)$, we fix an arbitrary\footnote{For example, if $(k_1,k_2)=(1,2)$, then $U=\{1,2\}$; if $(k_1,k_2)=(1,1)$, then $U$ can be either $\{1\}$ or $\{2\}$.} $U\subseteq[L]$ such that $U$ contains indices of distinct salts in $k_1,\ldots,k_L$.
Note that $U$ is deterministic given $(k_1,\ldots,k_L)$.
Now we further analyze the RHS of \Cref{eq:thm:classical_salt_2}:
\begin{align}
&\Pr_{\substack{f_1,\ldots,f_K\\(k_i,\ch_{k_i}),i\in[L]}}\left[\Bs^{f_1,\ldots,f_K}(k_i,\ch_{k_i})\in\mathcal{R}_{f_{k_i},\ch_{k_i}},\forall i\in[L]\right]
\notag\\
\le&\Pr_{\substack{f_1,\ldots,f_K\\(k_i,\ch_{k_i}),i\in[L]}}\left[\Bs^{f_1,\ldots,f_K}(k_i,\ch_{k_i})\in\mathcal{R}_{f_{k_i},\ch_{k_i}},\forall i\in U\right]
\notag\\
=&\E_{\substack{U\\(k_i,\ch_{k_i}),i\in U}}\left[\Pr_{f_1,\ldots,f_K}\left[\Bs^{f_1,\ldots,f_K}(k_i,\ch_{k_i})\in\mathcal{R}_{f_{k_i},\ch_{k_i}},\forall i\in U\middle|U,\{(k_i,\ch_{k_i})\}_{i\in U}\right]\right].
\label{eq:thm:classical_salt_3}
\end{align}
Observe that $(U,\{k_i\}_{i\in U})$ is independent from $(f_1,\ldots,f_K)$ and each $\ch_{k_i}$ only distorts the distribution of $f_{k_i}$ from $\mu$ to $\mu_{\ch_{k_i}}$ (recall \Cref{def:game_cond}).
Hence, conditioned on $(U,\{(k_i,\ch_{k_i})\}_i)$, the oracles $f_1,\ldots,f_K$ are still independent; and in particular, $f_{k_i}$ has distribution $\mu_{\ch_{k_i}}$ for $i\in U$ and other oracles have marginal distribution $\mu$.
Therefore, the RHS of \Cref{eq:thm:classical_salt_3} can be simplified as
\begin{align}\label{eq:thm:classical_salt_4}
\E_{\substack{U\\(k_i,\ch_{k_i}),i\in U\\f_j\sim\mu,j\notin\{k_i\colon i\in U\}}}\left[\Pr_{f_{k_i}\sim\mu_{\ch_{k_i}},i\in U}\left[\Bs^{f_1,\ldots,f_K}(k_i,\ch_{k_i})\in\mathcal{R}_{f_{k_i},\ch_{k_i}},\forall i\in U\right]\right].
\end{align}

We will prove the following claim by a strong direct product theorem in \Cref{sec:memoryless_sdpt}.

\begin{claim}\label{clm:thm:classical_salt}
For any fixed $(U,\{(k_i,\ch_{k_i})\}_i,\{f_j\}_j)$, we have 
$$
\Pr_{f_{k_i}\sim\mu_{\ch_{k_i}},i\in U}\left[\Bs^{f_1,\ldots,f_K}(k_i,\ch_{k_i})\in\mathcal{R}_{f_{k_i},\ch_{k_i}},\forall i\in U\right]
\le\prod_{i\in U}\epsC_{G_{\ch_{k_i}}}(T).
$$
\end{claim}

Assuming \Cref{clm:thm:classical_salt}, we can upper bound \Cref{eq:thm:classical_salt_4}:
\begin{align*}
\E_{\substack{U\\(k_i,\ch_{k_i}),i\in U\\f_j\sim\mu,j\notin\{k_i\colon i\in U\}}}\left[\prod_{i\in U}\epsC_{G_{\ch_{k_i}}}(T)\right]
&=\E_{\substack{U\\k_i,i\in U}}\left[\E_{\ch_{k_i},i\in U}\left[\prod_{i\in U}\epsC_{G_{\ch_{k_i}}}(T)\right]\right]\\
&=\E_{\substack{U\\k_i,i\in U}}\left[\prod_{i\in U}\E_{\ch_{k_i}}\left[\epsC_{G_{\ch_{k_i}}}(T)\right]\right]
\tag{by independence}\\
&=\E_{\substack{U\\k_i,i\in U}}\left[\prod_{i\in U}\epsC_G(T)\right]=\E_U\left[\epsC_G(T)^{|U|}\right]
\tag{by \Cref{fct:opt_game_cond}}\\
&\le\left(\epsC_G(T)+\frac LK\right)^L.
\tag{by \Cref{lem:draw_with_repl}}
\end{align*}
Combining with \Cref{eq:thm:classical_salt_2} and \Cref{eq:thm:classical_salt_3}, this establishes \Cref{eq:thm:classical_salt_1} and thus concludes the proof of \Cref{thm:classical_salt}.
\end{proof}

\subsection{Strong Direct Product Theorems for Memoryless Algorithms}\label{sec:memoryless_sdpt}

In this section, we will present the strong direct product theorem we use in \Cref{clm:thm:classical_salt}. This theorem was proved within the context of classical complexity theory in a much earlier work by Nisan, Rudich, and Saks \cite{nisan1998products}. 
While we only became aware of their work after the completion of our paper, we do \emph{not} consider our analysis in \Cref{sec:memoryless_sdpt} as our technical contribution, and we prove the theorem here within our model for completeness.

We start by defining the direct product of games.

\begin{definition}[Direct Product of Games]\label{def:dp}
Let $G_1,\ldots, G_k$ be $k$ plain games specified by $\mu_1,\ldots, \mu_k$ and collections of predicates $\mc{R}_{1,f_1}, \ldots, \mc{R}_{k,f_k}$. 
The direct product of these games is denoted by $G^\times = G_1 \times G_2 \times \cdots \times G_k$ with $\mu^\times = \mu_1 \times \mu_2 \times \cdots \times \mu_k$ and predicate $\mc{R}^\times_{f_1, \ldots, f_k} = \mc{R}_{1, f_1} \times \cdots \times \mc{R}_{k, f_k}$. 
\end{definition}
The probability that a query algorithm $\As$ wins $G^\times$ is $\Pr\left[ \As^{f_1, \ldots, f_k} \in \mc{R}_{1, f_1} \times \cdots \times \mc{R}_{k, f_k} \right]$, where $\As$ can make queries to any $f_1,\ldots,f_k$.
The direct product theorem intuitively says that the winning probability of $kT$-query algorithms for $G^\times$ should be roughly the product of the winning probabilities of $T$-query algorithms for each $G_i$.

To capture the algorithms in \Cref{clm:thm:classical_salt}, we focus on memoryless algorithms for direct product of games. In these algorithms, the workspace is repeatedly refreshed, which forbids long-range computation.

\begin{definition}[Memoryless Algorithm]\label{def:memoryless}
Let $G^\times = G_1\times \cdots \times G_k$ be a direct product of games.
A query algorithm $\As$ for $G^\times$ is memoryless if it is deterministic\footnote{The assumption about being deterministic is not necessary as one can fix the randomness of the algorithm. However assuming so is easier for our analysis.} and there exist query algorithms $\As_1, \ldots, \As_k$ such that $\As^{f_1, \ldots, f_k} \equiv \As_1^{f_1, \ldots, f_k}, \ldots, \As_k^{f_1, \ldots, f_k}$. 
That is, $\As$ runs $\As_1, \ldots, \As_k$ sequentially for $\mathcal{R}_{1,f_1},\ldots,\mathcal{R}_{k,f_k}$ and does not share memory between two executions. 
Moreover, we say $\As$ is $(T_1,\ldots, T_k)$-memoryless if each $\As_i$ makes at most $T_i$ queries.
\end{definition}

Note that, while $(T_1,\ldots,T_k)$-memoryless algorithms are guaranteed to make at most $T_i$ queries for each predicate $\mathcal{R}_{i,f_i}$, the queries can be interlaced and highly dependent on the query history. 
Therefore, it is generally not the case that a $(T_1,\ldots,T_k)$-memoryless algorithm would simply execute a $T_1$-query algorithm on $f_1$, followed by a $T_2$-query algorithm on $f_2$, and so on. 
However, as shown in \Cref{thm:strong_memoryless_dpt}, this special strategy turns out to be optimal.

\begin{theorem}[Strong Direct Product Theorem; {\cite[Theorem 3.1]{nisan1998products}}]\label{thm:strong_memoryless_dpt}
Let $G^\times = G_1\times \cdots \times G_k$ be a direct product of games.
If $\As$ is a $(T_1,\ldots,T_k)$-memoryless algorithm for $G^\times$, then its winning probability is at most $\prod_{i=1}^k \epsC_{G_i}(T_i)$. 
\end{theorem}

Given \Cref{thm:strong_memoryless_dpt}, we can easily deduce \Cref{clm:thm:classical_salt}.

\begin{proof}[Proof of \Cref{clm:thm:classical_salt}]
With $(U,\{(k_i,\ch_{k_i})\}_i,\{f_j\}_j)$ fixed, each $\Bs^{f_1,\ldots,f_K}(k_i,\ch_{k_i})$ is a deterministic $T$-query algorithm on $\{f_{k_i}\}_i$.
Overall, they form a $(T,\ldots,T)$-query algorithm to solve challenge-conditioned games $\{G_{\ch_{k_i}}\}_i$ in a memoryless way.
Hence the desired bound follows from \Cref{thm:strong_memoryless_dpt}.
\end{proof}

To prove \Cref{thm:strong_memoryless_dpt}, we will convert memoryless algorithms to another special type of algorithms for which strong direct product theorems are easier to establish.

\begin{definition}[Fair Algorithms]\label{def:fair_algo}
Let $G^\times = G_1\times \cdots \times G_k$ be a direct product of games. We say a query algorithm $\As$ for $G^\times$ is $(T_1, \ldots, T_k)$-fair if it is deterministic\footnote{The assumption about being deterministic is not necessary as one can fix the randomness of the algorithm. However assuming so is easier for our analysis.} and makes at most $T_i$ queries to each $f_i$ in any execution. 
\end{definition}

We emphasize that a $(T_1,\ldots,T_k)$-fair algorithm does not need to be memoryless.
Moreover, the queries to different oracles can be mixed up as well.
Therefore, fair algorithms and memoryless algorithms are incomparable.
By an induction on the total number of queries, one can prove a sharp direct product theorem for fair algorithms \cite{shaltiel2003towards}.

\begin{theorem}[{\cite[Theorem 5.2]{shaltiel2003towards}}]\label{thm:fair_algo}
Let $G^\times = G_1\times \cdots \times G_k$ be a direct product of games.
If $\As$ is a $(T_1, \dots, T_k)$-fair algorithm for $G^\times$, then its winning probability is at most $\prod_{i=1}^k \epsC_{G_i}(T_i)$.
\end{theorem}

We remark that \cite[Theorem 5.2]{shaltiel2003towards} proves for the case where $T_1=\cdots=T_k$ and each $G_i$ is a decision game (i.e., each $\mathcal{R}_{i,f_i}$ is a singleton set of a binary bit). However the proof can be extended to our setting with minimal changes.
We also note that Drucker \cite{drucker2012improved} gives an alternative proof using martingale analysis, which generalizes beyond fair algorithms with a necessary loss on the bounds.

In light of \Cref{thm:fair_algo} and to prove \Cref{thm:strong_memoryless_dpt}, it suffices to present a reduction from memoryless algorithms to fair algorithms which we present below. Note that \Cref{thm:strong_memoryless_dpt} follows immediately by the combination of \Cref{lem:reduction_memoryless_to_fair} and \Cref{thm:fair_algo}.

\begin{lemma}\label{lem:reduction_memoryless_to_fair}
Let $G^\times = G_1\times \cdots \times G_k$ be a direct product of games.
Let $\As$ be any $(T_1, \ldots, T_k)$-memoryless algorithm for $G^\times$. Then there exists a $(T_1, \ldots, T_k)$-fair algorithm $\Bs$ with winning probability at least the winning probability of $\As$. 
\end{lemma}

To prove \Cref{lem:reduction_memoryless_to_fair}, we will frequently use the notion of partial assignments.

\begin{definition}[Partial assignment]\label{def:partial_assignment}
Any string $u \in ([M] \cup \{*\})^N$ is a partial assignment for functions of range $[M]$ and domain $[M]$.
For a function $f\colon[M] \to [N]$, we say it is consistent with $u$ (or $f \lhd u$) if $f(x) = u_x$ holds for any $x$ that $u_x\neq*$.
\end{definition}

Now we present the reduction and prove \Cref{lem:reduction_memoryless_to_fair}.

\begin{proof}[Proof of \Cref{lem:reduction_memoryless_to_fair}]
Assume without loss of generality that the oracle of each game $G_i$ is a function from $[M]$ to $[N]$.
Recall \Cref{def:memoryless}. There exists $k$ algorithms $\As_1, \ldots, \As_k$, each making at most $T_1,\ldots,T_k$ queries respectively and $\As$ executes as follows:

\begin{mdframed}[]
    \textbf{Algorithm:} The memoryless algorithm $\As$
    \begin{itemize}
        \item $\As$ executes the algorithm one by one:
        \begin{enumerate}
            \item It runs $\As_1$ and obtains an output $e_1$;
            \item It runs $\As_2$ and obtains an output $e_2$;
            \item[]  $\cdots$
            \item[k.] It runs $\As_k$ and obtains an output $e_k$.
        \end{enumerate}

        \item Eventually, $\As$ outputs $e_1,\ldots, e_k$. 
    \end{itemize}
\end{mdframed}

    Since all the algorithms $\As_1, \ldots, \As_k$ commute with each other, $\As$ can run an ``out-of-order execution'' of them, so long as the query order within the same algorithm $\As_i$ does not change. 
    More precisely, the following algorithm $\widetilde{\As}$ has the same output distribution as $\As$:

\begin{mdframed}
    \textbf{Algorithm:} The ``out-of-order execution'' algorithm $\widetilde{\As}$
    \begin{itemize}
    \item $\widetilde{\As}$ initializes $\As_1, \ldots, \As_k$, each $\As_i$ makes no query yet.
    \item As long as there exists $j \in [k]$ such that $\As_j$ has not terminated yet.
        \begin{enumerate}
            \item $\widetilde{\As}$ picks a subset $\mc{J} \subseteq [k]$, where $\As_j$ has not terminated for every $j \in \mc{J}$. 
            \item For each $j\in\mc J$, it runs $\As_j$ by one query further.
        \end{enumerate}
    \item $\widetilde{\As}$ outputs the results $e_1,\ldots,e_k$ of $\As_1, \ldots, \As_k$. 
    \end{itemize}
\end{mdframed}
The subset $\mc J$ constructed by $\widetilde{\As}$ at Step 1 can be arbitrary and adaptive. Before presenting its construction, we note the following important observation.

\begin{lemma}\label{lem:simulate_fj}
Assume $\As_j$ terminates at some stage of $\widetilde{\As}$. Let $u$ be the partial assignment of $f_j$ that $\widetilde{\As}$ learned so far. Consider the following two cases:
\begin{enumerate}
    \item \textbf{The real execution.} $\widetilde{\As}$ finishes the rest of its execution with the real $f_j$. Let the winning probability in the case be $p$. 
    \item \textbf{The simulated execution.} Given $u$, $\widetilde{\As}$ picks some function $f'_j$ consistent with $u$ and finishes the rest of its execution with the fake $f'_j$. Let the winning probability in the case be $p'$. 
\end{enumerate}
Then there always exists a choice of $f'_j$ such that $p' \geq p$. 
\end{lemma}
\begin{proof}[Proof of \Cref{lem:simulate_fj}]
Without loss of generality assume $j = 1$. Since $\As_1$ already terminates, its outcome $e_1$ is fixed. We have
\begin{align*}
    p = \sum_{f_1 \lhd u} [e_1 \in \mathcal{R}_{1, f_1}] \cdot \Pr_{f \sim \mu_1^{(u)}}[f = f_1] \cdot \Pr_{f_2, \ldots, f_k}\left[ \As_2, \ldots, \As_k \text{ win} \middle| f_1 \right],
\end{align*}
where $[e_1 \in \mathcal{R}_{1, f_1}]$ is the binary indicator.
On the other hand, we have
\begin{align*}
    p' = \sum_{f_1 \lhd u} [e_1 \in \mathcal{R}_{1, f_1}] \cdot \Pr_{f \sim \mu_1^{(u)}}[f = f_1]  \cdot \Pr_{f_2, \ldots, f_k}\left[ \As_2, \ldots, \As_k \text{ win} \,|\, f'_1 \right]. 
\end{align*}
Therefore, we can pick $f'_1 := {\sf argmax}_{f_1 \lhd u} \Pr_{f_2, \ldots, f_k}\left[ \As_2, \ldots, \As_k \text{ win} \,|\, f_1 \right]$ and then $p' \geq p$. 
\end{proof}

With the above observation, we now consider the following algorithm $\Bs$ that is almost identical to $\widetilde{\As}$ except that whenever some $\As_j$ terminates, it picks a function $f_j'$ consistent with the partial assignement learned so far and runs the rest of the protocol using $f'_j$.
By \Cref{lem:simulate_fj}, $\Bs$ wins with probability at least that of $\widetilde{\As}$. 

\begin{mdframed}
    \textbf{Algorithm:} The algorithm $\Bs$
    \begin{itemize}
    \item $\Bs$ initializes $\As_1, \ldots, \As_k$, each $\As_i$ makes no query yet.
    \item As long as there exists $j \in [k]$ such that $\As_j$ has not terminated yet.
        \begin{enumerate}
            \item $\Bs$ picks a subset $\mc{J} \subseteq [k]$, where $\As_j$ has not terminated for every $j \in \mc{J}$. 
            \item For every $j \in \mc{J}$, it runs $\As_j$ by one query further.
            
            If $\As_j$ terminates, $\Bs$ picks $f'_j$ (as in \Cref{lem:simulate_fj}). In the rest of the execution, whenever $f_j$ is queried, $\Bs$ uses $f'_j$. 
        \end{enumerate}
    \item $\Bs$ outputs the results of $\As_1, \ldots, \As_k$. 
    \end{itemize}
\end{mdframed}

Finally, we show that there exists a way of picking $\mc{J}$ such that at every stage and for each $i\in[k]$, the number of queries to $f_i$ is at most the number of queries made by $\As_i$.
This is in general not true, for instance $\As_1$ and $\As_2$ can make all their $T_1+T_2$ queries to $f_1$. 
However by the reduction above and the memoryless property, we can now rearrange queries by carefully designing $\mc J$ to avoid these cases.

\begin{lemma}\label{lem:query_order}
    There exists a strategy of picking $\mc{J}$ such that at any stage and for each $i\in[k]$, we have $t_i\le q_i$, where $q_i$ is the number of queries made by $\As_i$ and $t_i$ be the number of queries made to (the real) $f_i$.
\end{lemma}

A direct consequence of the above lemma is that, when we use the strategy to pick $\mc{J}$ for $\Bs$, it becomes a fair algorithm since $t_i \leq q_i \leq T_i$.

\begin{proof}[Proof of \Cref{lem:query_order}]
    We prove the statement by induction. It is obviously true at the initialization stage where $t_i = q_i = 0$ for every $i \in [k]$. 
    
    For inductive hypothesis, assume that each $\As_i$ already makes $q_i$ queries and each $f_i$ gets queried $t_i$ times with $t_i \leq q_i$. 
    For the next stage, we define a directed graph $(V, E)$ as follows:
    \begin{itemize}
        \item Let $V\subseteq[k]$ be the set of $\As_i$'s that has not terminated yet. 
        \item For every $r \in V$,
        \begin{itemize} 
        \item if $A_r$ makes the next query to $f_v$ and $v \in V$, then a direct edge $\langle r, v \rangle \in E$;
        \item otherwise $A_r$ makes the next query to $f_v$ and $v \not\in V$, then no edge will be added. 
        
        Since $v \not\in V$ means $\As_v$ already terminates and thus every query to $f_v$ will be directed to the simulated oracle $f'_v$ constructed by \Cref{lem:simulate_fj}.
        \end{itemize}
    \end{itemize}

    Observe that every node in the graph has out-degree at most one and $V\neq\emptyset$. 
    We complete the proof of \Cref{lem:query_order} by the following case analysis.

    \paragraph*{There exists as a directed cycle.} 
    If $r_1\to r_2\to\cdots\to r_\ell\to r_1$ forms a directed cycle, then we define $\mc{J} = \{r_1,\ldots, r_\ell\}$. Therefore each $\As_{r_i}$ will make a query and each function $f_{r_i}$ gets one more query. The induction hypothesis holds. 

    \paragraph*{There is no directed cycle.}
    Then there exists a direct path $r_1\to r_2\to\cdots\to r_\ell$ where $\ell\ge1$.
    Assume $A_{r_\ell}$ makes the next query to $f_{v_\ell}$.
    Since $r_\ell$ is the endpoint of the path, we know that $\As_{v_\ell}$ already terminates and we have replaced $f_{v_\ell}$ with the simulated oracle $f_{v_\ell}'$ by \Cref{lem:simulate_fj}.
    
    Now let $\mc{J} = \{r_1, \ldots, r_\ell\}$. Each $\As_{r_i}$ makes a new query, and functions $f_{r_2}, \ldots, f_{r_\ell}$ get one more query respectively.
    Note that $f_1$ is not queried and the query of $\As_{r_\ell}$ is simulated by $f_{v_\ell}'$ without querying the actual oracles. Thus, the induction hypothesis holds. 
\end{proof}

By \Cref{lem:query_order}, the algorithm $\Bs$ is indeed $(T_1,\ldots,T_k)$-fair.
Then \Cref{lem:simulate_fj} guarantees that $\Bs$ has at least the winning probability as that of $\As$, which completes the proof of \Cref{lem:reduction_memoryless_to_fair}.
\end{proof}

\subsection{Improved Estimates for Games with Numerous Challenges}\label{sec:classical_salt_refine}

In this section, we provide an alternative way of analyzing the RHS of \Cref{eq:thm:classical_salt_2}, which implies improved bounds for games with numerous challenges.

For a game with challenge, we define the multi-challenge version of it by asking the algorithm to solve multiple independent challenges.
Note that the repetition here is on challenges instead of salts.

\begin{definition}[Multi-Challenge Game]\label{def:multiinstance}
Let $G$ be a game specified by $(\mu,\{\pi_f\}_f,\{\mathcal{R}_{f,\ch}\}_{f,\ch})$.
For each $n\in\mathbb N$, we define the multi-challenge game $G^n$ as follows:
\begin{itemize}
\item The oracle distribution is the same $\mu$.
\item For each possible oracle $f$ from $\mu$, the challenge $\ch=(\ch^1,\ldots,\ch^n)$ is produced by sampling independent $\ch^1,\ldots,\ch^n\sim\pi_f$.
\item For each possible oracle $f$ and challenge $\ch=(\ch^1,\ldots,\ch^n)$, the accepting outcome set $\mathcal{R}_{f,\ch}=\mathcal{R}_{f,\ch^1}\times\cdots\times\mathcal{R}_{f,\ch^n}$.

In other words, the algorithm needs to output accepting configuration for each challenge.
\end{itemize}
We remark that if $n=0$, then the multi-challenge game $G^0$ is a trivial game.
\end{definition}

Note that the definition of optimal winning probabilities carries over.
In particular, $\epsC_{G^n}(nT)$ refers to the optimal winning probability of a $nT$-query classical algorithm for the multi-challenge game $G^n$.
Oftentimes, $\epsC_{G^n}(nT)^{1/n}$ is a good upper bound for the non-uniform security of $G$ against algorithms with $n$-bit advice and $T$ queries (see e.g., \cite{impagliazzo2010constructive,impagliazzo2011relativized,gravin2021concentration}).

Using this notion, we provide the following variant of \Cref{thm:classical_salt} and then instantiate it on the function / permutation inversion problem (see e.g., the survey by Corrigan-Gibbs and Kogan \cite{corrigan2019function}).

\begin{theorem}\label{thm:classical_salt_mult}
Let $G$ be a game and let $G_K$ be the salted game for $K\in\mathbb{Z}_+$.
Then for any $S,T\in\mathbb N$, we have
$$
\epsC_{G_K}(S,T)\le\min_{L\in\mathbb{Z}_+}2^{S/L}\cdot\left(\frac{L!}{K^L}\sum_{\substack{n_1,\ldots,n_K\in\mathbb N\\n_1+\cdots+n_K=L}}\prod_{k=1}^K\frac{\epsC_{G^{n_k}}(n_kT)}{n_k!}\right)^{1/L}.
$$
\end{theorem}
\begin{proof}
The argument deviates from the proof of \Cref{thm:classical_salt} on the analysis of the RHS of \Cref{eq:thm:classical_salt_2}.
We will prove the following alternative bound which completes the proof of \Cref{thm:classical_salt_mult}:
\begin{equation}\label{eq:thm:classical_salt_mult_1}
\Pr_{\substack{f_1,\ldots,f_K\\(k_i,\ch_{k_i}),i\in[L]}}\left[\Bs^{f_1,\ldots,f_K}(k_i,\ch_{k_i})\in\mathcal{R}_{f_{k_i},\ch_{k_i}},\forall i\in[L]\right]
\le
\frac{L!}{K^L}\sum_{\substack{n_1,\ldots,n_K\in\mathbb N\\n_1+\cdots+n_K=L}}\prod_{k=1}^K\frac{\epsC_{G^{n_k}}(n_kT)}{n_k!}
\end{equation}
To establish \Cref{eq:thm:classical_salt_mult_1}, for each $k\in[K]$ we define $n_k=|\left\{i\in[L]\colon k_i=k\right\}|$ to be the number of salt $k$'s appearances in $k_1,\ldots,k_L$.
Note that $n_1,\ldots,n_K$ depend only on $k_1,\ldots,k_L$.
Then the following procedure is equivalent to the LHS of \Cref{eq:thm:classical_salt_mult_1}:
\begin{enumerate}
\item\label{itm:thm:classical_salt_mult_1} 
First sample $n_1,\ldots,n_K$ and $k_1,\ldots,k_L$.

For each $k\in[K]$, let $S_k\subseteq[L]$ of size $n_k$ be the appearances of salt $k$ in $k_1,\ldots,k_L$.
\item\label{itm:thm:classical_salt_mult_2} For each $k\in[K]$, we sample $f_k$ and execute the following subroutine:
for each $i\in S_k$, sample $\ch_{k_i}$ and execute $\Bs^{f_1,\ldots,f_K}(k,\ch_{k_i})$ to check if the outcome is in $\mathcal{R}_{f_k,\ch_{k_i}}$.
\end{enumerate}
Recall that $\Bs$ is a $T$-query algorithm.
Therefore, conditioned on \Cref{itm:thm:classical_salt_mult_1}, each subroutine in \Cref{itm:thm:classical_salt_mult_2} is an $n_kT$-query algorithm aiming to solve the multi-challenge game $G^{n_k}=|S_k|$.
In particular, their combination is a memoryless algorithm.
Therefore we have the following result analogous to \Cref{clm:thm:classical_salt}.

\begin{claim}\label{clm:thm:classical_salt_mult}
Conditioned on \Cref{itm:thm:classical_salt_mult_1}, \Cref{itm:thm:classical_salt_mult_2} succeeds with probability at most $\prod_{k=1}^K\epsC_{G^{n_k}}(n_kT)$.
\end{claim}
\begin{proof}[Proof of \Cref{clm:thm:classical_salt_mult}]
The reduction to \Cref{thm:strong_memoryless_dpt} is similar to the proof of \Cref{clm:thm:classical_salt}.
We only remark that while \Cref{thm:strong_memoryless_dpt} is proved for plain games (i.e., games without challenge), the same result naturally extends to games with challenge by conditioning on the challenge as in the proof of \Cref{thm:classical_salt}.
\end{proof}

Given \Cref{clm:thm:classical_salt_mult}, we now verify \Cref{eq:thm:classical_salt_mult_1}:
\begin{align*}
\text{LHS of \Cref{eq:thm:classical_salt_mult_1}}
&=\E_{\substack{n_1,\ldots,n_K\\k_1,\ldots,k_L}}\left[\prod_{k=1}^K\epsC_{G^{n_k}}(n_kT)\right]
=\E_{n_1,\ldots,n_K}\left[\prod_{k=1}^K\epsC_{G^{n_k}}(n_kT)\right]
\tag{by \Cref{clm:thm:classical_salt_mult}}\\
&=\sum_{\substack{n_1,\ldots,n_K\in\mathbb N\\n_1+\cdots+n_K=L}}\binom L{n_1,\ldots,n_K}\cdot\frac1{K^L}\cdot\prod_{k=1}^K\epsC_{G^{n_k}}(n_kT)\\
&=\sum_{\substack{n_1,\ldots,n_K\in\mathbb N\\n_1+\cdots+n_K=L}}\frac{L!}{K^L}\cdot\prod_{k=1}^K\frac{\epsC_{G^{n_k}}(n_kT)}{n_k!}
=\text{RHS of \Cref{eq:thm:classical_salt_mult_1}}
\tag*{\qedhere}
\end{align*}
\end{proof}

Now we present a variant of \Cref{thm:classical_salt_mult} that is potentially easier to use and is in fact tight for the large advice regime (see \Cref{cor:large_advice}) and some specific important games (see \Cref{cor:funcinv}).

\begin{corollary}\label{cor:classical_salt_mult}
Let $G$ be a game and let $G_K$ be the salted game for $K\in\mathbb{Z}_+$.
Then for any $S,T\in\mathbb N$, we have
$$
\epsC_{G_K}(S,T)\le 2e^2\cdot\min_{L\in\mathbb{Z}_+}2^{S/L}\max_{\substack{n_1,\ldots,n_K\in\mathbb N\\n_1+\cdots+n_K=L}}\sum_{k=1}^K\frac{\epsC_{G^{n_k}}(n_kT)^{1/n_k}}{\min\{K,L\}}.
$$
\end{corollary}
\begin{proof}
For any fixed $n_1,\ldots,n_K$ with $n_1+\cdots+n_K=L$, we have
\begin{align}
\left(\prod_{k=1}^K\frac{\epsC_{G^{n_k}}(n_kT)}{n_k^{n_k}}\right)^{1/L}
&=\left(\prod_{k=1}^K\left(\frac{\epsC_{G^{n_k}}(n_kT)^{1/n_k}}{n_k}\right)^{n_k}\right)^{1/L}
\notag\\
&\le\sum_{k=1}^K\frac{\epsC_{G^{n_k}}(n_kT)^{1/n_k}}{n_k}\cdot\frac{n_k}L
\tag{by the weighted AM-GM inequality}\\
&=\sum_{k=1}^K\frac{\epsC_{G^{n_k}}(n_kT)^{1/n_k}}L,
\label{eq:thm:classical_salt_mult_2}
\end{align}
where we note that $\epsC_{G^{n_k}}(n_kT)=1$ (and thus $\epsC_{G^{n_k}}(n_kT)^{1/n_k}=1$) if $n_k=0$.
Therefore for any fixed $L\in\mathbb{Z}_+$, we have
\begin{align*}
\epsC_{G_K}(S,T)
&\le2^{S/L}\cdot\left(\frac{L!}{K^L}\sum_{\substack{n_1,\ldots,n_K\in\mathbb N\\n_1+\cdots+n_K=L}}\prod_{k=1}^K\frac{\epsC_{G^{n_k}}(n_kT)}{n_k!}\right)^{1/L}
\tag{by \Cref{thm:classical_salt_mult}}\\
&\le2^{S/L}\cdot\left(\frac{L^L}{K^L}\sum_{\substack{n_1,\ldots,n_K\in\mathbb N\\n_1+\cdots+n_K=L}}\prod_{k=1}^K\frac{e^{n_k}\cdot \epsC_{G^{n_k}}(n_kT)}{n_k^{n_k}}\right)^{1/L}
\tag{since $a!\ge(a/e)^a$}\\
&\le e\cdot 2^{S/L}\cdot \frac LK\cdot\binom{L+K-1}{K-1}^{1/L}\cdot\max_{\substack{n_1,\ldots,n_K\in\mathbb N\\n_1+\cdots+n_K=L}}\left(\prod_{k=1}^K\frac{\epsC_{G^{n_k}}(n_kT)}{n_k^{n_k}}\right)^{1/L}\\
&\le e\cdot 2^{S/L}\cdot \frac LK\cdot\binom{L+K-1}{K-1}^{1/L}\cdot\max_{\substack{n_1,\ldots,n_K\in\mathbb N\\n_1+\cdots+n_K=L}}\sum_{k=1}^K\frac{\epsC_{G^{n_k}}(n_kT)^{1/n_k}}L.
\tag{by \Cref{eq:thm:classical_salt_mult_2}}
\end{align*}
By \Cref{fct:binom_choice}, we know that if $L\le K$, then $\binom{L+K-1}{K-1}^{1/L}\le 2e\cdot \frac KL$ and
$$
\epsC_{G_K}(S,T)\le 2e^2\cdot2^{S/L}\max_{\substack{n_1,\ldots,n_K\in\mathbb N\\n_1+\cdots+n_K=L}}\sum_{k=1}^K\frac{\epsC_{G^{n_k}}(n_kT)^{1/n_k}}L
$$
and otherwise $L\ge K$, then $\binom{L+K-1}{K-1}^{1/L}\le (2e\cdot \frac LK)^{K/L} \le 2e$ and
$$
\epsC_{G_K}(S,T)\le 2e^2\cdot2^{S/L}\max_{\substack{n_1,\ldots,n_K\in\mathbb N\\n_1+\cdots+n_K=L}}\sum_{k=1}^K\frac{\epsC_{G^{n_k}}(n_kT)^{1/n_k}}K
$$
as desired.
\end{proof}

To demonstrate the tightness of \Cref{thm:classical_salt_mult} and \Cref{cor:classical_salt_mult}, we consider the function / permutation inversion game $G$, for which an upper bound of $\epsC_{G^n}(nT)$ is known.

\begin{definition}[Function / Permutation Inversion Game]\label{def:funcinv}
The function inversion game $\mathsf{Inv}$ is the following game:
\begin{itemize}
\item The oracle is a uniformly random function $f\colon[N]\to[N]$.
\item The challenge $\ch$ is a uniformly random element in $[N]$.
\item An answer $a\in[N]$ is an accepting configuration if $f(a)=\ch$.
\end{itemize}
The permutation inversion game can be similarly defined where the oracle is then a uniformly random permutation.\footnote{Since \Cref{fct:funcinv} holds for both function inversion game and permutation inversion game, we do not distinguish them here with extra notation. However we remark that \Cref{fct:funcinv} is known to be tight for permutation inversion, but conjectured not tight for function inversion. Strengthening the latter will establish the non-uniform security of function inversion, which is a long-standing open problem in cryptography with deep connections to circuit lower bounds and data structure lower bounds \cite{corrigan2019function}.}
\end{definition}
That is, an adversary wins $\mathsf{Inv}$ iff it successfully inverts a random point in a random function / permutation.
For our purpose, we quote a known upper bound on the multi-challenge version of $\mathsf{Inv}$ from \cite{gravin2021concentration}.

\begin{fact}[\cite{gravin2021concentration}]\label{fct:funcinv}
$\epsC_{\mathsf{Inv}^n}(nT)\le O\left(\frac{nT}N\right)^n$.
\end{fact}

Plugging \Cref{fct:funcinv} into \Cref{cor:classical_salt_mult} and choosing $L=S$, we get the following immediate corollary, which gives another proof of the results in \cite{EC:DodGuoKat17}.

\begin{corollary}\label{cor:funcinv}
For any $S,T\in\mathbb N$, we have
$$
\epsC_{\mathsf{Inv}_K}(S,T)\le O\left(\frac{ST}{N\cdot\min\{S,K\}}\right)=O\left(\frac{T}{N} + \frac{ST}{KN}\right).
$$
\end{corollary}

We remark that using \Cref{thm:classical_informal}, we can only get a weaker bound of
$$
O\left(\frac TN+\frac SK\right).
$$
Indeed, the message from \Cref{cor:funcinv} is much stronger: salting can eliminate the additional advantages of non-uniform advice once the number of salts exceeds the length of the advice string.

\section{Non-Uniform Security of Salting in the QROM}\label{sec:salting_q}

In this section, we prove the non-uniform security of salting in the quantum random oracle model (QROM) for property finding problems.

\subsection{Preliminaries: Database and Property on Databases}

We consider a random oracle $f : [M] \to [N]$. For convenience, let $\X = [M]$ be the domain of the random oracle. Let $\Y = [N]$ be the range of the random oracle. 
For $0 \leq \nu \leq M$, let $\binom{\X}{\nu}$ be the set of all subsets of $\X$ with size $\nu$. We use $\vx \in \binom{\X}{\nu}$ to denote a vector $\vx$ of length $\nu$ such that $1 \leq x_1 < x_2 < \cdots < x_\nu \leq M$. 
\begin{definition}[Database]
A {database} $D = \{(x_i, y_i)\}_{i \in [\nu]}$ is defined by $\vx \in \binom{\X}{\nu}$, $\vy \in \Y^\nu$ of the same length.  $|D|=\nu$ is the length/size of the database. An empty database is denoted by $\emptyset$. 
\end{definition}

Loosely speaking, a database $D$ defines knowledge about a random oracle: namely, the random oracle $f$ will output $y_i$ on $x_i$ for $1 \leq i \leq \nu$; its behavior on other inputs is undefined. 

\begin{remark}
Although a database $D$ is similar to a partial assignment, we will use the word ``database'' only in the quantum setting. $D$ is used in the compressed oracle technique by Zhandry~\cite{zhandry2019record}; even if it intuitively describes a partial assignment that an algorithm learns about a function, databases are ``stored'' in superposition and can also be unlearned by a quantum algorithm. Due to the differences between classical and quantum algorithms, we will adopt the notation ``database'' in Zhandry's work for consistency in our analysis.
\end{remark}

For a pair $(x, y) \in \X \times \Y$, $(x, y) \in D$ denotes that $D$ contains this pair; for an input $x \in \X$, $x \in_\X D$ denotes that there exists $y \in \Y$ such that $(x, y) \in D$. We similarly define $(x, y) \not\in D$, $x \not\in_\X D$, $y \in_\Y D$, and $y \not\in_\Y D$. For $(x, y) \in D$, we define $D(x) := y$; otherwise $x \not\in_\X D$, $D(x)$ is set as a special symbol $\bot$.

We say a database $D \subseteq D'$ if every $(x_i, y_i) \in D$ is also contained in $D'$. For $x \not\in_{\X} D$ and $y \in \Y$, define $D \cup \{(x, y)\}$ as the database with $(x, y)$ added into $D$. 

Define $\Ds_{\X, \Y}$ as the set of all databases of a random oracle: 
\begin{align*}
    \Ds_{\X, \Y} = \left\{ D = \{(x_i, y_i)\}_{i \in [\nu]} \,:\,\vx \in \binom{\X}{\nu}, \vy \in \Y^\nu, \,\forall \nu \in [M] \right\}
\end{align*} %
We will ignore the subscripts when $\X$ and $\Y$ are clear from the context and simply write $\Ds$.

\begin{definition}[Property of Databases]
A property $P$ of databases $\Ds$ is a subset of $\Ds$. For a property $P$, let $P^{(\nu)}$ denote all databases with property $P$ of size $\nu$ and $P^{(\leq \nu)}$ denote all databases with property $P$ of size at most $\nu$.  
\end{definition}

We say a property $P$ is \emph{monotone} if for any $D \in P$ and $D \subseteq D'$, we have $D' \in P$.

\subsection{Preliminaries: Quantum Random Oracle and Compressed Oracle}

Here we recall the background of quantum random oracle model and the compressed oracle technique introduced by~\cite{zhandry2019record}. This section is mostly taken verbatim from Section 2 of~\cite{guo2021unifying}, and we make some changes whenever these changes work better with our strong threshold direct product theorem.  

\paragraph*{Quantum Random Oracle Model.} 
\label{sec:QROM}
An oracle-aided quantum algorithm can perform quantum computation as well as quantum oracle queries. 
A quantum oracle query for an oracle $f:[M] \to [N]$ is modeled as a unitary $U_f: \ket x \ket u = \ket x \ket {u + f(x)}$, where $+$ here is the addition modulo $N$.

A random oracle is a random function $f: [M] \to [N]$. The random function $f$ is chosen at the beginning. A quantum algorithm making $T$ oracle queries to $f$ can be modeled as the following: 
it has three registers $\ket x, \ket u, \ket z$, where $x \in [M]$, $u \in [N]$, and $z$ is the algorithm's internal working memory. It starts with some input state $\ket{0}_{\cA}=\ket0\ket0\ket\psi$, then it applies a sequence of unitary to the state: $U_1$, $U_f$, $U_2$, $U_f$, $\ldots$, $U_{T}$, $U_f$, $U_{\sf final}$ and a final measurement over the computational basis. Each $U_f$ is the quantum oracle query unitary: $U_f \ket x \ket u = \ket x \ket {u + f(x)}$.
Each $U_i$ is a local quantum computation independent of $f$, and $U_{\sf final}$ is the local quantum computation (independent of $f$) after the last query. For a uniform algorithm, we have $\ket{0}_{\cA}=\ket{0}\ket0\ket0$ as the algorithm can initialize its internal register by $U_1$. 
By postponing measurement, we assume that the measurement only happens in the end and is in the computational basis.

\paragraph*{Compressed Oracle.} Compressed oracle is an analogy of the classical lazy sampling method.
To simulate a random oracle, one can sample $f(x)$ for all inputs $x$ and store everything in quantum accessible registers.
Such an implementation of a random oracle is inefficient and hard to analyze. 
Instead of recording all the information of $f$ in the registers, Zhandry provides a solution to argue the amount of information an algorithm knows about the random oracle. 

The oracle register records a database that contains the output on each input $x$; the output is an element in $[N]\cup \{\bot\}$, where $\bot$ is a special symbol denoting that the value is ``uninitialized''.
The database is initialized as an empty list $D_0$ of length $M$, i.e., the pure state $\ket{\emptyset}_f := \ket{\bot, \bot, \ldots, \bot}$. Let $|D|$ denote the number of non-$\bot$ entries in $D$. 
Define $D(x)$ to be the $x$-th entry of $D$.
$D(x)$ can be seen as the output of the oracle on $x$ if $D(x) \ne \bot$;  otherwise, the oracle's output on $x$ is still undetermined. 

For the compressed oracle $\compressO$, we have several oracle variations as defined in \cite[Section 3]{zhandry2019record}. They are equivalent, and the only difference is that the oracle registers and/or the query registers are encoded
in different ways between queries. Therefore we use $\compressO$ to denote the compressed oracle query, and when we do the simulation, we can implement $\compressO$ by any of its variations.

Here we define two oracle variations, the compressed standard oracle $\csto$ and the compressed phase oracle $\cphso$. 
\begin{align*}
    \csto &:= \stddecomp \cdot \csto' \cdot \stddecomp,\\
    \cphso &:= \stddecomp \cdot \cphso' \cdot \stddecomp.
\end{align*}
They are both unitary operators that operates on the joint system of the algorithm registers $\cA$ and oracle's registers $f$.

\begin{itemize}
    \item $\csto'\, |x, u\rangle |D\rangle = |x, u + D(x) \rangle |D\rangle$ when $D(x) \ne \bot$, which writes the output of $x$ defined in $D$ to the $u$ register. This operator will never be applied on an $x, D$ where $D(x) = \bot$.
    \item $\cphso'\, |x, u\rangle |D\rangle = (-1)^{u\cdot D(x)}|x, u \rangle |D\rangle$ when $D(x) \ne \bot$, which writes the output of $x$ defined in $D$ to the phase. This operator will never be applied on an $x, D$ where $D(x) = \bot$.
    \item $\stddecomp(\ket x \otimes \ket D) = \ket x \otimes \stddecomp_x \ket D$, where ${\sf StdDecomp}_x\, |D\rangle$ works on the $x$-th register of the database $D(x)$ and swaps a uniform superposition $\frac 1{\sqrt N} \sum_y \ket y$ with $\ket \bot$ on the $x$-th register:
        \begin{itemize}
            \item If $D(x) = \bot$, $\stddecomp_x$ maps $\ket \bot$ to $\frac{1}{\sqrt{N}} \sum_y \ket y$, or equivalently, $\stddecomp_x |D\rangle = \frac{1}{\sqrt{N}} \sum_y |D \cup (x, y)\rangle$. Intuitively, if the database does not contain information about $x$, it samples a fresh $y$ as the output of $x$. 
            
            \item If $D(x) \ne \bot$, $\stddecomp_x$ is an identity on $\frac{1}{\sqrt{N}} \sum_y \omega_N^{u y} \ket y$ for all $u \ne 0$ and maps the uniform superposition $\frac{1}{\sqrt{N}} \sum_y \ket y$ to $\ket \bot$, where $\omega_N$ is the $N$-th root of unity.
            
            More formally, for $D'$ with $D'(x) = \bot$, we have
            \begin{align*}
                \stddecomp_x  \frac{1}{\sqrt{N}} \sum_y \omega_N^{u y} |D' \cup (x, y)\rangle 
                            = \frac{1}{\sqrt{N}} \sum_y \omega_N^{u y}  |D' \cup (x, y)\rangle  \text{ for any } u \ne 0  ,
            \end{align*}
            and
            \begin{align*}   
                \stddecomp_x  \frac{1}{\sqrt{N}} \sum_y |D' \cup (x, y)\rangle  =|D'\rangle      .
            \end{align*}
        \end{itemize}
        Since all $\frac{1}{\sqrt{N}} \sum_y \omega_N^{u y} \ket y$ and $\ket \bot$ form a basis, these define a unique unitary operation.
\end{itemize} 

A quantum algorithm making $T$ oracle queries to a compressed oracle can be modeled as follows: 
it has three registers $\ket x, \ket u, \ket z$, where $x \in [M]$, $u \in [N]$, and $z$ is the algorithm's internal working memory. Started with some input state $\ket{0}_{\cA}=\ket{0}\ket{0}\ket{\psi}$, the joint state of the algorithm and the compressed oracle is $\ket{0}_{\cA} \otimes \ket {\emptyset}_f$. It then applies a sequence of unitary to the state: $U_1$, $\compressO$, $U_2$, $\compressO$, $\ldots$, $U_{T}$, $\compressO$, $U_{\sf final}$, and a final measurement over computational basis. For a uniform algorithm, we have $\ket{0}_{\cA}=\ket{0}\ket0\ket0$ as the algorithm can initialize its internal register by $U_1$. 

Zhandry \cite{zhandry2019record} proved that the quantum random oracle model and the compressed standard oracle model are perfectly indistinguishable by any \textit{unbounded} quantum algorithm. Due to the equivalence between oracle variations, this also holds for the compressed phase oracle model.

In this work, we only consider query complexity, and thus simulation efficiency is irrelevant to us.
We simulate a random oracle as a compressed phase oracle, i.e., $\compressO=\cphso$, to help us analyze security games with the help from the following lemmas due to~\cite{zhandry2019record,chung2020tight}. They are proved for the compressed standard oracle model, and due to the equivalence, they also hold for the compressed phase oracle.

The first lemma gives a general formulation of the overall state of $\cA$ and the compressed oracle after $\cA$ makes $T$ queries, even conditioned on arbitrary measurement results. It provides the intuition that one oracle query can append at most one element to the database register in superposition. 
\begin{lemma}\label{lem:bounded-database}
If $\cA$ makes at most $T$ queries to a compressed oracle, assuming the overall state of $\cA$ and the compressed oracle is $\sum_{x,u,z, D} \alpha_{x,u,z, D}\, \ket{x,u,z}_{\cA}\otimes\ket{D}_f$ where $\ket{x,u,z}$ is $\cA$'s registers and $\ket D$ is the oracle's registers, then it only has support on all $D$ such that $|D| \leq T$. In other words, the overall state can be written as
    \begin{align*}
        \sum_{x,u,z, D: |D| \leq T} \alpha_{x,u,z, D}\, |x,u,z\rangle_{\cA} \otimes |D\rangle_f.
    \end{align*}
\end{lemma}

The second lemma provides a quantum analogue of lazy sampling in the classical ROM. It gives us a way to connect the compressed oracle with the original oracle $U_f$, and therefore we can upper bound the winning probability of a quantum algorithm with access to a random oracle by analyzing the database / oracle register.

\begin{lemma}[{\cite[Lemma 5]{zhandry2019record}}]\label{lem:lem5-salt=1}
    Let $f\colon[M]\to[N]$ be a random oracle.
    Consider a quantum algorithm $\cA$ making queries to the standard oracle and outputting tuples $(x_1, \ldots, x_c, y_1, \ldots, y_c, z)$. {Suppose} the random function $f$ is measured after $\cA$ produces its output. Let $R$ be an arbitrary set of such tuples. Suppose with probability $p$, $\cA$ outputs a tuple such that (1) the tuple is in $R$ and (2) $f(x_i) = y_i$ for all $i$. Now consider running $\cA$ with the compressed standard oracle $\csto$, and suppose the database $D$ is measured after $\cA$ produces its output. Let $p'$ be the probability that (1) the tuple is in $R$ and (2) $D(x_i) = y_i$ (in particular, $D(x_i) \ne \bot$) for all $i$. Then $\sqrt{p} \leq \sqrt{p'} + \sqrt{c / N}$.
    
    This is also true for the compressed phase oracle $\cphso$. 
\end{lemma}

\subsection{Classical and Quantum Property Finding Problem}

\paragraph*{Classical Property Finding.}
The optimal winning probability of many important games can be captured by the following property finding problem. 

\begin{definition}[Classical Property Finding for Property $P$]
Let $P$ be a monotone property.\footnote{We are only interested in the case of $P$ being monotone, i.e., learning more information is always better.}
We say a query algorithm wins if $D \in P$, where $D$ is the database corresponding to the partial assignment a classical algorithm learns by querying a random oracle.  
\end{definition}

\noindent \textbf{Some Examples.} In the pre-image search problem, the goal is to find an input $x$ such that $f(x) = 0$. The property is defined as $P = \{D: 0\in_\Y D\}$. For the collision finding problem, we want to find a pair of distinct inputs that map to the same output. In this case, $P = \{ D:  D(x) = D(x')\text{ for some distinct } x,x' \in_\X D\}$.

\begin{definition}[Transition Probability]\label{def:transition-prob}
    For a monotone property $P$, we define $\overline{P}$ to be the complement of $P$, and the probability $p_\nu$ as
    \begin{align*}
        p_\nu = \max_{\substack{D \in \overline{P}^{(\leq \nu)} \\ x \not\in_\X D }} \Pr_{y \gets \Y}[D \cup \{(x, y)\} \in P],
    \end{align*}
    where we recall that $\overline{P}^{(\leq \nu)}$ denotes all databases in $\overline{P}$ of size at most $\nu$.
    By the monotonicity of $P$, $p_\nu$ is non-decreasing on $\nu$.
\end{definition}

\begin{lemma}\label{lem:classical_transition}
    Let $P$ be a monotone property with transition probability $p_t$. 
    For any classical $T$-query algorithm, the probability of winning the property finding problem of $P$ is at most $p_0 + p_1 + \cdots+ p_{T-1} \leq T\cdot p_T$.
\end{lemma}
Since we are not interested in the above classical lemma for the purpose of this work, we just offer an illustrative example to compare with the quantum analog of it.
\begin{proof}[Proof Sketch of \Cref{lem:classical_transition}]
    The probability of an algorithm reaching $D \in P$ is upper bounded by the summation of its probability of reaching $D$ at the $i$-th step for every $i \in  \{1, 2, \ldots, T\}$. Thus, the probability is at most $p_0 + \cdots + p_{T-1} \leq T\cdot p_T$. 
\end{proof}

\paragraph*{Quantum Property Finding Problem.}
In the quantum case, we work in the compressed oracle framework. At any stage of an algorithm, the joint system of the algorithm and the random oracle is described by a state over the algorithm register and the database register. 

let $\Lambda_P$ be a projection that in superposition checks whether the database is in $P$: 
$$
\Lambda_P = \sum_{D \in P} \ket D \bra D.
$$
We can similarly define the quantum version of property finding problem in the quantum random oracle model. 

\begin{definition}[Quantum Property Finding]\label{def:qproperty_finding}
    The quantum property finding problem for property $P$ is defined as follows: a quantum algorithm interacts with a compressed oracle; when the algorithm stops, a binary-valued measurement $\{\Lambda_P, \I - \Lambda_P\}$ is applied to the compressed oracle register and the algorithm wins if the measurement outcome is $0$ (corresponding to $\Lambda_P$, it ``finds'' the property). 
\end{definition}

\begin{theorem}[{\cite[Theorem 5.7]{chung2021compressed}}]\label{thm:quantum_lift_property_find}
    Let $P$ be a monotone property with transition probability $p_t$. 
    For any quantum $T$-query algorithm, the probability of its winning the quantum property finding problem for property $P$ is at most $\gamma_T = O(\sqrt{p_0} + \sqrt{p_{1}} + \cdots + \sqrt{p_{T-1}})^2 = O(T^2\cdot p_T)$.
\end{theorem}

\subsection{Main Quantum Theorem}

\begin{definition}
    Let $R$ be a collection of tuples in $X^c \times Y^c$. 
    Define $P_R$ to be the property on $\mathcal{D}$ such that $D \in P_R$ if and only if $D$ consists of a tuple in $R$: 
    \begin{align*}
        D \in P_R  \quad \iff \quad
        \exists (x_1, \ldots, x_c, y_1, \ldots, y_c) \in R \text{ such that } D(x_i) = y_i, \forall i \in \{1, 2, \ldots, c\}.
    \end{align*}
\end{definition}
It is easy to see that $P_R$ is monotone for any $R$. 

\medskip
The game we consider in the QROM is the following. 
\begin{definition}[Game $G_R$]
    
Let $R$ be a collection of tuples. $G_R$ askes a quantum-query algorithm to output $(x_1, \ldots, x_c, y_1, \ldots, y_c) \in R$. 
\end{definition}
Here we are interested in the case when $c$ is small. 
Some examples of $G_R$ include: finding pre-images of zero, collision finding,  $k$-SUM. 
Problems like computing the parity of the function has large $c$ and therefore are not covered in our main theorem. 
We now present our main theorem in the quantum case.

\begin{theorem}\label{thm:main_quantum}
    Let $R$ be a collection of tuples and $R \subseteq \X^c \times \Y^c$. Let $G_R$ be the game defined above. Let $p_T$ be the transition probability, and $\gamma_T$ be an upper bound on the winning probability of any $T$-quantum-query algorithm, derived from Zhandry's compressed oracle technique (\Cref{thm:quantum_lift_property_find}); $\gamma_T=\Theta(T^2\cdot p_T)$. If the transition probability $p_t$ is polynomial in $t$, i.e., $p_t=\Theta(t^{c_1}/N^{c_2})$ for some constants $c_1,c_2\ge 0$, then for every $T \geq c$, 
    \begin{align*}
        \epsQ_{G_{R,K}}(S, T) \leq \widetilde{O}\left( \gamma_T + \frac{S}{K} + \frac{c}{N}\right),
    \end{align*}
    where $\widetilde{O}$ hides low order terms in $\log N$.
\end{theorem}

Since an algorithm needs to output a $c$-tuple of input-output pairs that are consistent with a random oracle, we can without loss of generality only consider the case when $T \geq c$. 
For $G_R$ that we are interested in the work, the compressed oracle technique shows that $\epsQ_{G_R}(T) \leq O(\gamma_T + \frac{c}{N})$.
When the technique is tight, our theorem is also tight. 
Zhandry~\cite{zhandry2019record} demonstrated the tightness of~\Cref{thm:quantum_lift_property_find} for some special cases. Thus, we  get the following consequences. 
\begin{corollary}[Consequences of \Cref{thm:main_quantum}]
    For ${\sf CRHF}$ whose goal is to find two distinct inputs whose images are identical, we have
    \begin{align*}
        \epsQ_{{\sf CRHF}_K}(S, T) \leq \widetilde{O}\left(\frac{T^3}{N} + \frac{S}{K} \right). 
    \end{align*}

    For ${\sf Inv_0}$ whose goal is to find a pre-image of zero, we have
    \begin{align*}
        \epsQ_{{\sf Inv}_{0,K}}(S, T) \leq \widetilde{O}\left(\frac{T^2}{N} + \frac{S}{K} \right). 
    \end{align*}

    For ${\sf kSUM}$  whose goal is to find $k$ distinct inputs whose images sum up to zero, we have
    \begin{align*}
        \epsQ_{{\sf kSUM}_K}(S, T) \leq \widetilde{O}\left(\frac{T^{k+1}}{N} + \frac{S}{K} \right). 
    \end{align*}
\end{corollary}

\subsection{Quantum Direct Product Theorems for Property Finding Problem}
\label{sec:property_finding_sdpt}

To prove our main theorem \Cref{thm:main_quantum}, one crucial step is to give a (threshold) strong direct product theorem for the property finding problem. In this part, we introduce the notion of quantum direct product of property finding problems, and present the strong direct product theorems (SDPT) in this case.

\paragraph*{Quantum Direct Product of Property Finding Problem.} Now a random oracle $f$ is treated as $[K] \times \X \to \Y$. For convenience, $\KX := [K] \times \X$. The random oracle can be viewed as a concatenation of $K$ independent random oracles; namely,  for each $k \in [K]$, $f_k = f(k,\cdot)$ is a random oracle $\X \to \Y$. 

An input is denoted by $(k, x) \in \KX$ for $k \in [K]$ and $x \in \X$. A database $D$ is similarly defined by $\vx \in \binom{\KX}{\nu}$ and $\vy \in \Y^\nu$ of the same length $\nu$. It is also nature to define $((k, x), y) \in D, (k, x) \in_{\KX} D, y \in_\Y D$ and $D \subseteq D'$.

\begin{definition}[Restriction of Databases]
For a database $D \in \Ds_{\KX, \Y}$ and an index $k \in [K]$, we define the restriction of $D$ to $k$ as $D|_k \in \Ds_{\X, \Y}$: 
\begin{align*}
    D|_k = \{(x, y) \,:\, ((k, x), y) \in D\}. 
\end{align*}
In other words, $D|_k$ is a subset of $D$ with prefix $k$ removed.
\end{definition}

\begin{definition}[Direct Product of Property]
    Let $P$ be a property on $\Ds_{\X, \Y}$. 
    Let $\kappa \in \mathbb{Z}_+$ and $\KX = [\kappa] \times \X$.  
    Define the direct product property on $\Ds_{\KX, \Y}$ as $P^{\times \kappa}$:
    \begin{align*}
        D \in P^{\times \kappa} \quad \Longleftrightarrow \quad D|_{k} \in P, \forall k \in [\kappa]. 
    \end{align*}
\end{definition}

In the quantum setting, when the database register is a superposition state of databases, $D|_k\in P$ corresponds to the projected state where every superposition satisfies $D|_k\in P$. 

\begin{definition}[Success on Salts]
    Let $P$ be a property on $\Ds_{\X, \Y}$. 
    If $D|_k\in P$, we say that the database $D$ succeeds on salt $k$. If $\mathcal{K}$ is the set of $k$ such that $D|_k\in P$, we say that the database $D$ succeeds on $|\mathcal{K}|$ salts.
\end{definition}

\paragraph*{Quantum Direct Product Theorems.} 
We have the following direct product theorem for quantum property finding problem: 

\begin{theorem}[SDPT For Quantum Property Finding]\label{thm:qsdpt}
    Let $P$ be a monotone property with transition probability $p_t$.
    Let $\gamma_T$ be an upper bound on the winning probability of any $T$-quantum-query algorithm, derived from Zhandry's compressed oracle technique (\Cref{thm:quantum_lift_property_find}).
 
     Let $\kappa \in \mathbb{Z}_+$ and $P^{\times \kappa}$ be the direct product. If $p_T$ is polynomial in $T$, i.e., $p_T=\Theta(T^{c_1}/N^{c_2})$ for some constants $c_1,c_2\ge 0$, then with $\gamma_T=\Theta(T^2\cdot p_T)$, the maximum winning probability of any $\kappa T$-query quantum algorithm for the quantum property finding problem for $P^{\times \kappa}$ is at most $O(1)^\kappa\cdot \gamma_T^\kappa$.
\end{theorem}

Actually we can have a threshold version of \Cref{thm:qsdpt}.
\begin{theorem}[Threshold SDPT for Quantum Property Finding]\label{thm:threshold-qsdpt}
Let $P$ be a monotone property with transition probability $p_t$. 
Let $\gamma_T$ be an upper bound on the winning probability of any $T$-quantum-query algorithm, derived from Zhandry's compressed oracle technique (\Cref{thm:quantum_lift_property_find}). 

For salt space $[K]$ and threshold $\kappa\le K$, if $p_T$ is polynomial in $T$, i.e., $p_T=\Theta(T^{c_1}/N^{c_2})$ for some constants $c_1,c_2\ge 0$, the maximum winning probability for any $B$-query quantum algorithm ($B\in\mathbb{Z}_{+}$) to succeed on at least $\kappa$ salts in this quantum property finding problem is at most $O(1)^\kappa\cdot (\gamma_{B/\kappa})^\kappa$.
\end{theorem}

\begin{remark}\label{rmk:rounding?}
For \Cref{thm:threshold-qsdpt}, we can rewrite $p_T$ as the polynomial $p(T)=\Theta(T^{c_1}/N^{c_2})$, and in this way we define $p_{B/\kappa}:=p(B/\kappa)$ when $B/\kappa$ is not an integer. 
Similarly we define $\gamma_{B/\kappa}=t^2\cdot p_{t}=t^2\cdot p(t)$ for $t=B/\kappa$. 
\end{remark}

The proof of \Cref{thm:qsdpt} and \Cref{thm:threshold-qsdpt} can be found in \Cref{sec:qsdpt}, where we prove the strong direct product theorem for quantum property finding, and the threshold version as a corollary.

\subsection{\texorpdfstring{Proof of \Cref{thm:main_quantum}}{Proof of SDPT in the QROM}}
\label{subsec:prf-main_quantum}

In this section, we consider the following two-stage quantum algorithms $(\Bs, \As)$. 
\begin{itemize}
    \item Let $\Bs$ be a $2 S T$-quantum-query algorithm. First, $\Bs$ executes and outputs a binary predicate $b$, together with a quantum state.
    \item The quantum state is then given to $\As$. The goal of $\As$ is to win $G_{R, K}$ within $T$ quantum queries and with the state received from $\Bs$. 
\end{itemize}
We denote the event of $b = 0$ by $\mathbf{W}$.
To prove \Cref{thm:main_quantum}, we only need to prove the following theorem. 
\begin{theorem}\label{thm:switch_to_QBF}
     Let $G_{R,K}$ be a game and $c \in \mathbb{Z}^+$ as defined in \Cref{thm:main_quantum}.
     Assume there exists $\delta > \frac{1}{N}$ such that for every two-stage quantum algorithms $(\Bs, \As)$ with $2 S T$ and $T$ queries each, we have $\Pr[\mathbf{W}] > 1/N^S$ and
     \begin{align*}
         \Pr_{\substack{f\sim\mu\\ k\sim[K]}}\left[\cA^f(k)\text{ wins } \middle| \bW\right] \leq \delta.
     \end{align*}
     Then $\epsQ_{G_{R,K}}(S, T) \leq \delta$.
\end{theorem}
The proof is intuitively explained in the technical overview (Quantum Bit-Fixing section) and a formal proof is deferred to \Cref{sec:switch_to_QBP}.

Thus, our goal becomes to give an upper bound for $\Pr[\cA^f(k)\text{ wins } | \bW]$ when $\Pr[\bW]$ is not ``too small''.
We define $\bE$ to be the event that the database registers succeed on at most $2S\log N$ salts.
With $\bE$, we can give the following bound 
of $\Pr[\cA^f(k)\text{ wins } | \bW]$.

\begin{claim}\label{clm:upbound-W-E}
Define $\bE$ to be the event that when the algorithm $\Bs$ finishes, the database registers succeed on at most $2S\log N$ salts. Define $\overline{\bE}$ to be the complement of $\bE$ where the database registers can succeed on more than $2S\log N$ salts. Then,
\begin{align}\label{ineq:upbound-W-E}
    \Pr_{\substack{f\sim\mu\\ k\sim[K]}}[\cA^f(k)\text{ wins } | \bW] 
    &\le\Pr_{\substack{f\sim\mu\\ k\sim[K]}}[\cA^f(k)\text{ wins } | \bW\land \bE]+\frac{\Pr_{f\sim\mu}[\overline{\bE}]}{\Pr_{f\sim\mu}[\bW]}.
\end{align}
\end{claim}

From \Cref{clm:upbound-W-E}, we can give an upper bound for the two terms in RHS of \Cref{ineq:upbound-W-E} respectively. This leads us to the proof of \Cref{thm:main_quantum}. 
The main idea of the proof is as follows:
\begin{itemize}
    \item For the first term $\Pr[\cA^{f}(k)\text{ wins }|\bW\land \bE]$,
    \begin{enumerate}
        \item We apply \Cref{lem:lem5-salt=1} and reduce it to giving an upper bound for the quantum property finding problem, as in \Cref{clm:conditioned_lem5}.
        \item Then we upper bound that quantum property finding problem, as in \Cref{clm:conditioned_propertyfinding}.
    \end{enumerate} 
    \item For the second term $\Pr[\overline{\bE}]/\Pr[\bW]$, since $\Pr[\bW]\ge 1/N^S$, we only need to give an upper bound for $\Pr[\overline{\bE}]$. This corresponds to a threshold strong direct product theorem for the quantum property finding problem, as in \Cref{clm:threshold-qsdpt-E}.
\end{itemize}

\begin{claim} \label{clm:conditioned_lem5} For $\cA^f$ that outputs answers as a tuple $(x_1,\ldots,x_c,y_1,\ldots,y_c)$, we have the following relation between the original game and its corresponding property finding problem:
\begin{align}\label{ineq:conditioned_lem5}
    \Pr_{\substack{f\sim\mu\\ k\sim[K]}}[\cA^f(k)\text{ wins } | \bW\land \bE]
    \le 2\cdot\Pr_{\substack{f\sim\mu\\ k\sim[K]}}[(D|_k\in P)| \bW\land \bE]+\frac{2c}{N}.
\end{align}
Here $D$ corresponds to the database registers, and $D|_k\in P$ corresponds to the projected state that $D|_k\in P$ for every superposition; it means that $\cA^f$ succeeds on salt $k$ for the quantum property finding problem of $P$.
\end{claim}

\begin{claim}\label{clm:conditioned_propertyfinding} Conditioned on $\bW\land\bE$, we have the following upper bound of the maximum winning probability for the property finding problem on a random salt $k\sim [K]$: 
\begin{align}\label{ineq:conditioned_propertyfinding}
   \Pr_{\substack{f\sim\mu\\ k\sim[K]}}[(D|_k \in P)| \bW\land \bE] \le O(\gamma_{2T})+\widetilde{O}\left(\frac{S}{K}\right).
\end{align}
\end{claim}
\noindent In the theorem statement, $\gamma_t$ is assumed  to be a polynomial in $t$.
Thus, $\gamma_{2T} = O(\gamma_{T})$.

\begin{claim}\label{clm:threshold-qsdpt-E} For $\overline{\bE}$ as defined in \Cref{clm:upbound-W-E}, %
\begin{align}\label{ineq:threshold-qsdpt-E}
    \Pr_{f\sim\mu}[\overline{\bE}]\le \Pr_{f\sim\mu}[\bW]\cdot\frac{1}{N}.
\end{align}
\end{claim}

With \Cref{clm:upbound-W-E}, \Cref{clm:conditioned_lem5}, \Cref{clm:conditioned_propertyfinding}, and \Cref{clm:threshold-qsdpt-E}, we can conclude an upper bound for $\Pr[\cA^f(k)\text{ wins } | \bW]$. Together with \Cref{thm:switch_to_QBF}, this gives us an upper bound for the non-uniform security as in \Cref{thm:main_quantum}. 

Proofs of all the claims and \Cref{thm:main_quantum} can be found in \Cref{sec:app-prf-main_quantum}.

\section*{Acknowledgement}
KW wants to thank Guangxu Yang and Penghui Yao for references in direct product theorems.

\printbibliography
\appendix
\section{\texorpdfstring{Missing Proofs in \Cref{sec:prelim}}{Missing Proofs for the Classical Case}}
\label{app:missing_prelim}

\begin{proof}[Proof of \Cref{lem:draw_with_repl}]
We use standard martingale analysis.
For each $i=0,1,\ldots,L$, let $\ell_i$ be the number of distinct elements in $k_1,\ldots,k_i$.
Then $\ell_0=0$ and $\ell_L=\ell$.
Since each $\ell_i$ is at most $L$, we have
$$
\E\left[c^{\ell_i}\right]
=\E\left[\E\left[c^{\ell_i}\middle|k_1,\ldots,k_{i-1}\right]\right]
=\E\left[c^{\ell_{i-1}}\cdot\left(\frac{\ell_{i-1}}K+c\cdot\left(1-\frac{\ell_{i-1}}K\right)\right)\right]
\le\left(c+\frac LK\right)\cdot\E\left[c^{\ell_{i-1}}\right].
$$
The desired bound follows by applying the above inequality $L$ times.
\end{proof}

\section{\texorpdfstring{Missing Proofs in \Cref{subsec:prf-main_quantum}}{Missing Proofs for the Quantum Case}}
\label{sec:app-prf-main_quantum}

\subsection{Proof of Theorem \ref{thm:switch_to_QBF}}
\label{sec:switch_to_QBP}
\begin{proof}
We first apply the result in \cite{liu2023non} to relate the winning probability of a game against non-uniform quantum algorithms in QROM with that in the quantum bit-fixing (QBF) model. Then we analyze the maximum winning probability in the QBF model by connecting it to another game in the QROM. 

We first recall the QBF model. The definition is adapted from \cite{liu2023non}.
\begin{definition}[Games in the $\varrho$-BF-QROM]\label{def:P-BF-QROM}
Games in the $\varrho$-BF-QROM are similar to games in the standard QROM, except now the random function $f$ has a different distribution. 
\begin{itemize}
    \item Before a game starts, a $\varrho$-query quantum algorithm $\cB$ (having no input) is picked and fixed by an adversary.
    \item \textbf{Rejection Sampling Stage.} A random oracle $f$ is picked uniformly at random conditioned on $\cB^f$ outputs 0. In other words, the distribution of $f$ is defined by a rejection sampling:
    \begin{enumerate}
        \item $f\sim\mu$, where $\mu$ is the uniform distribution over $\{f:[KM]\to[N]\}$.
        \item Run $\cB^f$ and output a binary outcome $b$ together with a quantum state $\ket{\sigma}$.
        \item Restart from step 1 if $b\neq 0$.
    \end{enumerate}
    \item \textbf{Online Stage.} The game is then executed with oracle access to $f$, and an (online) algorithm $\cA$ starts with the state $\ket{\sigma}$.
\end{itemize}
\end{definition}
\noindent Furthermore, they showed that we only need to care in the case when $\Pr[\mathbf{W}]$ is not very small (at least $1/N^S$), where  $\mathbf{W}$ is defined as an event for $b = 0$. Otherwise, then the RHS in the inequality \cite[Lemma 6.5]{liu2023non} is already bounded by $1/N^S < \delta^S$, for any $\delta$ that we are interested in the theorem statement. 

We directly apply the result in \cite{liu2023non} to go from non-uniform quantum algorithms to algorithms in $\varrho$-BF-QROM.
\begin{theorem}[{\cite[Theorem 6.1]{liu2023non}}]\label{thm:qadvice-6.1}
Let $G$ be any game with $T_{\text{Samp}}, T_{\text{Verify}}$ 
being the number of queries made by $\text{Samp}$ and $\text{Verify}$ defined in \cite[Definition 4.3]{liu2023non}. For any $S,T$, let $\varrho=S\cdot(T+T_{\text{Verify}}+T_{\text{Samp}})$.

If $G$ has security $\mathcal{V}(\varrho,T)$ in the $\varrho$-BF-QROM, then it has maximum winning probability $\epsQ_{G}(S,T)\le 2\cdot\mathcal{V}(\varrho,T)$ against $(S,T)$ non-uniform quantum algorithms with quantum advice.
\end{theorem}

In the game $G_{R, K}$, we have $T_{\text{Samp}} = 0$ (since there is no challenge) and $T_{\text{Verify}} = c$ (recall $c$ is the size of interested tuples). Since we assume $T \geq c$ in the theorem statement, $\varrho = S\cdot(T+c) \leq 2 S T$. 
\Cref{thm:qadvice-6.1} gives us a way to upper bound $\epsQ_G(S,T)$ by upper bounding $\mathcal{V}(2ST,T)$. 

Thus, we are interested in the probability, 
\begin{align}\label{eq:BF-w/E}
    \mathcal{V}(2ST,T)=\Pr_{\substack{f\sim\mu\\ k\sim[K]}}\left[\cA^f(k)\text{ wins } \middle| \bW\right],
\end{align}
where $k\sim[K]$ is a uniformly sampled salt, for any $\Pr[\bW]\ge 1/N^S$. This corresponds to the probability that an algorithm first runs $\cB$ by making $2ST$ queries, and conditioned that $\cB^f$ outputs $0$, it receives a uniformly sampled salt $k$ and runs $\cA^f(k)$. We can use the compressed oracle technique to simulate the random oracle due to their perfect indistinguishability, where we will analyze the joint state of the algorithm's registers and the database/oracle registers. 
\end{proof}

\subsection{Proof of Claim \ref{clm:upbound-W-E}}

\begin{proof}[Proof of \Cref{clm:upbound-W-E}] 
This is by the following inequality: for every event $\bA$, $\bB$ and $\bC$ that occur with probability $>0$, $\Pr[\bB\land\bC]>0$,
\begin{align*}
    \Pr[\bA|\bB]&=\frac{\Pr[\bA\land\bB]}{\Pr[\bB]}=\frac{\Pr[\bA\land\bB\land\bC]}{\Pr[\bB]}+\frac{\Pr[\bA\land\bB\land\overline{\bC}]}{\Pr[\bB]}\\
    &\le \frac{\Pr[\bA\land\bB\land\bC]}{\Pr[\bB\land\bC]}+\frac{\Pr[\bB\land\overline{\bC}]}{\Pr[\bB]}\\
    &\le \Pr[\bA | \bB\land \bC]+\frac{\Pr[\overline{\bC}]}{\Pr[\bB]}.
\end{align*}
Letting $\bA$ be the event that $\cA^f(k)$ wins, $\bB=\bW$, and $\bC=\bE$, we can obtain \Cref{clm:upbound-W-E}.
\end{proof}

\subsection{Proof of Claim \ref{clm:conditioned_lem5}}

\begin{proof}[Proof of \Cref{clm:conditioned_lem5}]
Note that the conditioning on $\bW\land \bE$ only projects the algorithm's initial state $\ket{\sigma}$ to the space where $\bW\land\bE$ holds, and this does not affect the correctness of the analysis in \cite[Lemma 5]{zhandry2019record}. That is, for $c=1$, we can still write the final state for $\cA^f(k)$ conditioned on $\bW\land\bE$ in the form of 
\begin{align*}
    \sum_{x,u,z,D}\alpha_{x,u,z,D,0}\ket{x,u,z,D} +\sum_{r\neq0^n}\alpha_{x,u,z,D,r}\frac{1}{\sqrt{N}}\sum_{y}\omega_N^{y\cdot r}\ket{x,u,z,D\cup(x,y)},
\end{align*}
with the compressed oracle.
The generalization for $c\ge 1$ in \cite[Lemma 5]{zhandry2019record} also holds.
Therefore,
\begin{align*}
\sqrt{\Pr_{\substack{f\sim\mu\\ k\sim[K]}}[\cA^f(k)\text{ wins } | \bW\land \bE]}
    &\le \sqrt{\Pr_{\substack{f\sim\mu\\ k\sim[K]}}[(D|_k\in P) | \bW\land \bE]}+\sqrt{\frac{c}{N}}\\
    &\le \sqrt{
    2\cdot\left(\Pr_{\substack{f\sim\mu\\ k\sim[K]}}[(D|_k\in P) | \bW\land \bE]+\frac{c}{N}\right)
    }
    \tag{since $\sqrt a+\sqrt b\le\sqrt{2a+2b}$}
\end{align*}
as desired.
\end{proof}

\subsection{Proof of Claim \ref{clm:conditioned_propertyfinding}}

\begin{proof}[Proof of \Cref{clm:conditioned_propertyfinding}]
For an algorithm that first succeeds on $\bW$ and $\bE$, and then succeeds on salt $k$ after $T$ queries, we can denote its final state as:
\begin{align*}
    \ket{\phi}=\Gamma_{k}^P\cdot \compressO \cdot U_T\cdots \compressO \cdot U_2\cdot \compressO \cdot U_1\ket{\sigma}. 
\end{align*}
Here $\Gamma_{k}^P$ is the projection on the database register that $D|_k\in P$. 
$\ket{\sigma}$ is the state after projecting to the space where $\cB^f$ outputs $0$ and $\bE$ happens, and thus every superposition of $\ket{\sigma}$ succeeds on at most $2S\log N$ salts.
Here the algorithm $\cA^f(k)$ starts from $\ket{\sigma}$ and makes $T$ oracle queries.

With $\ket{\phi}$, we can have
\begin{align*}
    \Pr_{\substack{f\sim\mu\\ k\sim[K]}}[(D|_k \in P)| \bW\land \bE]=\mathbb{E}_{k\sim[K]}\left[\|\ket{\phi}\|^2\right],
\end{align*}
and now our goal is to upper bound the expected norm of $\ket{\phi}$.

We define $\Gamma_{k}^{>T}$ as the projection on the database register satisfying that the size of $D|_k$ is larger than $T$ and that $D|_k\in \bar{P}$.
Define $\overline{\Gamma}=\I-\Gamma_k^P-\Gamma_k^{>T}$, which is the projection onto the remaining cases.
Now we can write $\ket{\phi}$ as the following:
\begin{align*}
    \ket{\phi}&=\Gamma_{k}^P\cdot \compressO \cdot U_T\cdots \compressO \cdot U_1\ket{\sigma}\\
    &=\Gamma_{k}^P\cdot \compressO \cdot U_T\cdots \compressO \cdot U_1\cdot (\Gamma_k^P+\Gamma_k^{>T}+\overline{\Gamma})\ket{\sigma}\\
    &=\ket{\phi_1}+\ket{\phi_2}+\ket{\phi_3}.
\end{align*}
Here $\ket{\phi_1}$ is the part of $\ket{\phi}$ that starts from $\Gamma_k^P\ket{\sigma}$, $\ket{\phi_2}$ is the one that starts from $\Gamma_{k}^{>T}\ket{\sigma}$, and $\ket{\phi_3}$ is the one that starts from $\overline{\Gamma}\ket{\sigma}$.

Then we can upper bound expected $\|\ket{\phi}\|^2$ by upper bounding expected $\|\ket{\phi_i}\|^2$, $i=1,2,3$:
\begin{align*}
    \mathbb{E}_{k\sim[K]}\left[\|\ket{\phi}\|^2\right]
    \le\mathbb{E}_{k\sim[K]} \left[3\cdot(\|\ket{\phi_1}\|^2+\|\ket{\phi_2}\|^2+\|\ket{\phi_3}\|^2)\right].
\end{align*}

\paragraph*{For $\ket{\phi_1}$.} 
Note that every superposition of $\ket{\sigma}$ succeeds on at most $2S\log N$ salts, 
\begin{align}\label{ineq:cond-propertyfind-1}
    \mathbb{E}_{k\sim[K]}\left[\|\ket{\phi_1}\|^2\right]
    &=\mathbb{E}_{k\sim[K]}\left[\|\Gamma_{k}^P\cdot \compressO \cdot U_T\cdots \compressO \cdot U_1\cdot \Gamma_k^P\ket{\sigma}\|^2\right] \notag\\
    &\le \mathbb{E}_{k\sim[K]}\left[\|\Gamma_k^P\ket{\sigma}\|^2\right]
    =\frac{1}{K} \cdot\sum_{k\in[K]}\|\Gamma_k^P\ket{\sigma}\|^2 \notag\\
    &\le\frac{2S\log N}{K} = \widetilde{O}\left(\frac{S}{K}\right).
\end{align}

\paragraph*{For $\ket{\phi_2}$.}  
Note that $\ket{\sigma}$ is obtained by $\cB^f$ making $2ST$ queries. Thus for every superposition in $\ket{\sigma}$, the database $D$ can have at most $2ST/(T+1)=O(S)$ salts $i$ such that the size of $D|_i$ is larger than $T$. Therefore, 
\begin{align}\label{ineq:cond-propertyfind-2}
    \mathbb{E}_{k\sim[K]}\left[\|\ket{\phi_2}\|^2\right]
    &=\mathbb{E}_{k\sim[K]}\left[\|\Gamma_{k}^P\cdot \compressO \cdot U_T\cdots \compressO \cdot U_1\cdot \Gamma_k^{>T}\ket{\sigma}\|^2\right]
    \notag\\
    &\le \mathbb{E}_{k\sim[K]}\left[\|\Gamma_k^{>T}\ket{\sigma}\|^2\right]
    =\frac{1}{K} \cdot\sum_{k\in[K]}\|\Gamma_k^{>T}\ket{\sigma}\|^2\notag\\
    &\le \frac{2S}{K} = O\left(\frac{S}{K}\right).
\end{align}

\paragraph*{For $\ket{\phi_3}$.}
$\ket{\phi_3}$ starts from $\overline{\Gamma}\ket{\sigma}$, where $\overline{\Gamma}$ means that at that time it does not succeed on salt $k$, and it does not contain more than $T$ elements of salt $k$ in the database. 
For every fixed $k$, we define a sequence of properties $P_t$ ($j=T,T+1,\ldots,2T$) such that the complement of each property $\neg P_t$ is defined as $P\cup \text{SZ}_{\le t}$ (similar to \cite[Remark 5.8]{chung2021compressed}), where $\text{SZ}_{\le t}$ corresponds to the property that the size of the restricted database $D|_k$ is no larger than $t$. 
Here the upper bound $2T$ of $t$ comes from the fact that $\overline\Gamma\ket\sigma$ starts with at most $T$ elements and the algorithm makes at most $T$ extra queries.

Now we apply \cite[Lemma 5.6]{chung2021compressed} and the maximum winning probability is bounded by
\begin{align*}
    [\![\neg P_T\stackrel{T}{\to} P_{2T}]\!]\le \sum_{t=T+1}^{2T}\ [\![\neg P_{t-1}\to P_{t}]\!],
\end{align*}
where $[\![P\to P']\!]$ corresponds to the maximum transition capacity (quantum analogue of transition probability) to go from property $P$ to $P'$ after one query, and $[\![P\stackrel{T}{\to} P']\!]$ is the case for $T$ queries.

We use the following fact about an upper bound for the transition capacity $[\![\neg P_t\to P_{t+1}]\!]$. The fact is restated and proved in \Cref{prf:transition-g-h}.
\begin{fact}\label{fct:cO-transition}
Let $P$ be a monotone property with transition probability $p_t$. For any joint state $\ket{\psi}$ of the algorithm's registers and the oracle registers, define $\Gamma_{k}^{\le t}$ as the projection that projects every superposition $\ket{i,x,u,z}\ket{D}$ to the states such that $i=k$ and the size of the database $D|_k$ is at most $t$. Define $\Gamma_{k}^P$ as the projection that projects to $D|_k\in P$. Then we can have the following result for the decrease of norm during the transition from $D|_k\in \bar{P}$ to $D|_k\in P$:
\begin{align*}
    \left\|\Gamma_k^P\cdot \cphso\cdot (\I-\Gamma_k^P)\Gamma_k^{\le t}\ket{\psi}\right\|
    \le \sqrt{8\cdot p_t}\left\|(\I-\Gamma_k^P)\Gamma_k^{\le t}\ket{\psi}\right\|.
\end{align*}
\end{fact}

This implies $[\![\neg P_t\to P_{t+1}]\!]\le \sqrt{8p_t}$, and gives us the following upper bound of the maximum winning probability for every $k\in[K]$:
\begin{align}\label{ineq:cond-propertyfind-3}
    O(\sqrt{p_T}+\sqrt{p_{T+1}}+\cdots+\sqrt{p_{2T-1}})^2\le O(\gamma_{2T}),
\end{align}
where we used the definition of $\gamma_T$ from \Cref{thm:quantum_lift_property_find}.

Therefore, from \Cref{ineq:cond-propertyfind-1}, \Cref{ineq:cond-propertyfind-2}, and \Cref{ineq:cond-propertyfind-3},
\begin{align*}
    \mathbb{E}_{k\sim[K]}\left[\|\ket{\phi}\|^2\right]
    \le 3\cdot\mathbb{E}_{k\sim[K]} \left[\|\ket{\phi_1}\|^2+\|\ket{\phi_2}\|^2+\|\ket{\phi_3}\|^2\right]
    \le O(\gamma_{2T})+\widetilde{O}\left(\frac{S}{K}\right).
    \tag*{\qedhere}
\end{align*}
\end{proof}

\subsection{Proof of Claim \ref{clm:threshold-qsdpt-E}}

\begin{proof}[Proof of \Cref{clm:threshold-qsdpt-E}]
Without loss of generality, we assume $\gamma_T=o(1)$.
By letting $B=2ST$ and $k=2S\log N$ in \Cref{thm:threshold-qsdpt}, we can upper bound $\Pr_{f\sim\mu}[\overline{\bE}]\le (O(1)\cdot \gamma_{{T}/{\log N}})^{2S\log N}\le (1/2)^{2S\log N}=(1/N)^{2S}$.
As $\Pr_{f\sim\mu}[\bW]\ge 1/N^S$, we now obtain 
\begin{equation*}
\Pr_{f\sim\mu}[\overline{\bE}]\le \frac{\Pr_{f\sim\mu}[\bW]}{N^S} \le \frac{\Pr_{f\sim\mu}[\bW]}N.
\tag*{\qedhere}
\end{equation*}
\end{proof}

\subsection{Proof of Theorem \ref{thm:main_quantum}}

\begin{proof}[Proof of \Cref{thm:main_quantum}]
We put together the claims in \Cref{subsec:prf-main_quantum} and conclude \Cref{thm:main_quantum}:
\begin{align*}
    \Pr_{\substack{f\sim\mu\\ k\sim[K]}}[\cA^f(k)\text{ wins } | \bW]
    &\le\Pr_{\substack{f\sim\mu\\ k\sim[K]}}[\cA^f(k)\text{ wins } | \bW\land \bE]+\frac{\Pr_{f\sim\mu}[\overline{\bE}]}{\Pr_{f\sim\mu}[\bW]} 
    \tag{by \Cref{clm:upbound-W-E}}\\
    &\le 2\cdot\Pr_{\substack{f\sim\mu\\ k\sim[K]}}[(D|_k\in P)| \bW\land \bE]+\frac{2c}{N} + \frac{\Pr_{f\sim\mu}[\overline{\bE}]}{\Pr_{f\sim\mu}[\bW]} 
    \tag{by \Cref{clm:conditioned_lem5}}\\
    &\le O(\gamma_{2T})+\widetilde{O}\left(\frac{S}{K}\right)+\frac{2c}{N}+\frac{\Pr_{f\sim\mu}[\overline{\bE}]}{\Pr_{f\sim\mu}[\bW]} 
    \tag{by \Cref{clm:conditioned_propertyfinding}}\\
    &\le O(\gamma_{2T})+\widetilde{O}\left(\frac{S}{K}\right)+O\left(\frac{c}{N}\right) 
    \tag{by \Cref{clm:threshold-qsdpt-E}}.
\end{align*}
According to \Cref{eq:BF-w/E}, this gives an upper bound for $\mathcal{V}(2ST,T)$. Therefore, since $\gamma_{T}$ is polynomial in $T$, by \Cref{thm:qadvice-6.1} we can obtain
\begin{align*}
    \epsQ_{G_{R,K}}(S, T) \leq \widetilde{O}\left( \gamma_T + \frac{S}{K} + \frac{c}{N}\right).
    \tag*{\qedhere}
\end{align*}
\end{proof}

\section{Proof of SDPT for Quantum Property Finding}\label{sec:qsdpt}

In this section, we present the proof of our strong direct product theorem \Cref{thm:qsdpt} for quantum property finding problem. 
This proof can also capture a threshold version (\Cref{thm:threshold-qsdpt}) as a corollary.

We first characterize the winning probability we need to compute for quantum direct product algorithms. Then we present how to use the idea of ``path integral'' to give an upper bound of this probability: split the final state into paths, and then give an upper bound of the maximum norm for all paths.

\paragraph*{The Winning Probability of a Quantum Direct Product Algorithm.} 
For quantum direct product of property finding problems, we model the result state of a quantum algorithm that makes $kT$ queries to the oracle using the compressed oracle technique.
We extend the algorithm's first register from $\ket{x}$ to $\ket{k,x}$ with salt register $\ket{k}$, and the salt space here is $[\kappa]$. After $\kappa T$ queries, we will apply the binary-value measurement $\{\Lambda_{P^{\times \kappa}}, \I-\Lambda_{P^{\times \kappa}}\}$ on the oracle register to check in superposition whether the database is in $P^{\times \kappa}$: 
$$ 
\Lambda_{P^{\times \kappa}}=\sum_{D\in P^{\times \kappa}}\ket{D}\bra{D}.
$$
The algorithm wins if the measurement outcome of $\{\Lambda_{P^{\times \kappa}}, \I-\Lambda_{P^{\times \kappa}}\}$ is $0$. Since the final local computation $U_{\sf final}$ after the $\kappa T$-th query is independent of the oracle register, we can ignore that operator, and the winning probability for this algorithm can be written as 
\begin{align*}
    \Pr[\cA \text{ wins }P^{\times \kappa}]=
    \left\|
    (\I_{\cA}\otimes \Lambda_{P^{\times \kappa}})\cdot \compressO \cdot U_{\kappa T} \cdot \compressO \cdot U_{\kappa T-1} \cdots \compressO \cdot U_1 \ket{\psi_0}
    \right\|^2.
\end{align*}

We say that a state (or a projector) succeeds on salt $k\in [\kappa]$ if the state (or every state in the image of the projector) is in the image of $\sum_{D|_k\in P}\ket D\bra D$. 

Here we fix this property $P$. We further define $\Lambda^{=r}$ to be the projection onto the states that can succeed on exactly $r$ salts, and define $\Lambda^{\ge r}$ to be the projection that succeeds on at least $r$ salts, i.e., $\Lambda^{\ge r'}=\sum_{r=r'}^\kappa\Lambda^{=r}$ for salt space $[\kappa]$.
We also define $\Lambda_{k}$ to be the projection onto the states that can succeed on salt $k$.
Therefore, for salt space $[\kappa]$, $\Lambda^{\ge \kappa} =\Lambda^{=
\kappa} =\prod_{k\in[\kappa]}\Lambda_{k}=\I_{\cA}\otimes \Lambda_{P^{\times \kappa}}$. 

We define the final winning state after the measurement as $\ket{\psi_{\sf win}}$,
\begin{align}\label{def:win-state}
    \ket{\psi_{\sf win}}=\Lambda^{=\kappa} \cdot \compressO \cdot U_{\kappa T} \cdot \compressO \cdot U_{\kappa T-1} \cdots \compressO \cdot U_1 \ket{\psi_0},
\end{align}
and the winning probability of a $\kappa T$-query algorithm can be written as 
\begin{align}\label{prob:qsdpt-qalgo}
    \Pr[\cA \text{ wins }P^{\times \kappa}]=\|\ket{\psi_{\sf win}}\|^2.
\end{align}

To prove SDPT for quantum property finding, our goal is to give an upper bound for \Cref{prob:qsdpt-qalgo}.

\subsection{Splitting the Final Winning State} 
To give an upper bound for \Cref{prob:qsdpt-qalgo}, we will first split the final state $\ket{\psi_{\sf win}}$ into a sum of several ``path'' terms according to the ``measurement results'' after each oracle query.

We define $\cP_\kappa$ to be the set of all permutations of salts $[\kappa]$. From $\ket{\psi_{\sf win}}$ as in \Cref{def:win-state}, we define a set of $\{\ket{\psi_{\vz, \pi}}\}$ based on the result after each oracle query. Here $\vz=(z_1,z_2,\ldots,z_\kappa)$ satisfies $1\le z_1<z_2<\cdots<z_\kappa\le \kappa T$ and $\pi=(\pi_1, \pi_2,\ldots, \pi_\kappa)\in\cP_\kappa$ is a permutation all salts $[\kappa]$. 
$\ket{\psi_{\vz, \pi}}$ is defined as follows:
\begin{align}\label{def:zpi}
    \ket{\psi_{\vz, \pi}} =\left(\cE_{\vz, \pi}^{(\kappa T)}\cdot \compressO \cdot U_{\kappa T}\right)\cdot\left(\cE_{\vz, \pi}^{(\kappa T-1)}\cdot \compressO \cdot U_{\kappa T-1}\right)\cdots\left(\cE_{\vz, \pi}^{(1)}\cdot \compressO \cdot U_1\right)\ket{\psi_0},
\end{align}
where $\cE_{\vz,\pi}^{(t)}$ is a projection on the database register and it will be applied after the $t$-th oracle query.
More precisely, $\cE_{\vz,\pi}^{(t)}$ is a product of projections on each $D|_k$:
\begin{align*}
    &\cE_{\vz, \pi}^{(t)}=\cE_{\vz, \pi, 1}^{(t)}\cdot \cE_{\vz, \pi, 2}^{(t)}\cdots \cE_{\vz, \pi, \kappa}^{(t)},\\
    &\text{where }\ 
    \cE_{\vz, \pi, k}^{(t)} =\begin{cases}
        \Lambda_k, & \text{if } t\ge z_r, \text{ for $r\in[\kappa]$ being the location of $k$ in $\pi$, i.e., $\pi_r=k$}, \\
        \I - \Lambda_k, & \text{if } t=z_r-1, \text{ for $r$ as defined above}, \\
        \I, & \text{if } t<z_r-1, \text{ for $r$ as defined above}.
    \end{cases}
\end{align*}

Here we give an example for $\kappa=2$, $T=3$, with $\vz=(2,5)$ and $\pi=(\pi_1,\pi_2)=(1,2)$:
\begin{alignat*}{3}
    \ket{\psi_{\vz,\pi}}=\ 
    \left( \Lambda_{1}\Lambda_{2}\cdot\compressO\cdot U_6\right)&\ 
    \cdot\ 
    \left(
    \Lambda_{1}\Lambda_{2} \cdot \compressO \cdot U_{5}\right)\ 
    \cdot &\ 
    \left(
    \Lambda_{1}(\I-\Lambda_{2}) \cdot \compressO \cdot U_{4}\right)&\ \\
    \cdot\ 
    \left(
    \Lambda_{1}\cdot \compressO\cdot U_{3}\right)&\  \cdot\ \quad\  
    \left(\Lambda_{1} \cdot\compressO\cdot U_{2}\right)\ 
    \cdot &\ 
    \left(
    (\I-\Lambda_{1})\cdot\compressO\cdot U_{1}\right)&\ 
    \ket{\psi_{0}}
\end{alignat*}

Intuitively, by fixing $\vz$ and $\pi$, we obtain $\ket{\psi_{\vz, \pi}}$ as a part of $\ket{\psi_{\sf win}}$, where we fix the $z_r$-th query as the latest time when the state moves from non-success to success on salt $k=\pi_r$, and always succeeds on salt $\pi_r$ afterwards, for every $r\in[\kappa]$.  

This idea comes from a way of partitioning the winning event in the classical case. For classical direct product of property finding problem, the algorithm has a database $D$ corresponding to the partial assignments it learns. It starts from an empty database, and appends at most one element $(x,y)$ to $D$ after one oracle query. It wins if its final database succeeds on every salt $k\in [\kappa]$ for property $P$. 
\begin{itemize}
    \item Following the evolution of the algorithm's database when we actually run an algorithm, we can record the number of queries $z_r$ when it succeeds on $r$ salts, and at that time which new salt $\pi_r\in[\kappa]$ it succeeds on. Note that the algorithm's database will never forget the elements it learns, $\pi=(\pi_1,\pi_2,\ldots,\pi_\kappa)$ will form a permutation of $[\kappa]$, and $z_1<z_2<\cdots<z_\kappa$ since the algorithm can only learn one element after one oracle query. These give us definitions of $\vz=(z_1,z_2,\ldots,z_\kappa)\in Z$ and $\pi=(\pi_1,\pi_2,\ldots,\pi_\kappa)\in\cP_\kappa$.
    \item The probability of winning the property finding game can be written as a sum of the winning probability with respect to every $(\vz,\pi)$. In light of the symmetry of different salts, we can sum over all $\pi$ and only consider the contribution with respect to $\vz$. This corresponds to the probability of ``first time succeeding on $r$ salts after $z_r$-th query for each $r\in[\kappa]$''.
\end{itemize}

When it comes to the quantum case, querying an element that is already in the database might lead to ``forgetting'', which means that each superposition of the database might not always append one element after one oracle query. This might result in the case that there are multiple times when the database moves from non-success to success on some salt. Note that the final winning state is in the image of $\Lambda^{=\kappa}=\prod_{k\in[\kappa]}\Lambda_k$. Therefore we can define a term with respect to $(\vz, \pi)$, where we think of the $z_r$-th query as the latest time when it moves from non-success to success on salt $\pi_r$ after this oracle query, and always succeeds on salt $\pi_r$ afterwards. This corresponds to the definition in \Cref{def:zpi}.

Similarly to the classical case, we do not distinguish between permutations and define a term $\ket{\psi_{\vz}}$ with respect to every $\vz$ as ``a path'' by considering all possible $\pi$'s.

\begin{definition}[Paths in $\ket{\psi_{\sf win}}$]\label{def:splitting-path}
Let $Z$ to be the set of all $\vz=(z_1,z_2,\ldots,z_\kappa)$ such that $1\le z_1< z_2<\cdots< z_\kappa\le \kappa T$. 
Recall \Cref{def:zpi}. We define $\ket{\psi_{\vz}}$ to be the path corresponding to $\vz=(z_1,z_2,\ldots,z_\kappa)$:
\begin{align*}
    \ket{\psi_{\vz}} =\sum_{\pi\in\cP_\kappa} \ket{\psi_{\vz,\pi}}.
\end{align*}
\end{definition}

Here we claim that all the paths form a perfect partition of $\ket{\psi_{\sf win}}$ as formalized in \Cref{clm:splitting}.
The proof of \Cref{clm:splitting} can be found in \Cref{app:prf_splitiing}.
\begin{claim}\label{clm:splitting}
With the definition of $\vz$, $Z$, and $\ket{\psi_{\vz}}$ in \Cref{def:splitting-path}, we have
\begin{align*}
    \ket{\psi_{\sf win}}=\sum_{\vz\in Z} \ket{\psi_{\vz}}.
\end{align*}
\end{claim}

\subsection{Maximum Norm of a Path}
With the path splitting, \Cref{clm:splitting} gives us a way to upper bound the norm of $\ket{\psi_{\sf win}}$ by considering the maximum norm of each path. This follows from the triangle inequality:
\begin{align*}
    \left\|\ket{\psi_{\sf win}}\right\|
    \le \sum_{\vz\in Z}\left\|\ket{\psi_{\vz}}\right\|\le |Z|\cdot \max_{\vz\in Z}\left\|\ket{\psi_{\vz}}\right\|.
\end{align*}

For each path, we can give an upper bound of its maximum norm.

\begin{claim}\label{clm:qsdpt-maxpath}
Recall that $p_t$ is the transition probability defined in \Cref{def:transition-prob}. 
We have
    \begin{align*}%
    \max_{\vz\in Z}\left\|\ket{\psi_{\vz}}\right\|^2
    \le 8^\kappa\cdot \max_{\substack{t_1 + \cdots + t_\kappa = \kappa T \\ t_1, \ldots, t_\kappa \in \mathbb{Z}_{+}}}  
    \left\{\prod_{i=1}^\kappa p_{t_i-1}\right\}.
    \end{align*}
\end{claim}

To prove \Cref{clm:qsdpt-maxpath}, we analyze the evolution of $\ket{\psi_{\vz}}$. 
A natural idea is to upper bound the norm of each $\ket{\psi_{\vz,\pi}}$ for all $\pi\in\cP_\kappa$ and make use of some orthogonality relation between different $\pi$'s. 
However, the issue is that we cannot trivially obtain even pairwise orthogonality in $\{\ket{\psi_{\vz,\pi}}\}_{\pi\in\cP_\kappa}$ by looking at the support on computational basis. An example is that for $3$-query final states $\ket{\psi_{\vz,(1,2,3)}}$ and $\ket{\psi_{\vz,(2,1,3)}}$, they can share the same database register, which means they are not orthogonal.
Therefore, we need to make detailed analysis following the evolution of $\ket{\psi_{\vz,\pi}}$ to explore more subtle orthogonality relation during the process of reaching the final state.

Now we fix some $\vz=(z_1,z_2,\ldots,z_\kappa)$ and analyze this path. 

We define $\Lambda_{\cS}$ to be the projection that succeeds on salts in $\cS$, $\Lambda_{\cS}=\prod_{k\in\cS}\Lambda_k$. We also define $\tau_k$ as a projection on the query register that projects to queries on salt $k$. That is, $\tau_k\ket{k,x,u,z,D}=\ket{k,x,u,z,D}$ and $\tau_k\ket{k',x,u,z,D}=0$ for $k'\neq k$. 

For every $\pi\in\cP_\kappa$, we first rewrite $\ket{\psi_{\vz,\pi}}$. Note that $\cE_{\vz,\pi}^{(t)}$ only applies on the database register and $U_t$ only applies on the query register (and the algorithm's local memory). Therefore $\cE_{\vz,\pi}^{(t)}$ and $U_{t'}$ commutes, and from \Cref{def:zpi} we can rewrite $\ket{\psi_{\vz,\pi}}$ as
\begin{align}
    \ket{\psi_{\vz,\pi}} &=\left(\cE_{\vz, \pi}^{(\kappa T)}\cdot \compressO \cdot U_{\kappa T}\right)\cdot\left(\cE_{\vz, \pi}^{(\kappa T-1)}\cdot \compressO \cdot U_{\kappa T-1}\right)\cdots\left(\cE_{\vz, \pi}^{(1)}\cdot \compressO \cdot U_1\right)\ket{\psi_0}
    \tag{defined in \Cref{def:zpi}}\\
    \ &=\zeta_{\vz,\pi}^{(\kappa T)}\cdot \zeta_{\vz,\pi}^{(\kappa T-1)}\cdots\zeta_{\vz,\pi}^{(2)}\cdot \zeta_{\vz,\pi}^{(1)}\ket{\psi_0},\label{def:zpi_redefine}\\
    \text{where } &\ 
    \zeta_{\vz,\pi}^{(t)}=
    \begin{cases}
        \compressO\cdot U_t, & \text{if } t<z_1,\\
        \Lambda_{\{\pi_1,\ldots,\pi_r\}}\cdot \compressO \cdot U_t, & \text{if } z_r<t<z_{r+1}\text{ for some }r\in[\kappa-1],\\
        \Lambda_{\{\pi_1,\ldots,\pi_r\}}\cdot \compressO\cdot (\I-\Lambda_{\pi_r})\cdot\tau_{\pi_r}\cdot U_t, & \text{if } t=z_r\text{ for some }r\in[\kappa].
    \end{cases}\notag
\end{align}
This is by applying the following for every $r\in[\kappa]$ of the $(z_r-1)$-th query, moving all $(\I-\Lambda_{\pi_r})$ factor from $\cE_{\vz,\pi}^{(z_r-1)}$ to a projector of the $z_r$-th query,
\begin{align*}
    \left(\cE_{\vz,\pi}^{(z_r)}\cdot \compressO\cdot U_{z_r}\right)\cdot \left(\cE_{\vz,\pi}^{(z_r-1)}\cdot \compressO\cdot U_{z_r-1}\right)
    &=\left(\cE_{\vz,\pi}^{(z_r)}\cdot \compressO\cdot U_{z_r}\right)\cdot \left((\I-\Lambda_{\pi_r})\Lambda_{\{\pi_1,\ldots,\pi_{r-1}\}}\cdot \compressO\cdot U_{z_r-1}\right) \\
    &= \left(\cE_{\vz,\pi}^{(z_r)}\cdot \compressO\cdot(\I-\Lambda_{\pi_r})\cdot U_{z_r}\right)\cdot \left(\Lambda_{\{\pi_1,\ldots,\pi_{r-1}\}}\cdot \compressO\cdot U_{z_r-1}\right),
\end{align*}
and this gives us the definition of $\zeta_{\vz,\pi}^{(t)}$.
As for every superposition of the queried salt, one oracle query can change the database only restricted to that salt, we know $\Lambda_k\cdot \compressO\cdot (\I-\Lambda_k)\tau_{k'}=0$ for $k'\neq k$. 
For $t=z_r$, we can write $\Lambda_{\{\pi_1,\ldots,\pi_r\}}\cdot \compressO\cdot (\I-\Lambda_{\pi_r})\cdot U_t=\Lambda_{\{\pi_1,\ldots,\pi_r\}}\cdot \compressO\cdot (\I-\Lambda_{\pi_r})\cdot\tau_{\pi_r}\cdot U_t$.

Inside the definition of $\zeta_{\vz,\pi}^{(z_r)}$ is a transition pattern. 
Before the $z_r$-th oracle query, the database register of the state is projected to $(\I-\Lambda_{\pi_r})$, i.e., $D|_{\pi_r}\in\bar{P}$ and $\bar P$ is the complement of the property $P$ (from \Cref{def:transition-prob}). After the oracle query $\compressO$, it is projected on $\Lambda_{\pi_r}$, where $D|_{\pi_r}\in P$. 
If we know the size $u$ of $D|_{\pi_r}$, it will fit with the notion of the transition probability $p_u$ (recall \Cref{def:transition-prob}), and there might be a decrease on the norm with respect to this transition probability, from the intuition of ``lazy sampling'' of the compressed oracle.

We follow from this intuition and analyze the evolution of the process.  We define $\ket{\psi_{\vz,\pi}^t}$ from \Cref{def:zpi_redefine} as applying only the first $t$ operators $\zeta_{\vz,\pi}^{(\theta)}$ for $\theta\in[t]$. Note that for $z_r\le t<z_{r+1}$, these operators will only depend on the first $r$ terms $(\pi_1,\ldots,\pi_r)$ of $\pi$, and thus for $\pi,\pi'$ that share the same $(\pi_1,\ldots,\pi_r)$, they share the same $\ket{\psi_{\vz,\pi}^t}$. 
Let $\cP_r(\kappa)$ be the set of all permutations of sets of size $r$ chosen from salt space $[\kappa]$.
In this way, we can redefine $\ket{\psi_{\vz,\iota}^t}$ ($z_r\le t<z_{r+1}$) for $\iota\in \cP_r(\kappa)$: 
\begin{align}\label{def:zpit}
    \ket{\psi_{\vz,\iota}^t}= \zeta_{\vz,\iota}^{(t)}\cdot \zeta_{\vz,\iota}^{(t-1)}\cdots \zeta_{\vz,\iota}^{(2)}\cdot \zeta_{\vz,\iota}^{(1)}
    \ket{\psi_0}.
\end{align}
Here $\zeta_{\vz,\iota}^{(\theta)}$ is still well-defined, as $\zeta_{\vz,\pi}^{(\theta)}$ only makes use of the first $r$ elements of $\pi$, and now $\iota$ is of length $r$ as a truncated permutation. 
By adding all possible $\iota\in\cP_r(\kappa)$, we can similarly define $\ket{\psi_{\vz}^t}$:
\begin{align*}
    \ket{\psi_{\vz}^t} =\sum_{\iota\in\cP_r(\kappa)} \ket{\psi_{\vz,\iota}^t}.
\end{align*}
In addition, we define $\ket{\psi_{\vz,\cS}^t}$ for $\cS\subset[\kappa]$ and $|\cS|=r$ as the sum of $\ket{\psi_{\vz,\iota}^t}$ over all permutations $\iota$ of $\cS$.

We will need the following notion of used elements as a proxy to analyze the progress the algorithm makes towards finding the property. It intuitively quantifies the number of queries for each winning salt, where an additional $t-|D|$ factor comes from the quantum power of forgetting.

\begin{definition}[Number of Used Elements]
\label{def:used_elements} 
For state $\ket{\psi_{\vz}^t}$ ($z_{r}\le t<z_{r+1}$ for some $r\in[\kappa -1]$), we define the number of used elements for the database register of each superposition of $\ket{\psi_{\vz, \iota}^t}$ as 
$$t-|D|+\sum_{k\in \{\iota_1,...,\iota_r\}}|D_k|,$$
where $D_k=D|_k$ is the restriction of $D$ to salt $k$. 
\end{definition}

That is, for each state $\ket{\psi_{\vz, \cS}^t}$ (which equals $\sum \ket{\psi_{\vz,\iota}^t}$ that sums over every permutation $\iota$ of $\cS$), all the used elements in a database include ``empty'' entries and elements with salt index in $\cS=\{\iota_1,\ldots,\iota_r\}$. Here the number of ``empty'' entries correspond to the number of forgotten entries; that is, a database is ``supposed'' to have $t$ entries after $t$ queries, but it might have $|D|<t$ entries due to the ability of forgetting in QROM, and the number of the ``empty'' entries are $t-|D|$ in this case.

\begin{definition}\label{def:G&g}
We define $Q_{B}$ as a projection on the database register onto the database states that contain $B$ used elements, and $Q_{\le B}=\sum_{j=0}^{B}Q_j$.

For $r\in[\kappa], j, B\in\mathbb{N}, \cS\subseteq[\kappa]$, and $|\cS|=r$, we define
\begin{align*}
    g_{r,\cS}^j=\|Q_j\ket{\psi_{\vz, \cS}^{z_r}}\|^2
\end{align*}
as the square of the norm of $\ket{\psi_{\vz,\cS}^{z_r}}$ with exactly $j$ used elements in the database.
\end{definition}

\begin{remark} Here the definition of $g_{r,\cS}^j$ depends on $\vz$. As we are fixing $\vz$ throughout this section, $\vz$ does not show up in the notation of $g_{r,\cS}^j$ for simplicity.
\end{remark}

By defining $z_0=0$, we extend the definition to $g_{0,\cS}^B$ (for $\cS=\emptyset$). Thus $\sum_{j=0}^B g_{0,\cS}^j=g_{0,\cS}^0=1$.

\begin{definition}
\label{def:P_l^B}
For $\ell\in[\kappa]$, $B\in\mathbb{N}$, and the transition probability $p_t$ from \Cref{def:transition-prob}, we define  
\begin{align*}
    P^B_{\ell}=\max_{\substack{t_1 + \cdots + t_\ell = B \\ t_1, \ldots, t_\ell \in \mathbb{Z}_{+}}}  \left\{ \prod_{i=1}^\ell p_{t_i-1}\right\}.
\end{align*}
\end{definition}

With the above definitions, we have the following claim of the evolution of the norm. 
\begin{claim}\label{clm:qsdpt-G_P_g} 
With the definition of $g_{r,\cS}^j$ and $P^B_{\ell}$, for $r\in[\kappa]$ and any $\ell\in[r]$, for summation over all $\cS,\cS'\subseteq[\kappa]$ such that $|\cS|=r$ and $|\cS'|=r-\ell$,
\begin{align*}
    \sum_{\cS}\sum_{j\le B}g_{r,\cS}^j\le 8^{\ell}\cdot \sum_{\cS'}\sum_{j=0}^B P^{B-j}_{\ell}\cdot g_{r-\ell, \cS'}^j.
\end{align*}

As a corollary, by letting $\ell=r$, we have $\sum_\cS\sum_{j\le B}g_{r,\cS}^j\le 8^r\cdot P_{r}^B$.
\end{claim}

The proof of \Cref{clm:qsdpt-G_P_g} can be found in \Cref{app:qsdpt-G_P_g}. 
Now we can prove \Cref{clm:qsdpt-maxpath} given \Cref{clm:qsdpt-G_P_g}. 

\begin{proof} [Proof of \Cref{clm:qsdpt-maxpath}]
Let $r=\kappa$, $\ell=\kappa$, and $B=\kappa T$. Note that for the direct product case, salt space is $[\kappa]$, and $|\cS|=\kappa$ means that $\cS=[\kappa]$. Thus 
for all $\vz\in Z$, we can have
\begin{align*}
    \|\ket{\psi_{\vz}}\|^2 &\le \|\ket{\psi_{\vz}^{z_\kappa}}\|^2=\sum_{j=0}^{\kappa T}\|Q_j\ket{\psi_{\vz,\cS}^{z_\kappa}}\|^2 =\sum_{j=0}^{\kappa T}g_{r,\cS}^j
    \tag{by orthogonality on $j$}\\
    &\le 8^\kappa\cdot P_{\kappa}^{\kappa T}
    \tag{by \Cref{clm:qsdpt-G_P_g}}
    \\
    &=8^\kappa\cdot \max_{\substack{t_1 + \cdots + t_\kappa = \kappa T \\ t_1, \ldots, t_\kappa \in \mathbb{Z}_{+}}}  
    \left\{\prod_{i=1}^\kappa p_{t_i-1}\right\}.
    \tag*{\qedhere}
\end{align*}
\end{proof}

\begin{corollary*}[\Cref{thm:qsdpt} Restated]\label{cor:qsdpt}
With every path bounded by \Cref{clm:qsdpt-maxpath}, if the transition probability $p_t$ is polynomial in $t$,\footnote{Recall that this means $p_t=\Theta(t^{c_1}/N^{c_2})$ for some constants $c_1,c_2\ge 0$.} then with $\gamma_T=\Theta(T^2\cdot p_T)$, we have a strong direct product for quantum property finding:
\begin{align*}
    \Pr[\cA \text{ wins }P^{\times \kappa}]
    \le O(1)^\kappa\cdot \gamma_T^\kappa.
\end{align*}
\end{corollary*}
\begin{proof} %
The size of set $Z$ is $|Z|=\binom{\kappa T}{\kappa}\le(e\cdot T)^\kappa$, and thus
\begin{align*}
    \Pr[\cA \text{ wins }P^{\times \kappa}]&=\|\ket{\psi_{\sf win}}\|^2
    \le |Z|^2\cdot \max_{\vz\in Z}\|\ket{\psi_\vz}\|^2\\
    &\le e^{2\kappa}\cdot T^{2\kappa} \cdot \max_{\substack{t_1 + \cdots + t_\kappa = \kappa T \\ t_1, \ldots, t_\kappa \in \mathbb{Z}_{+}}} 
    \left\{\prod_{i=1}^\kappa p_{t_i-1}\right\}.
\end{align*} 
Recall that $p_t=\Theta(t^{c_1}/N^{c_2})$ is polynomial in $t$ and thus $\gamma_t=\Theta(t^{2+c_1}/N^{c_2})$.
Therefore
\begin{align*}
    \Pr[\cA \text{ wins }P^{\times \kappa}]
    &\le {O(1)}^\kappa\cdot T^{2\kappa} \cdot \max_{\substack{t_1 + \cdots + t_\kappa = \kappa T \\ t_1, \ldots, t_\kappa \in \mathbb{Z}_{+}}} 
    \left\{\prod_{i=1}^\kappa p_{t_i-1}\right\}\\
    &= {O(1)}^\kappa\cdot T^{2\kappa} \cdot \max_{\substack{t_1 + \cdots + t_\kappa = \kappa T \\ t_1, \ldots, t_\kappa \in \mathbb{Z}_{+}}} 
    \left\{\prod_{i=1}^\kappa\frac{t_i^{c_1}}{N^{c_2}}\right\}\\
    &\le O(1)^\kappa\cdot T^{2\kappa}\cdot (T^{c_1}/N^{c_2})^\kappa\\
    &= O(1)^\kappa\cdot \gamma_T^\kappa
\end{align*}
as desired.
\end{proof}

\subsection{Threshold SDPT}

For the threshold version, we consider the case of finding properties for at least $\kappa$ salts among all $K\ge \kappa$ salts, with $B\in\mathbb{Z}_{+}$ queries to the oracle.
To this end, we extend \Cref{def:win-state}
\begin{align*}
    \ket{\psi_{\sf win}}=\Lambda^{\ge \kappa}\cdot \compressO\cdot U_{B}\cdot \compressO\cdot U_{B-1}\cdots \compressO\cdot U_1\ket{\psi_0}.
\end{align*}
By orthogonality, we project $\ket{\psi_{\sf win}}$ onto different spaces where the database register succeeds on exactly $\kappa\le r\le K$ salts as $\Lambda^{\ge \kappa}=\sum_{r=\kappa}^{K}\Lambda^{=r}$, and 
\begin{align*}
    \|\ket{\psi_{\sf win}}\|^2=\sum_{r=\kappa}^K \|\Lambda^{=r}\ket{\psi_{\sf win}}\|^2.
\end{align*}

For every $r$, we apply \Cref{clm:splitting} on $\Lambda^{=r}\ket{\psi_{\sf win}}$ and obtain the triangle inequality:
\begin{align*}
    \|\Lambda^{=r}\ket{\psi_{\sf win}}\|
    \le \sum_{\vz\in Z}\|\Lambda^{=r}\ket{\psi_{\vz}}\|
    \le |Z|\cdot \max_{\vz\in Z}\|\Lambda^{=r}\ket{\psi_{\vz}}\|.
\end{align*}
\begin{remark}
    Here $Z$ is defined to be the set of all $\vz=(z_1,\ldots,z_r)$ such that $1\le z_1<z_2<\cdots<z_r\le B$. The definition of $Z$ depends on both the number of winning salts $r$ and the number of all queries $B$.
\end{remark}

We also have a different version of \Cref{clm:qsdpt-maxpath}:
\begin{claim}\label{clm:tqsdt-maxpath}
    With $\vz$, $Z$, and $\{\Lambda^{=r}\ket{\psi_{\vz}}\}_{\vz\in Z}$ defined for this case, we have
    \begin{align*}
    \max_{\vz\in Z}\left\|\Lambda^{=r}\ket{\psi_{\vz}}\right\|^2
    \le 8^r\cdot \max_{\substack{t_1 + \cdots + t_r = B \\ t_1, \ldots, t_r \in \mathbb{Z}_{+}}}  
    \left\{\prod_{i=1}^r p_{t_i-1}\right\},
    \end{align*}
where $p_t$ is the transition probability defined in \Cref{def:transition-prob}.
\end{claim}
\begin{proof}
For every $\vz\in Z$, we split the final state $\Lambda^{=r}\ket{\psi_{\vz}}$ according to the set of salts the database register succeeds on. That is, for all $\cS\subseteq[K]$ such that $|\cS|=r$, 
\begin{align*}
    \|\Lambda^{=r} \ket{\psi_{\vz}}\|^2
    &= \sum_{\cS}\|\Lambda^{=r} \ket{\psi_{\vz,\cS}}\|^2
    \tag{by orthogonality}\\
    &\le \sum_{\cS}\| \ket{\psi_{\vz,\cS}^{z_r}}\|^2
    =\sum_{\cS}\sum_{j=0}^{B}\|Q_j \ket{\psi_{\vz,\cS}^{z_r}}\|^2=\sum_{\cS}\sum_{j=0}^{B}g_{r,\cS}^j\\
    &\le 8^r\cdot \max_{\substack{t_1 + \cdots + t_r = B \\ t_1, \ldots, t_r \in \mathbb{Z}_{+}}}  
    \left\{\prod_{i=1}^r p_{t_i-1}\right\},
\end{align*}
where the last inequality follows from \Cref{clm:qsdpt-G_P_g}.
\end{proof}

We now show how to obtain the threshold SDPT \Cref{thm:threshold-qsdpt}.

\begin{corollary*}[\Cref{thm:threshold-qsdpt} Restated]
\Cref{thm:qsdpt} also holds for the threshold version. That is, for salt space $[K]$ and an algorithm $\cA$ that makes $B\in\mathbb{Z}_{+}$ queries, if $p_t=\Theta(t^{c_1}/N^{c_2})$ for some constants $c_1,c_2\ge 0$, then with $\gamma_t=\Theta(t^2\cdot p_t)$, the probability of succeeding on more than $\kappa\le K$ salts will have the following upper bound:
\begin{align*}
    \Pr[\cA\text{ wins }P^{(\ge \kappa)}]\le O(1)^\kappa\cdot (\gamma_{B/\kappa})^\kappa,
\end{align*}
where we recall $\gamma_{B/r}$ from \Cref{rmk:rounding?} when $B/r$ is fractional.
\end{corollary*}
\begin{proof}
We can assume without loss of generality that 
\begin{equation}\label{eq:B/k5.19}
B/\kappa\le O\left(N^{c_2/(c_1+2)}\right).
\end{equation}
This is because if otherwise 
$$
\gamma_{B/\kappa}=\Theta\left(\frac{(B/\kappa)^{c_1+2}}{N^{c_2}}\right)\ge1
$$
and hence the statement holds trivially.

Now we proceed to proving \Cref{thm:threshold-qsdpt} given \Cref{eq:B/k5.19}.
Similar to \Cref{cor:qsdpt},\footnote{Actually the analysis of \Cref{cor:qsdpt} only holds for $r\le B$. In the case of $r>B$, we have $\|\Lambda^{r}\ket{\psi_{\sf win}}\|=0$, which is trivially bounded by $O(1)^r\cdot (\gamma_{B/r})^r$. 
This is because now the database register cannot contain $r$ elements, and thus no adversary can win the property finding game on $r$ salts (recall \Cref{def:qproperty_finding} for quantum property finding).} we first have
\begin{align}
    \|\Lambda^{=r}\ket{\psi_{\sf win}}\|^2 
    &\le |Z|^2\cdot \max_{\vz\in Z} \|\Lambda^{=r}\ket{\psi_{\vz}}\|^2\le O(1)^r\cdot (\gamma_{B/r})^r
    \label{eq:B/ksimilar}\\
    &\le\Theta\left(\frac{(B/r)^{c_1+2}}{N^{c_2}}\right)^r
    \notag\\
    &\le2^{-r}.
    \tag{by \Cref{eq:B/k5.19}}
\end{align}
Then we sum over different $r$ for $\|\ket{\psi_{\sf win}}\|^2$:
\begin{align*}
    \|\ket{\psi_{\sf win}}\|^2 
    &= \sum_{r=\kappa}^K\|\Lambda^{=r}\ket{\psi_{\sf win}}\|^2\\
    &\le 2\cdot\|\Lambda^{=\kappa}\ket{\psi_{\sf win}}\|^2
    \tag{since $\|\Lambda^{=r}\ket{\psi_{\sf win}}\|^2\le2^{-r}$ by above}\\
    &\le O(1)^\kappa\cdot(\gamma_{B/\kappa})^\kappa
    \tag{by \Cref{eq:B/ksimilar}}
\end{align*}
as desired.
\end{proof}

\section{\texorpdfstring{Missing Proofs in \Cref{sec:qsdpt}}{Missing Proofs for Quantum Property Finding Direct Product}}

\subsection{Proof of Claim \ref{clm:splitting}}\label{app:prf_splitiing}

\begin{proof}[Proof of \Cref{clm:splitting}]
For $\ket{\psi_{\sf win}}$ defined in \Cref{def:win-state} and $\vz$, $Z$, $\ket{\psi_{\vz}}$ defined in \Cref{def:splitting-path}, we generalize $Z$ to be the set of all $\vz=(z_1,\ldots,z_\kappa)$ such that $1\le z_1<\cdots<z_\kappa\le B$, and we write $\ket{\psi_{\sf win}}$ for $\kappa$ salts and $B$ queries as $\ket{\psi_{\sf win}^{\kappa,B}}$. We also generalize the salt space to $[K]$, while $\ket{\psi_{\sf win}^{\kappa,B}}$ is redefined as the final state where every superposition of database succeeds on \emph{exactly} $\kappa$ salts. We define $\ket{\psi_{\vz}^{\kappa,B}}$ as in \Cref{def:splitting-path} correspondingly. 

Now we prove that \Cref{clm:splitting} holds for all $\kappa,B\in\mathbb{Z}_{+}$ with $\kappa\le \min\{B,K\}$. In fact, we prove the following: for every $\kappa,B\in\mathbb{Z}_{+}$ with $\kappa\le \min\{B,K\}$, and $Z$ as the set of all $(z_1,\ldots,z_\kappa)$ with $1\le z_1<\cdots <z_\kappa\le B$, we have
\begin{align*}
    \ket{\psi_{\sf win}^{\kappa,B}}= \sum_{\vz\in Z}\ket{\psi_{\vz}^{\kappa,B}}.
\end{align*}
From this we can immediately obtain \Cref{clm:splitting} by letting $\kappa=K$ and $B=\kappa T$.

We additionally define $\ket{\psi_{\sf win}^{\cS,B}}$ for $\cS\subset[K]$, $|\cS|=\kappa$ to be the part of $\ket{\psi_{\sf win}^{\kappa,B}}$ that succeeds on salts \emph{exactly} from $\cS$. Thus $\ket{\psi_{\sf win}^{\kappa,B}}=\sum_{\cS}\ket{\psi_{\sf win}^{\cS,B}}$ for summation over $\cS\subset[K]$ such that $|\cS|=\kappa$. We only need to prove that for every $\cS\subset [K]$ and $|\cS|=\kappa$, $\ket{\psi_{\sf win}^{\cS,B}}$ also has this kind of splitting $\{\ket{\psi_{\vz}^{\cS,B}}\}_{\vz\in Z}$. 
We prove it by induction on $(\kappa,B)$. 

\paragraph*{Base case of $(1,1)$.} We start from $(\kappa,B)=(1,1)$. Without loss of generality, we can let $\cS=\{1\}$. By definition, $\ket{\psi_{\sf win}^{\cS,1}}=\Lambda_{1}\cdot\compressO\cdot U_1\ket{\psi_0}$. 
As $1\le z_1\le B$, there is only one possible path with $\vz=(z_1)$, $z_1=1$, and $Z=\{\vz\}$. %
Thus $\ket{\psi_{\sf win}^{\cS,1}}=\sum_{\vz\in Z}\ket{\psi_{\vz}^{\cS,1}}$ for $(\kappa,B)=(1,1)$.

\paragraph*{From $(1,B)$ to $(1,B+1)$.} 
Now we assume that it holds for $(1,B)$ and prove the case of $(1,B+1)$. Without loss of generality, we can assume $\cS=[1]$. Now $Z=\{(t)\}_{t=1}^{B+1}$. Here we denote the state just after the $t$-th query as $\ket{\psi^t}$, and here $\ket{\psi_{\sf win}^{[1],B+1}}=\Lambda_1\ket{\psi^{B+1}}$. When it comes to the path, we have the following:
\begin{align*}
    \ket{\psi_{\sf win}^{[1], B+1}} &= \Lambda_{1}\cdot\compressO\cdot U_{B+1}(\Lambda_{1}+\I-\Lambda_{1})\ket{\psi^B}\\ 
    &= \Lambda_{1}\cdot\compressO\cdot U_{B+1}\ket{\psi_{\sf win}^{[1],B}}+ (\Lambda_{1}\cdot\compressO\cdot U_{B+1}) \cdot (\I-\Lambda_{1})\ket{\psi^{B}} \\
    &=\left(\sum_{\vz'=(t), t\le B}\Lambda_{1}\cdot\compressO\cdot U_{B+1}\ket{\psi_{\vz'}^{[1],B}}\right)+\ket{\psi_{(B+1)}^{[1],B+1}}
    \tag{by induction hypothesis}\\
    &=\sum_{\vz\in Z}\ket{\psi_{\vz}^{[1],B+1}}.
\end{align*}
Therefore we conclude that the case of $(1,B)$ holds for all $B\in\mathbb{Z}_{+}$.

\paragraph*{From $(\kappa-1, \kappa-1)$ to $(\kappa,\kappa)$.} 
Now we assume that the formula holds for $(\kappa-1,\kappa-1)$ and prove the case of $(\kappa,\kappa)$. For $(\kappa,\kappa)$ and $\cS\subset[K]$ of size $\kappa$, without loss of generality, we can let $\cS=[\kappa]$. Now the set of path has only one element defined by $\vz=(1,2,\ldots,\kappa)$, $Z=\{\vz\}$, as we require $z_1<z_2<\cdots <z_\kappa$. When it comes to the path, we have the following:
\begin{align*}
    \ket{\psi_{\sf win}^{[\kappa],\kappa}}
    &=\Lambda^{=\kappa}\cdot \compressO\cdot U_{\kappa}(\Lambda^{=\kappa}+\I-\Lambda^{=\kappa})\ket{\psi^{\kappa-1}}\\
    &=\Lambda^{=\kappa}\cdot \compressO\cdot U_\kappa\ket{\psi_{\sf win}^{[\kappa],\kappa-1}}+\sum_{k\in [\kappa]}\Lambda^{=\kappa}\cdot \compressO\cdot U_\kappa \ket{\psi_{\sf win}^{[\kappa]\backslash\{k\}, \kappa-1}}\\
    &=\sum_{k\in [\kappa]}\Lambda^{=\kappa}\cdot \compressO\cdot U_\kappa \ket{\psi_{\sf win}^{[\kappa]\backslash\{k\}, \kappa-1}}\\
    &=\sum_{k\in [\kappa]}\Lambda^{=\kappa}\cdot\compressO\cdot U_\kappa\ket{\psi_{\vz'}^{[\kappa]\backslash\{k\}, \kappa-1}}=\ket{\psi_{\vz}^{[\kappa],\kappa}}
    \tag{$\vz'=(1,\ldots,\kappa-1)$ by induction hypothesis}\\ 
    &=\sum_{\vz\in Z}\ket{\psi_{\vz}^{[\kappa],\kappa}}.
\end{align*}
For the second line uses $\Lambda^{=\kappa}\cdot \compressO\cdot U_{\kappa}(\I-\Lambda^{=\kappa})\ket{\psi^{\kappa-1}}=\sum_{k\in [\kappa]}\Lambda^{=\kappa}\cdot \compressO\cdot U_\kappa \ket{\psi_{\sf win}^{[\kappa]\backslash\{k\}, \kappa-1}}$ as we must win exactly $\kappa-1$ salts before the final query to win all of them.
The third line is by the fact that we will need at least $\kappa$ queries to satisfy the property for each salt in $[\kappa]$, and thus the first term in the second line is $0$.

Therefore we can conclude that the case of $(\kappa,\kappa)$ holds for all $\kappa\le K$.

\paragraph*{From $(\kappa-1,B), (\kappa,B)$ to $(\kappa,B+1)$.} Now we assume that it holds for any $(\kappa,b)$ that $2\le \kappa\le K$ and all $b\le B$, for some $K,B\in\mathbb{Z}_{+}$. We prove the case of $(\kappa,B+1)$ for all $\kappa\le K$. Without loss of generality, we let $\cS=[\kappa]$.

In this case, $1\le z_1<\cdots<z_\kappa\le B+1$, we start with the value of $z_\kappa$.
If $z_\kappa\le B$, we can work with $Z_0=\{\vz:1\le z_1<\cdots <z_\kappa\le B\}$, and it is handled by the case of $(\kappa,B)$.
If $z_\kappa=B+1$, we can work with $Z_1=\{\vz:1\le z_1<\cdots <z_{\kappa-1}\le B\}$, and it is handled by the case of $(\kappa-1,B)$. Note that $Z=\{\vz:\vz=\vz_0\text{ for $\vz_0\in Z_0$ or }\vz=\vz_1\cup\{B+1\}\text{ for $\vz_1\in Z_1$}\}$. 
Thus $Z$ is the disjoint union of $Z_0,Z_1'$, where $Z_1'=\{\vz\cup\{B+1\}:\vz\in Z_1\}$.

When it comes to the paths, we can decompose $\ket{\psi_{\sf win}^{[\kappa],B+1}}$ as follows:
\begin{align*}
    \ket{\psi_{\sf win}^{[\kappa],B+1}}&=\Lambda^{=\kappa}\cdot\compressO\cdot U_{B+1}(\Lambda^{=\kappa}+\I-\Lambda^{=\kappa})\ket{\psi^B}\\
    &=\Lambda^{=\kappa}\cdot\compressO\cdot U_{B+1}\ket{\psi_{\sf win}^{[\kappa],B}} + \sum_{k\in[\kappa]}\Lambda^{=\kappa}\cdot\compressO\cdot U_{B+1}\ket{\psi_{\sf win}^{[\kappa]\backslash\{k\}, B}}\\
    &=\sum_{\vz\in Z_0}\Lambda^{=\kappa}\cdot\compressO\cdot U_{B+1}\ket{\psi_{\vz}^{[\kappa],B}} + \sum_{k\in[\kappa]}\sum_{\vz'\in Z_1}\Lambda^{=\kappa}\cdot\compressO\cdot U_{B+1}\ket{\psi_{\vz'}^{[\kappa]\backslash\{k\},B}}
    \\
    &=\sum_{\vz\in Z_0\cup Z_1'}\ket{\psi_{\vz}^{[\kappa],B+1}}.
\end{align*}
The second line is by the fact that we can only succeed on one more salt after one query, and thus $\ket{\psi^B}$ should succeed on at least $\kappa-1$ salts. The third line is by the assumptions of $(\kappa,B)$ and $(\kappa-1,B)$.

\paragraph*{Conclusion.} By induction, we can conclude that the splitting of paths holds for all $(\kappa,B)$, $\kappa\le\min\{B,K\}$ as desired.
\end{proof}

\subsection{Proof of Claim \ref{clm:qsdpt-G_P_g}}\label{app:qsdpt-G_P_g}

We first present the following lemma.
\begin{lemma}\label{ineq:abel}
Assume $a_1, \ldots, a_n$ and $b_1, \ldots, b_n$ satisfy $\sum_{i=1}^{k} a_i \le \sum_{i=1}^k b_i$ for all $k\in [n]$. Then for every non-increasing $c_1\ge c_2\ge \cdots \ge c_n\ge 0$, we have 
\begin{align*}
    \sum_{i=1}^k c_i\cdot a_i \le \sum_{i=1}^k c_i \cdot b_i.
\end{align*}
\end{lemma}
\begin{proof}
We denote $A_k=\sum_{i=1}^k a_i$ and $B_k=\sum_{i=1}^k b_i$. Thus we have $A_k\le B_k$ for all $k\in [n]$ and
\begin{align*}
    \sum_{i=1}^k c_i\cdot a_i = c_k \cdot A_k + \sum_{i=1}^{k-1}(c_i-c_{i+1})\cdot A_i 
     \le c_k\cdot B_k + \sum_{i=1}^{k-1}(c_i-c_{i+1})\cdot B_i
    =\sum_{i=1}^k c_i\cdot b_i.
    \tag*{\qedhere}
\end{align*}
\end{proof}

Now we present the proof for \Cref{clm:qsdpt-G_P_g}.
\begin{proof}[Proof of \Cref{clm:qsdpt-G_P_g}]
We extend \Cref{def:P_l^B} to $\ell=0$ and define $P_0^B=1$ for all $B\in\mathbb{N}$. 
The case of $r=0$ is true by definition, as both sides are identical to $\sum_{j\le B}g_{0,\emptyset}^j=g_{0,\emptyset}^0=1$.

For any $r\in\mathbb{Z}_{+}$, we prove the statement by induction on $\ell\in [r]\cup\{0\}$. 

\paragraph*{Base case $\ell=0$.}
In this case, since $P_{0}^{B-j}=1$, we have
\begin{align*}
    \sum_{\cS}\sum_{j=0}^B g_{r,\cS}^j= \sum_S\sum_{j=0}^B P_0^{B-j}\cdot g_{r,\cS}^j.
\end{align*}

\paragraph*{From $\ell$ to $\ell+1$.}
Now we assume that it holds for any value $\le \ell$ and prove the case of $\ell+1$. To go from $\ell$ to $\ell+1$, we will at least need to analyze how much the norm is reduced when the database just succeeds on one more salt. Here we need to build some relations between $\ket{\psi_{\vz,\cS}^{z_{r-\ell}}}$ and $\ket{\psi_{\vz,\cS'}^{z_{r-\ell-1}}}$. We consider the case for $\ket{\psi_{\vz,\cS}^{z_r}}$ and $\ket{\psi_{\vz,\cS'}^{z_{r-1}}}$. We remark that here the size of $\cS$ automatically corresponds to the index $r$ of $z_r$ in $\ket{\psi_{\vz,\cS}^{z_r}}$. That is, $|\cS|=r$ and $|\cS'|=r-1$.

We start from $\ket{\psi_{\vz,\cS}^{z_{r}}}$. 
By definition, $g_{r,\cS}^j=\|Q_j\ket{\psi_{\vz,\cS}^{z_r}}\|^2$. 
To analyze the transition pattern, we can further decompose $g_{r,\cS}^B$ into a sum of $g_{r,\cS,k}^{j,\nu}$ to indicate the number of elements $\nu$ on that queried salt $k$ in the database. 
Informally here $j$ corresponds to the number of used elements that are not of the $z_r$-th queried salt $k$. 
Here we define $Q^\nu$ as the projection that projects every superposition $\ket{k,x,u,z}\ket{D}$ to the states such that the size of the database $D|_k$ (restricted on that salt $k$) is $\nu$. 
Recall that $\tau_k$ is the projection on the query register that projects to queries on salt $k$ and $Q_{j+\nu}$ projects to states that the database contain $j+\nu$ used elements (see \Cref{def:G&g}).
Then we define $g_{r,\cS,k}^{j,\nu}$ ($k\in\cS$) as the following:
\begin{align}
    g_{r,\cS,k}^{j,\nu}
    &=\|\tau_k\cdot Q_{j+\nu}Q^\nu\ket{\psi_{\vz,\cS}^{z_r}}\|^2\notag \\
    &=\|\tau_k\cdot Q_{j+\nu}Q^\nu\cdot\Lambda_{\cS}\cdot\compressO(\I-\Lambda_k) U_{z_r}\ket{\psi_{\vz,\cS\backslash\{k\}}^{z_r-1}}\|^2
    \tag{by the definition of $\ket{\psi_{\vz,\cS}^{z_r}}$} \\
    &=\|Q_{j+\nu} Q^\nu\Lambda_{\cS}\cdot\compressO\cdot (\I-\Lambda_k) \tau_k \cdot U_{z_r}\ket{\psi_{\vz, \cS\backslash\{k\}}^{z_r-1}}\|^2.
    \label{def:grsijt}
\end{align}
For the last equality above, we note that $Q_{j+\nu}Q^\nu\Lambda_{\cS}$ is still a projection on the database register, it commutes with $\tau_k$; and $\compressO$ commutes with $\tau_k$.
For fixed $\cS$ and by orthogonality, we have
\begin{align*}
    g_{r,\cS}^{B}
    =\sum_{k\in \cS}\sum_{j=0}^{B} g_{r,\cS,k}^{j,B-j}.
\end{align*}

Now we look at $(\I-\Lambda_k)\tau_k \cdot U_{z_r}\ket{\psi_{\vz, \cS\backslash\{k\}}^{z_r-1}}$, which is the state just before the query that succeeds on a new salt $k$. We want to analyze the norm of that state. 

Let $\tQ_{j+\nu}^\nu$ be the projection onto the the state that every superposition can contribute to falling in $Q_{j+\nu}Q^\nu$ after one oracle query $\compressO$. That is, $Q_{j+\nu}Q^\nu\cdot\compressO\cdot(\I-\tQ_{j+\nu}^\nu)\ket{\psi}=0$ for any state $\ket{\psi}$. Then, for any $\cS'\subset[\kappa]$ such that $|\cS'|=r-1$, we define
\begin{align}
    h_{r,\cS',k}^{j,\nu} = \|
    \tQ_{j+\nu}^\nu\cdot \Lambda_{\cS'}(\I-\Lambda_k)\cdot \tau_k \cdot U_{z_r}\ket{\psi_{\vz, \cS'}^{z_r-1}}
    \|^2.\label{def:hrsijt}
\end{align}
It corresponds to the norm of the state just before the $z_r$-th query, such that each of its superposition queries on salt $k$ and can contribute to falling in $g_{r,\cS'\cup\{k\},k}^{j,\nu}$. 

\begin{remark}
In fact, for every superposition $\ket{k,x,u,z}\ket{D}$ that lies in the image of $\tQ_{j+\nu}^\nu$ in \Cref{def:hrsijt}, it must belong to either one of these two cases (we can also define $\tQ_{j+\nu}^\nu$ in this way):
\begin{itemize}
    \item[(1)] $(k,x)\notin D$, the size of $D|_k$ is $\nu-1$, and the number of used elements (with respect to $\cS'$, as in \Cref{def:hrsijt}) in $D$ is $j$. 
    
    This corresponds to the case when the $z_r$-th query is on an entry $(k,x)$ that is not in the database $(k,x)\notin D$. The $z_r$-th query will append one element to the database, adding the size of $D|_k$ by 1 and the number of used elements by $\nu$. 
    \item[(2)] $(k,x)\in D$, the size of $D|_k$ is $\nu$, and the number of used elements (with respect to $\cS'$) is $j-1$. 
    
    This corresponds to the case when the $z_r$-th query is on an entry $(k,x)$ that is already in the database $(k,x)\in D$. After the $z_r$-th query, as the state moves from $(\I-\Lambda_k)$ to $\Lambda_k$, $|D|$ will remain the same, the size of $D|_k$ remaining $\nu$ and the number of used elements increased by $\nu+1$. 
\end{itemize}
We can also obtain orthogonality for $\tQ_{j+\nu}^\nu$ for different $j,\nu$ (with the same $\cS'$) from the above property, by looking at their support. An immediate result is $\|\sum_{k\notin \cS}\sum_{\nu\in\mathbb{N}}\sum_{j\in\mathbb{N}}\tau_k\cdot \tQ_{j+\nu}^\nu\ket{\psi_{\vz,\cS}}\|\le \|\ket{\psi_{\vz,\cS}}\|$ for any state $\ket{\psi_{\vz,\cS}}$ with used elements with respect to $\cS$.
\end{remark}

We have the following two facts about $g_{r,\cS,k}^{j,\nu}, h_{r,\cS',k}^{j,\nu}, g_{r-1,\cS}^j$, which gives us a way to go from the $z_r$-th query back to the $z_{r-1}$-th.

\begin{fact}\label{fct:transition-g-h}
With the above definition, by the compressed oracle technique,
\begin{align*}
    g_{r,\cS,k}^{j,\nu}\le 8\cdot p_{\nu-1}\cdot h_{r,\cS\backslash\{k\},k}^{j,\nu}.
\end{align*}
\end{fact}

\Cref{fct:transition-g-h} formalizes the way to go from $z_r$ to $z_r-1$. Its proof follows from the analysis of the compressed oracle and is left in \Cref{prf:transition-g-h}. 

\begin{fact}\label{fct:transition-h-g}
For $r\in[\kappa]$, $B\in\mathbb{N}$, and $\cS'\subset[\kappa]$ such that $|\cS'|=r-1$,
\begin{align*}
    \sum_{j=0}^B\sum_{\nu\in\mathbb{N}} \sum_{k\notin \cS'} h_{r,\cS',k}^{j,\nu}
    \le \sum_{j=0}^B g_{r-1,\cS'}^{j}.
\end{align*}
\end{fact}

\Cref{fct:transition-h-g} gives us a way to go from $z_r-1$ back to $z_{r-1}$. The intuition is the following: for every superposition $\ket{k,x,u,z}\ket{D}$ in $\ket{\psi_{r-1,\cS'}^{z_{r-1}}}$, the number of the used elements in (each superposition of) $D$, which is with respect to the set $\cS'$, will not decrease during the evolution from the $z_{r-1}$-th query to the $(z_r-1)$-th query. This is because during the evolution for this path, we apply $U_t$, $\compressO$, and $\Lambda_{\cS'}$ on the state just after the $z_{r-1}$-th query. $U_t$ performs local computation and it does not change the database register; $\Lambda_{\cS'}$ will be a projection on the computational basis and it does not reduce the number of used elements in a database. $\compressO$ performs an oracle query, and it can only add, change, or delete an element. While adding and changing do not reduce the number of used elements, deleting does not reduce either, as empty entry is also considered as a used element. Therefore, for any superposition that is counted in $g_{r-1,\cS'}^j$, the number of used elements in the database register will be at least $j$ after it applies the $(z_{r}-1)$-th query, and it will ``fall into'' $\{h_{r,\cS',k}^{j',\nu}\}_{j'\ge j,\ \nu\in\mathbb{N}, k\notin \cS'}$.

\begin{proof}[Proof of \Cref{fct:transition-h-g}] Following from the non-decreasing property of the number of used elements with respect to $\cS'$ ($\cS'\subset[\kappa], |\cS'|=r-1$), we can obtain $\|Q_{\le B}\ket{\psi_{\vz,\cS'}^{z_{r}-1}}\|\le \|Q_{\le B}\ket{\psi_{\vz,\cS'}^{z_{r-1}}}\|$, and thus 
\begin{align*}
    \sum_{j=0}^B\sum_{\nu\in\mathbb{N}} \sum_{k\notin \cS'} h_{r,\cS',k}^{j,\nu}
    &=\sum_{j=0}^B\sum_{\nu\in\mathbb{N}} \sum_{k\notin \cS'}\|
    \tQ_{j+\nu}^\nu\cdot (\I-\Lambda_k)\tau_k \cdot U_{z_r}\ket{\psi_{\vz, \cS'}^{z_r-1}}
    \|^2\\
    &\le \sum_{j=0}^B\sum_{\nu\in\mathbb{N}} \sum_{k\notin \cS'}\|
    \tQ_{j+\nu}^\nu\cdot \tau_k \cdot U_{z_r}\ket{\psi_{\vz, \cS'}^{z_r-1}}
    \|^2
    = \sum_{j=0}^B\|\sum_{\nu\in\mathbb{N}} \sum_{k\notin \cS'}\tQ_{j+\nu}^\nu\cdot \tau_k \cdot U_{z_r}\ket{\psi_{\vz, \cS'}^{z_r-1}}\|^2\\
    &\le \sum_{j=0}^B\|Q_j\cdot U_{z_r}\ket{\psi_{\vz,\cS'}^{z_r-1}}\|^2
    = \|Q_{\le B}\ket{\psi_{\vz,\cS'}^{z_r-1}}\|^2\\
    &\le \|Q_{\le B}\ket{\psi_{\vz,\cS'}^{z_{r-1}}}\|^2=\sum_{j=0}^B g_{r-1,\cS'}^j.\tag*{\qedhere}
\end{align*}
\end{proof}

With \Cref{fct:transition-g-h} and \Cref{fct:transition-h-g}, we have the following:
\begin{align*}
    \sum_{\cS}\sum_{j=0}^B P_{\ell}^{B-j}\cdot g_{r-\ell,\cS}^j
    &= \sum_{\cS}\sum_{j=0}^B  P_{\ell}^{B-j}\cdot \left(\sum_{\nu=0}^j\sum_{k\in \cS} g_{r-\ell,\cS,k}^{\nu, j-\nu}\right)\\
    &\le \sum_{\cS}\sum_{j=0}^B  P_{\ell}^{B-j}\cdot \left(\sum_{\nu=0}^j\sum_{k\in \cS} 8\cdot p_{j-\nu-1}\cdot h_{r-\ell, \cS\backslash\{k\},k}^{\nu, j-\nu}\right)
    \tag{by \Cref{fct:transition-g-h}} \\
    &\le 8\cdot \sum_{\cS'}\sum_{k\notin \cS'}\sum_{j=0}^B \sum_{\nu=0}^j P_{\ell+1}^{B-\nu}\cdot h_{r-\ell, \cS', k}^{\nu, j-\nu}
    \tag{by \Cref{def:P_l^B}} \\
    &= 8\cdot \sum_{\cS'}\sum_{\nu=0}^B P_{\ell+1}^{B-\nu}\cdot\left(\sum_{k\notin \cS'}\sum_{d=0}^{B-\nu} h_{r-\ell, \cS', k}^{\nu, d}\right).
    \tag{by $d=j-\nu$}
\end{align*}
Here $|\cS|=r-\ell$ and $|\cS'|=r-\ell-1$.

Note that $P_{\ell+1}^{B-\nu}$ is non-increasing on $\nu$ by the monotonicity of the transition probability $p_t$, and \Cref{fct:transition-h-g} gives the inequality of two sums of arrays. Thus following \Cref{ineq:abel}, for every $\cS'$ ($|\cS'|=r-\ell-1$), we can continue our calculation as:
\begin{align*}
   \sum_{\nu=0}^B P_{\ell+1}^{B-\nu}\cdot\left(\sum_{k\notin \cS'}\sum_{d=0}^{B-\nu} h_{r-\ell, \cS', k}^{\nu, d}\right)
    &\le \sum_{\nu=0}^B P_{\ell+1}^{B-\nu}\cdot\left(\sum_{k\notin \cS'}\sum_{d\in \mathbb{N}} h_{r-\ell, \cS', k}^{\nu, d}\right)\\
    &\le \sum_{\nu=0}^B P_{\ell+1}^{B-\nu}\cdot g_{r-\ell-1,\cS'}^\nu. \tag{by \Cref{ineq:abel}}
\end{align*}
Therefore 
\begin{align*}
    \sum_{\cS}\sum_{j=0}^B P_{\ell}^{B-j} \cdot g_{r-\ell,\cS}^j\le {8}\cdot\sum_{\cS'}\sum_{\nu=0}^B P_{\ell+1}^{B-\nu}\cdot g_{r-\ell-1,\cS'}^\nu
\end{align*}
giving us the induction step from $\ell$ to $\ell+1$, with $|\cS|=r$, $|\cS'|=r-\ell$, and $|\cS''|=r-\ell-1$:
\begin{align*}
    \sum_{\cS}\sum_{j=0}^B g_{r,\cS}^j
    &\le {8}^{\ell}\cdot \sum_{\cS'}\sum_{j=0}^B P^{B-j}_{\ell}\cdot g_{r-\ell, \cS'}^j\\
    &\le {8}^{\ell+1}\cdot \sum_{\cS''}\sum_{j=0}^B P^{B-j}_{\ell+1}\cdot g_{r-\ell-1, \cS'}^j.
\end{align*}

\paragraph*{Conclusion.} Therefore, by induction, \Cref{clm:qsdpt-G_P_g} holds for all $\ell\in[r]$ (for all fixed $r\in[\kappa]$).
\end{proof}

\subsection{Proof of Fact \ref{fct:transition-g-h}}\label{prf:transition-g-h}

\begin{proof}[Proof of \Cref{fct:transition-g-h}] 
In the following computation, we will use $\compressO=\cphso$.

According to the definition of $g_{r,\cS,k}^{j,\nu}$ in \Cref{def:grsijt} and $h_{r,\cS',k}^{j,\nu}$ in \Cref{def:hrsijt}, we can write them as the following, as we define $\cS'=\cS\backslash\{k\}$:
\begin{alignat*}{3}
    g_{r,\cS,k}^{j,\nu}&=\|Q_{j+\nu} Q^\nu\Lambda_{\cS}\cdot\  &&\cphso\ \cdot\   &&(\I-\Lambda_k)\tau_k \cdot U_{z_r}\ket{\psi_{\vz,\cS'}^{z_r-1}}\|^2,\\
    &=\|Q_{j+\nu} Q^\nu\Lambda_{\cS}\ \cdot\ &&\cphso\ \cdot\   &&\tQ_{j+\nu}^\nu \cdot
    \Lambda_{\cS'}(\I-\Lambda_{k})\cdot
    \tau_k\cdot U_{z_r}\ket{\psi_{\vz,\cS'}^{z_r-1}}\|^2,\\
    h_{r,\cS',k}^{j,\nu}&=\ &&\ &&\|\tQ_{j+\nu}^\nu \cdot  
    \Lambda_{\cS'}(\I-\Lambda_{k})\cdot
    \tau_k\cdot  U_{z_r}\ket{\psi_{\vz,\cS'}^{z_r-1}}\|^2.
\end{alignat*}
Here $Q^\nu$ is defined in \Cref{def:grsijt}; it is the projection that projects every superposition $\ket{k,x,u,z}\ket{D}$ to the states such that the size of the restricted database $D|_k$ on salt $k$ is $\nu$. Here $\tQ_{j+\nu}^\nu$ is defined in \Cref{def:hrsijt} as the projection onto the state that every superposition can contribute to falling in $Q_{j+\nu}Q^\nu$ after one oracle query $\cphso$. That is, $Q_{j+\nu}Q^\nu\cdot\cphso\cdot(\I-\tQ_{j+\nu}^\nu)\ket{\psi}=0$ for all possible $\ket{\psi}$, and $Q_{j+\nu}Q^\nu\cdot\cphso\cdot\tQ_{j+\nu}^\nu\neq 0$. 

We define that state with the norm of $h_{r,\cS',k}^{j,\nu}$ as $\ket{\phi}$, and it can be written as
\begin{align*}
    \ket{\phi}&=\tQ_{j+\nu}^\nu \cdot \Lambda_{\cS'}(\I-\Lambda_k)\cdot \tau_k\cdot U_{z_r}\ket{\psi_{\vz,\cS'}^{z_r-1}}\\
    &=\sum_{x,u,z,D}\alpha_{x,u,z,D}
    \ket{k,x,u,z}\otimes\ket{D} + 
    \sum_{x,u,z,D,y}\beta_{x,u,z,D,y} \ket{k,x,u,z}\otimes\ket{D\cup((k,x),y)}.
\end{align*}
Here the sum is over $D$ such that $(k,x)\notin D$. And we also define $\ket{\phi_{\alpha}}$ and $\ket{\phi_{\beta}}$ as
\begin{align*}
    \ket{\phi_{\alpha}}&=\sum_{x,u,z,D} \alpha_{x,u,z,D}\ket{k,x,u,z}\otimes\ket{D},\\
    \ket{\phi_{\beta}}&=\sum_{x,u,z,D,y} \beta_{x,u,z,D,y}\ket{k,x,u,z}\otimes\ket{D\cup((k,x),y)}\\
    &=\sum_{x,u,z,D}\ket{k,x,u,z}\otimes\left(
    \sum_y\beta_{x,u,z,D,y}\ket{D\cup((k,x),y)}
    \right).
\end{align*}
By definition, $\ket{\phi_{\alpha}}$ and $\ket{\phi_\beta}$ are orthogonal.

Now we analyze how the norm decreases for $\ket{\phi_{\alpha}}$ and $\ket{\phi_\beta}$ when $\cphso$ and the projector $Q_{j+\nu}Q^\nu\Lambda_{k}$ are applied. 

\paragraph*{Case of $\ket{\phi_{\alpha}}$.} For $\ket{\phi_{\alpha}}$ of which superpostion is over $D$ such that $(k,x)\notin D$,
\begin{align*}
    \cphso\ket{\phi_{\alpha}}&= \sum_{x,u,z,D}\alpha_{x,u,z,D}\cdot\left( \stddecomp_{(k,x)}\cdot\cphso'\cdot\stddecomp_{(k,x)}\ket{k,x,u,z}\otimes\ket{D}\right)\\
    &=\sum_{x,u,z,D}\alpha_{x,u,z,D}\ket{k,x,u,z} \otimes\left(\frac{1}{\sqrt{N}}
    \sum_{y}\omega_N^{uy}\ket{D\cup((k,x),y)}
    \right).
\end{align*}
By definition, every superposition of $\ket{\phi_{\alpha}}$ is in the image of $\tQ_{j+\nu}^\nu$, and we do not care the part that will lie in the image of $\I-Q_{j+\nu}Q^\nu\Lambda_{k}$ after one oracle query. Thus we can assume that the size of $D\cup((k,x),y)|_k$ is $\nu$, and the size of $D|_k$ is $\nu-1$, and
\begin{align*}
    \left\|Q_{j+\nu}Q^\nu \Lambda_{k}\cphso\ket{\phi_{\alpha}}\right\|
    &\le \|\Lambda_{k}\cphso\ket{\phi_{\alpha}}\|\\
    &=\left\|
    \sum_{x,u,z,D}\alpha_{x,u,z,D}\ket{k,x,u,z} \otimes\left(
    \frac{1}{\sqrt{N}}\sum_{\substack{y \text{ s.t. }\\ D\cup((k,x),y)\in P}}\omega_N^{u y}\ket{D\cup((k,x),y)}
    \right)
    \right\|\\
    &\le \sqrt{p_{\nu-1}}\cdot\left\|
    \sum_{x,u,z,D}\alpha_{x,u,z,D}\ket{k,x,u,z} \otimes\left(\frac{1}{\sqrt{N}}
    \sum_{y}\omega_N^{uy}\ket{D\cup((k,x),y)}
    \right)
    \right\|\\
    &= \sqrt{p_{\nu-1}}\cdot\|\cphso\ket{\phi_{\alpha}}\|.
\end{align*}
Note that for all $y\in[N]$, the number of $y$ such that $D\cup\{((k,x),y)\}|_k\in P$ while $D|_k\in \bar{P}$ is no more than $N\cdot p_{\nu-1}$, with the size of $D|_k$ no more than $\nu-1$. This is by the definition of the transition probability $p_{\nu-1}$. Thus the norm of every superposition of $\ket{k,x,u,z}$ in $\ket{\phi_{\alpha}}$ is reduced to at most $\sqrt{p_{\nu-1}}$ times its original norm, which means,
\begin{align}\label{ineq:tr_x_notinD}
    \left\|Q_{j+\nu}Q^\nu \Lambda_{k}\cphso\ket{\phi_{\alpha}}\right\|^2\le p_{\nu-1}\cdot \|\ket{\phi_{\alpha}}\|^2.
\end{align}

\paragraph*{Case of $\ket{\phi_{\beta}}$.} For $\ket{\phi_{\beta}}$ of which superposition is over $D$ such that $(k,x)\notin D$,
\begin{alignat*}{2}
    &\cphso\ket{\phi_{\beta}}\\
    =& \sum_{x,u,z,D,y}&&\beta_{x,u,z,D,y}\cdot\left( \stddecomp_{(k,x)}\cdot\cphso'\cdot\stddecomp_{(k,x)}\ket{k,x,u,z}\otimes\ket{D\cup((k,x),y)}\right)\\
    =& \sum_{x,u,z,D,y}
    &&\beta_{x,u,z,D,y}\ket{k,x,u,z}\otimes\\
    &   &&\left(\omega_N^{uy}\ket{D\cup((k,x),y)}+(1+\omega_N^{uy})\ket{D}-\frac{1}{N}\sum_{y'}(\omega_N^{uy}+\omega_N^{uy'})\ket{D\cup((k,x),y')}\right)
    .
\end{alignat*}
By definition, every superposition of $\ket{\phi_{\beta}}$ is in the image of $\tQ_{j+\nu}^\nu$, and we do not care the part that will lie in the image of $\I-Q_{j+\nu}Q^\nu\Lambda_{k}$ after one oracle query. Since $D|_k\in \bar{P}$, the size of $D\cup((k,x),y')|_k$ is $\nu$, and as $D\cup((k,x),y)|_k\in \bar{P}$, 
\begin{align}\label{eq:norm_beta}
    \|Q_{j+\nu} & Q^\nu\Lambda_{k}\cphso
    \ket{\phi_{\beta}}\| \le\|\Lambda_{k}\cphso\ket{\phi_{\beta}}\|\nonumber\\
    &=\left\|\sum_{x,u,z,D,y}\beta_{x,u,z,D,y} \ket{k,x,u,z}\otimes \left(\frac{1}{N}\sum_{\substack{y'\text{ s.t. }\\ D\cup((k,x),y')\in P}}(\omega_N^{uy}+\omega_N^{uy'})\ket{D\cup((k,x),y')}\right)\right\| \nonumber\\
    &=\left\|\sum_{x,u,z,D}\ket{k,x,u,z}\otimes \left(\sum_{\substack{y'\text{ s.t. }\\ D\cup((k,x),y')\in P}}\sum_y\frac{\beta_{x,u,z,D,y}}{N} \cdot(\omega_N^{uy}+\omega_N^{uy'})\ket{D\cup((k,x),y')}
    \right)
    \right\|. 
\end{align}
For every superposition of $\ket{k,x,u,z}$ in $\ket{\phi_{\beta}}$, the norm squared is $\sum_y |\beta_{x,u,z,D,y}|^2$. Note that for all $y\in[N]$, the number of $y$ such that $D\cup\{((k,x),y)\}|_k\in P$ while $D|_k\in \bar{P}$ is no more than $N\cdot p_{\nu-1}$, with the size of $D|_k$ no more than $\nu-1$. This is by the definition of the transition probability $p_{\nu-1}$. Therefore, for every superposition of $\ket{k,x,u,z}$ in \Cref{eq:norm_beta}, the norm squared is
\begin{align*}
    \sum_{\substack{y'\text{ s.t. }\\ D\cup((k,x),y')\in P}}\left|
    \sum_y\frac{\beta_{x,u,z,D,y}}{N} \cdot(\omega_N^{uy}+\omega_N^{uy'})
    \right|^2 &\le \frac{p_{\nu-1}}{N}\cdot 
    \max_{y'\in[N]}\left|
    \sum_y\beta_{x,u,z,D,y} \cdot(\omega_N^{uy}+\omega_N^{uy'})
    \right|^2\\
    &\le \frac{p_{\nu-1}}{N}\cdot
    \left(\sum_y |\beta_{x,u,z,D,y}|^2\right)\cdot \max_{y'\in[N]}\left(\sum_y|\omega_N^{uy}+\omega_N^{uy'}|^2\right)\\
    &\le \frac{p_{\nu-1}}{N}\cdot\left(\sum_y |\beta_{x,u,z,D,y}|^2\right)\cdot 4N\\
    &=4 p_{\nu-1}\cdot\left(\sum_y |\beta_{x,u,z,D,y}|^2\right).
\end{align*}
The second line is by Cauchy–Schwarz inequality $|\sum_i a_ib_i|^2\le(\sum_i|a_i|^2)(\sum_i|b_i|^2)$.
Therefore, we can conclude from \Cref{eq:norm_beta} that
\begin{align}\label{ineq:tr_x_inD}
    \|Q_{j+\nu} & Q^\nu\Lambda_{k}\cphso
    \ket{\phi_{\beta}}\|^2\le 4p_{\nu-1}\cdot\|\ket{\phi_{\beta}}\|^2.
\end{align}

\paragraph*{Conclusion.} Note that for every $\ket{\phi_{\alpha}}$ and $\ket{\phi_{\beta}}$, $\forall k\in[\kappa]$, $Q_{j+\nu}Q^\nu\Lambda_{k}\cphso$ will not change their first register $\ket{k}$. Therefore, by \Cref{ineq:tr_x_notinD}, \Cref{ineq:tr_x_inD}, and the orthogonality on the first register,
\begin{align*}
    g_{r,\cS,k}^{j,\nu}&=\|Q_{j+\nu}Q^\nu\Lambda_{\cS} \cphso \ket{\phi}\|^2\\
    &\le \left\|Q_{j+\nu}Q^\nu\Lambda_{k}\cphso \ket{\phi_{\alpha}} +Q_{j+\nu}Q^\nu\Lambda_{k}\cphso\ket{\phi_{\beta}}\right\|^2\\
    &\le 2\cdot \left(
    \left\|Q_{j+\nu}Q^\nu\Lambda_{k}\cphso \ket{\phi_{\alpha}}\right\|^2 +\left\|Q_{j+\nu}Q^\nu\Lambda_{k}\cphso \ket{\phi_{\beta}}\right\|^2
    \right)
    \tag{since $\|a+b\|^2\le 2(\|a\|^2+\|b\|^2)$}
    \\
    &\le 8\cdot p_{\nu-1}\cdot \left(\|\ket{\phi_{\alpha}}\|^2+\|\ket{\phi_{\beta}}\|^2\right)\tag{by \Cref{ineq:tr_x_notinD} and \Cref{ineq:tr_x_inD}}
    \\
    &=8\cdot p_{\nu-1}\cdot\|\ket{\phi}\|^2
    =8\cdot p_{\nu-1}\cdot h_{r,\cS',k}^{j,\nu}.\tag{by orthogonality}
\end{align*}
Thus \Cref{fct:transition-g-h} holds.
\end{proof}

Actually, within the analysis to prove \Cref{fct:transition-g-h}, we can have the following corollary that shows the transition pattern.

\begin{fact*}[\Cref{fct:cO-transition} Restated]
Let $P$ be a monotone property with transition probability $p_t$. For any joint state $\ket{\psi}$ of the algorithm's registers and the oracle registers, define $\Gamma_{k}^{\le t}$ as the projection that projects every superposition $\ket{i,x,u,z}\ket{D}$ to the states such that $i=k$ and the size of the database $D|_k$ is less than $t$. Define $\Gamma_{k}^P$ as the projection that projects to $D|_k\in P$. Then we can have the following result for the decrease of norm during the transition from $D|_k\in \bar{P}$ to $D|_k\in P$:
\begin{align*}
    \left\|\Gamma_k^P\cdot \cphso\cdot (\I-\Gamma_k^P)\Gamma_k^{\le t}\ket{\psi}\right\|
    \le \sqrt{8\cdot p_t}\left\|(\I-\Gamma_k^P)\Gamma_k^{\le t}\ket{\psi}\right\|.
\end{align*}
\end{fact*}
\begin{proof}[Proof Outline of \Cref{fct:cO-transition}]
We similarly define $\ket{\psi_{\alpha,k}}$ to be the part where every superposition $\ket{k,x,u,z}\ket{D}$ satisfies $D(k,x)=\bot$, and define $\ket{\psi_{\beta,k}}$ to be the part $D(k,x)\neq\bot$. Note that every state in the image of $\Gamma_k^{\le t}$ satisfies that the size of $D|_k$ is no more than $t$. 

Now we can adapt the proof of \Cref{ineq:tr_x_notinD} for every $(\I-\Gamma_k^P)\Gamma_k^{\le t}\ket{\psi_{\alpha,k}}$, and obtain
\begin{align}\label{ineq:tr_x_notinD_general}
    \|\Gamma_k^P\cdot \cphso\cdot (\I-\Gamma_k^P)\Gamma_k^{\le t}\ket{\psi_{\alpha,k}}\|^2\le p_{t}\cdot
    \|(\I-\Gamma_k^P)\Gamma_k^{\le t}\ket{\psi_{\alpha,k}}\|^2.
\end{align}
Similarly, we can adapt the proof of \Cref{ineq:tr_x_inD} for every $(\I-\Gamma_k^P)\Gamma_k^{\le t}\ket{\psi_{\beta,k}}$ and obtain
\begin{align}\label{ineq:tr_x_inD_general}
    \|\Gamma_k^P\cdot \cphso\cdot (\I-\Gamma_k^P)\Gamma_k^{\le t}\ket{\psi_{\beta,k}}\|^2\le 4p_{t}\cdot
    \|(\I-\Gamma_k^P)\Gamma_k^{\le t}\ket{\psi_{\beta,k}}\|^2.
\end{align}
Since $\Gamma_k^{\le t}(\ket{\psi_{\alpha,k}}+\ket{\psi_{\beta,k}})=\Gamma_k^{\le t}\ket{\psi}$, we can obtain
\begin{align*}
    \|\Gamma_k^P\cdot \cphso\cdot (\I-\Gamma_k^P)\Gamma_k^{\le t}\ket{\psi}\|^2
    &=\|\Gamma_k^P\cdot \cphso\cdot (\I-\Gamma_k^P)\Gamma_k^{\le t}(\ket{\psi_{\alpha,k}}+\ket{\psi_{\beta,k}})\|^2
    \\
    &\le 2\cdot\left(\text{LHS of \Cref{ineq:tr_x_notinD_general}}+\text{LHS of \Cref{ineq:tr_x_inD_general}}\right)
    \tag{since $\|a+b\|^2\le 2\cdot(\|a\|^2+\|b\|^2)$}\\
    &\le 8p_t\cdot\left(\|(\I-\Gamma_k^P)\Gamma_k^{\le t}\ket{\psi_{\alpha,k}}\|^2+\|(\I-\Gamma_k^P)\Gamma_k^{\le t}\ket{\psi_{\beta,k}}\|^2\right)
    \tag{by \Cref{ineq:tr_x_notinD_general} and \Cref{ineq:tr_x_inD_general}}\\
    &=8p_t\cdot\|(\I-\Gamma_k^P)\Gamma_k^{\le t}\ket{\psi}\|^2\tag{by orthogonality}
\end{align*}
as desired.
\end{proof}

\end{document}